\documentclass[a4paper,english,cleveref, autoref, thm-restate, numberwithinsect]{lipics-v2021}

\usepackage[pdflatex,recompilepics=false]{gastex}
\ifpdf
  \usepackage{todonotes}
  \newcounter{todocounter}

\else
  \AtBeginDocument{%
    \RemoveFromHook{shipout/firstpage}[hyperref]%
    \RemoveFromHook{shipout/before}[hyperref]%
  }
\fi

\newcommand{\auto}{\mathcal{A}}
\newcommand{\basic}{\mathcal{B}}
\newcommand{\C}{\mathcal{C}}

\newcommand{\trans}{\mathcal{T}}
\newcommand{\rtrans}{\reverse{\trans}}
\newcommand{\twtrans}{\mathcal{T}}
\newcommand\dom[1]{\mathsf{dom}(#1)}
\newcommand\sem[1]{[\![ #1 ]\!]}
\newcommand{\leftend}{{\mathop\vdash}}
\newcommand{\rightend}{{\mathop\dashv}}
\newcommand{\lrend}{{\#}}
\newcommand{\splend}{{\vdash\!\!\dashv}}
\newcommand{\Sigmae}{\Sigma_{\#}}
\newcommand{\pos}[1]{\mathsf{pos}(#1)}
\newcommand{\epos}[1]{\mathsf{epos}(#1)}

\newcommand{\head}{\mathsf{h}}
\newcommand{\pebble}{\mathsf{p}}
\newcommand{\pebbles}{\mathsf{peb}}

\newcommand{\drop}[1]{\mathsf{drop}_{#1}}
\newcommand{\lift}[1]{\mathsf{lift}_{#1}}
\newcommand{\push}{\mathsf{drop}}
\newcommand{\pop}{\mathsf{lift}}
\newcommand{\op}{\mathsf{op}}
\newcommand{\nop}{\mathsf{nop}}
\newcommand{\rev}{\mathsf{rev}}
\newcommand{\false}{\mathsf{false}}

\newcommand{\true}{\mathsf{true}}

\newcommand{\test}[1]{\overline{#1}}
\newcommand{\reverse}[1]{#1^{r}}

\newcommand{\xla}{\xleftarrow}
\newcommand{\xra}{\xrightarrow}
\newcommand{\enc}[1]{\mathsf{encode}(#1)}

\newcommand{\qright}[1]{#1_r}
\newcommand{\qstay}[1]{#1_0}
\newcommand{\qlefto}[1]{#1_1}
\newcommand{\qleftt}[1]{#1_2}
\newcommand{\qminus}[1]{#1_-}
\newcommand{\qplus}[1]{#1_+}

\title{Reversible Pebble Transducers}

\author{Luc Dartois}{Université Paris Est Creteil, LACL, F-94010 Créteil, France}{luc.dartois@u-pec.fr}{https://orcid.org/0000-0001-9974-1922}{}

\author{Paul Gastin}{Université Paris-Saclay, ENS Paris-Saclay, CNRS, LMF, 91190, 
  Gif-sur-Yvette, France \and CNRS, ReLaX, IRL 2000, Siruseri, India}
  {paul.gastin@lmf.cnrs.fr}{https://orcid.org/0000-0002-1313-7722}{}

\author{Loïc {Germerie Guizouarn}}{Univ Rennes, CNRS, Inria, IRISA - UMR 6074, F-35000 Rennes, France}{loic.germerie-guizouarn@univ-rennes.fr}{https://orcid.org/0000-0002-3843-5427}{}

\author{Shankaranarayanan Krishna}{Indian Institute of Technology Bombay, Mumbai, India}{krishnas@cse.iitb.ac.in}{https://orcid.org/0000-0003-0925-398X}{}

\authorrunning{L. Dartois, P. Gastin, L. Germerie Guizouarn and S. Krishna}

\nolinenumbers

\Copyright{Luc Dartois, Paul Gastin, Loïc Germerie Guizouarn and Shankaranarayanan Krishna}

\hideLIPIcs  

\keywords{Transducers, Polyregular functions, Reversibility, Composition, Uniformization}

\ccsdesc[500]{Theory of computation~Transducers}

\EventEditors{Patricia Bouyer and Jaco van de Pol}
\EventNoEds{2}
\EventLongTitle{36th International Conference on Concurrency Theory (CONCUR 2025)}
\EventShortTitle{CONCUR 2025}
\EventAcronym{CONCUR}
\EventYear{2025}
\EventDate{August 26--29, 2025}
\EventLocation{Aarhus, Denmark}
\EventLogo{}
\SeriesVolume{348}
\ArticleNo{24}

\begin{document}
\maketitle
% 
%%%%%%%%%%%%%%%%%%%%%%%%%%%%%%%%%%%%%%%%%%%%%%%%%%%%%%%%%%%%%%%%%%%%
%%% FIGURES
%%%%%%%%%%%%%%%%%%%%%%%%%%%%%%%%%%%%%%%%%%%%%%%%%%%%%%%%%%%%%%%%%%%%
\begin{gpicture}[name=zeroP,ignore]
	\gasset{Nw=7,Nh=7,loopdiam=4,loopheight=5,loopangle=90,curvedepth=0}
	\unitlength=1.1mm
	\gasset{linecolor=black}
	
	\node(0)(0,0){}
	\node(1)(20,0){$0$}

	\gasset{ExtNL=y,NLangle=90,NLdist=1}
	\nodelabel(0){\( q \)}
	\nodelabel(1){\( p \)}
	
	\drawedge(0,1){\( a \)}
\end{gpicture}
\begin{gpicture}[name=zeroT,ignore]
	\gasset{Nw=7,Nh=7,loopdiam=4,loopheight=5,loopangle=90,curvedepth=0}
	\unitlength=1.1mm
	\gasset{linecolor=black}
	
	\node(0)(0,0){$+$}
	\node(1)(20,0){$-$}
	\node(2)(40,0){$+$}
		
	\gasset{ExtNL=y,NLangle=90,NLdist=1}
	\nodelabel(0){\( \qright{q} \)}
	\nodelabel(1){\( \qstay{p} \)}
	\nodelabel(2){\( \qright{p} \)}
	
	\drawedge(0,1){\( a \)}
	\drawedge(1,2){\( \Sigma_\leftend \)}
\end{gpicture}
\begin{gpicture}[name=minusP,ignore]
	\gasset{Nw=7,Nh=7,loopdiam=4,loopheight=5,loopangle=90,curvedepth=0}
	\unitlength=1.1mm
	\gasset{linecolor=black}
	
	\node(0)(0,0){}
	\node(1)(20,0){$-1$}
	
	\gasset{linecolor=red}
	\node(5)(0,-17){}
	\gasset{linecolor=black}
	
	\gasset{ExtNL=y,NLangle=90,NLdist=1}
	\nodelabel(0){\( q \)}
	\nodelabel(1){\( p \)}
	\nodelabel(5){\( q' \)}
	
	\drawedge(0,1){\( a \)}
	
	\gasset{linecolor=red}
	\drawedge(5,1){\( a' \)}
\end{gpicture}
\begin{gpicture}[name=minusT,ignore]
	\gasset{Nw=7,Nh=7,loopdiam=4,loopheight=5,loopangle=90,curvedepth=0}
	\unitlength=1.1mm
	\gasset{linecolor=black}
	
	\node(0)(0,0){$+$}
	\node(1)(20,0){$-$}
	\node(2)(40,0){$-$}
	\node(3)(60, 0){$+$}
	\node(4)(30,-17){$+$}
	
	\gasset{linecolor=red}
	\node(5)(0,-17){$+$}
	\gasset{linecolor=black}
	
	\gasset{ExtNL=y,NLangle=90,NLdist=1}
	\nodelabel(0){\( \qright{q} \)}
	\nodelabel(1){\( \qlefto{p} \)}
	\nodelabel(2){\( \qleftt{p} \)}
	\nodelabel(3){\( \qright{p} \)}
	\nodelabel[NLangle=0](4){\( \qplus{p} \)}
	\nodelabel(5){\( \qright{q'} \)}
	
	\drawedge(0,1){\( a \)}
	\drawedge(1,2){\( \Sigma \)}	
	\drawedge(2,3){\( \Sigma_\leftend \)}
	
	\drawloop[loopangle=180](4){\( \Sigma \)}
	
	\drawedge[ELside=r](1,4){\( \leftend \)}
	\drawedge[ELside=r](4,2){\( \rightend \)}
	
	\gasset{linecolor=red}
	\drawedge(5,1){\( a' \)}
\end{gpicture}
\begin{gpicture}[name=squaring-RPT1,ignore]
  \gasset{Nw=7,Nh=7,loopdiam=4,loopheight=5,loopangle=90,curvedepth=0}
  \unitlength=1.1mm
  \gasset{linecolor=black}
  \node[Nmarks=i,ilength=5](0)(0,0){$0$}
  \node(1)(25,0){$+1$}
  \node[Nmarks=f,flength=5](2)(50,0){$0$}
  \drawedge[ELside=l](0,1){$\lrend$}
  \drawedge[ELside=l](1,2){$\lrend$}
  \gasset{linecolor=blue}
  \node(3)(0,-25){$\textcolor{blue}{+1}$}
  \node(4)(25,-25){$\textcolor{blue}{+1}$}
  \node(5)(50,-25){$\textcolor{blue}{+1}$}
  \gasset{ExtNL=y,NLangle=90,NLdist=1}
  \nodelabel(0){$q_{0}$}
  \nodelabel(1){$q_{1}$}
  \nodelabel(2){$q_{2}$}
  \nodelabel[NLangle=-90](3){$q_{3}$}
  \nodelabel[NLangle=-90](4){$q_{4}$}
  \nodelabel[NLangle=-90](5){$q_{5}$}
  \drawedge[linecolor=black,curvedepth=-4,ELside=r,ELdist=0,ELpos=75](1,3){$a\in\Sigma, \textcolor{red}{\drop{1}}$}
  \drawloop[loopangle=180](3){$a\in\Sigma, \textcolor{red}{\neg\pebble_{1}}$}
  \drawedge(3,4){$\lrend$}
  \drawloop[loopangle=90,ELdist=0](4){$\begin{array}{l@{\;\mid\;}l}
    a\in\Sigma, \textcolor{red}{\neg\pebble_{1}} & a \\
    a\in\Sigma, \textcolor{red}{\pebble_{1}} & \underline{a}
  \end{array}$}
  \drawedge(4,5){$\lrend$}
  \drawloop[loopangle=0](5){$a\in\Sigma, \textcolor{red}{\neg\pebble_{1}}$}
  \drawedge[linecolor=black,curvedepth=-4,ELside=r,ELdist=0,ELpos=25](5,1){$a\in\Sigma, \textcolor{red}{\lift{1}}$}
\end{gpicture}
\begin{gpicture}[name=all-prefixes-RPT1,ignore]
	\gasset{Nw=7,Nh=7,loopdiam=4,loopheight=5,loopangle=90,curvedepth=0}
	\unitlength=1.1mm
	\node[Nmarks=i,ilength=5](0)(0,0){$0$}
	\node(1)(25,0){$+1$}
	\node[Nmarks=f,flength=5](2)(50,0){$0$}
	\drawedge[ELside=l](0,1){$\lrend$}
	\drawedge[ELside=l](1,2){$\lrend$}
	\gasset{linecolor=blue}
	\node(3)(10,-20){$\textcolor{blue}{-1}$}
	\node(5)(40,-20){$\textcolor{blue}{+1}$}
	\gasset{ExtNL=y,NLangle=90,NLdist=1}
	\nodelabel(0){$q_{0}$}
	\nodelabel(1){$q_{1}$}
	\nodelabel(2){$q_{2}$}
	\gasset{NLangle=-90}
	\nodelabel(3){$q_{3}$}
	\nodelabel(5){$q_{4}$}
  
	\drawloop[loopangle=180](3){$a\in\Sigma,\textcolor{red}{\neg\pebble_{1}} \mid a$}
	\drawedge(3,5){$\lrend \mid ! $}
	\drawloop[loopangle=0](5){$a\in\Sigma,\textcolor{red}{\neg\pebble_{1}}$}

	\gasset{linecolor=black,curvedepth=0,ELside=r,ELdist=1}
	\drawedge[ELpos=75](1,3){$a\in\Sigma,\textcolor{red}{\drop{1}} \mid a$}
	\drawedge[ELpos=30](5,1){$a\in\Sigma,\textcolor{red}{\lift{1}}$}
\end{gpicture}
\begin{gpicture}[name=composition-drop,ignore]
	\gasset{Nw=7,Nh=7,loopdiam=4,loopheight=7,loopangle=90,curvedepth=0}
	\unitlength=1mm
	\node(0)(0,0){$0$}
	\node(1)(24,0){$+1$}
	\node(2)(51,0){$+1$}
	\node[Nframe=n](dots)(62.5,0){$\cdots$}
	\node(4)(74,0){$+1$}
	\node(5)(103,0){$+1$}
	\node(6)(127,0){$0$}
  \gasset{ExtNL=y,NLangle=-90,NLdist=1}
  \nodelabel(0){$q''$}
  \nodelabel(1){$(q'',z,1)$}
  \nodelabel(2){$(q'',z,2)$}
  \nodelabel(4){$(q'',z,z-1)$}
  \nodelabel(5){$(q'',z,z)$}
  \nodelabel(6){$s''$}
	\drawedge[ELdist=0.5](0,1){$\begin{array}{l} 
    a,\xi \\ 
    \drop{d_{k}+z}
  \end{array}$}
	\drawloop(1){$\begin{array}{l} 
    \textcolor{blue}{\neg\pebble_{d_{k}+z}} \\ 
    {}\wedge \neg\pebble_{d_{k}+1} 
  \end{array}$}
	\drawedge[ELdist=0.5](1,2){$\begin{array}{l} 
    \pebble_{d_{k}+1} \\ 
    \drop{d_{k}+z+1}
  \end{array}$}
	\drawloop(2){$\begin{array}{l} 
    \textcolor{blue}{\neg\pebble_{d_{k}+z+1}} \\ 
    {}\wedge \neg\pebble_{d_{k}+2} 
  \end{array}$}
	\drawloop(4){$\begin{array}{c} 
    \textcolor{blue}{\neg\pebble_{d_{k}+2z-2}} \\ 
    {}\wedge \neg\pebble_{d_{k}+z-1} 
  \end{array}$}
	\drawedge[ELdist=0.5](4,5){$\begin{array}{l} 
    \pebble_{d_{k}+z-1} \\ 
    \drop{d_{k}+2z-1}
  \end{array}$}
	\drawloop(5){$\begin{array}{c} 
    \textcolor{blue}{\neg\pebble_{d_{k}+2z-1}} \\ 
    {}\wedge \neg\pebble_{d_{k}+z} 
  \end{array}$}
	\drawedge(5,6){$a,\xi',\nop$}
\end{gpicture}
\begin{gpicture}[name=composition-lift,ignore]
	\gasset{Nw=7,Nh=7,loopdiam=4,loopheight=7,loopangle=90,curvedepth=0}
	\unitlength=1mm
	\node(0)(2,0){$0$}
	\node(1)(24,0){$+1$}
	\node(2)(50,0){$+1$}
	\node[Nframe=n](3)(60,0){$\cdots$}
	\node(4)(70,0){$+1$}
	\node(5)(93,0){$+1$}
	\node(6)(110,0){$0$}
	\node(7)(130,0){$0$}
  \gasset{ExtNL=y,NLangle=-90,NLdist=1}
  \nodelabel(0){$q''$}
  \nodelabel(1){$(q'',-1)$}
  \nodelabel(2){$(q'',-2)$}
  \nodelabel(4){$(q'',-y_{k}+1)$}
  \nodelabel(5){$(q'',-y_{k})$}
  \nodelabel(6){$(q'',0)$}
  \nodelabel(7){$s''$}
	\drawedge(0,1){$a,\xi,\nop$}
	\drawloop(1){$\begin{array}{l} 
    \textcolor{blue}{\neg\pebble_{d_{k}} \wedge{}} \\ 
    \neg\pebble_{d_{k}+y_{k}-1} 
  \end{array}$}
	\drawedge[ELdist=0.5](1,2){$\begin{array}{l} 
    \textcolor{blue}{\pebble_{d_{k}-1}} \\ 
    \lift{d_{k}+y_{k}-1} 
  \end{array}$}
	\drawloop(2){$\begin{array}{l} 
    \textcolor{blue}{\neg\pebble_{d_{k}-1} \wedge{}} \\ 
    \neg\pebble_{d_{k}+y_{k}-2} 
  \end{array}$}
	\drawloop(4){$\begin{array}{c} 
    \textcolor{blue}{\neg\pebble_{d_{k}-y_{k}+2}} \\ 
    {}\wedge \neg\pebble_{d_{k}+1} 
  \end{array}$}
	\drawedge[ELdist=0.5](4,5){$\begin{array}{l} 
    \textcolor{blue}{\pebble_{1+d_{k-1}}} \\ 
    \lift{d_{k}+1}
  \end{array}$}
	\drawloop(5){$\begin{array}{c} 
    \textcolor{blue}{\neg\pebble_{d_{k}-y_{k}+1}} \\ 
    {}\wedge\neg\pebble_{d_{k}} 
  \end{array}$}
	\drawedge(5,6){$\lift{d_{k}}$}
	\drawedge(6,7){\small{$a,\xi',\nop$}}
\end{gpicture}
\begin{gpicture}[name=iterRev,ignore]
  \gasset{Nw=7,Nh=7,loopdiam=4,loopheight=5,loopangle=90,curvedepth=0}
  \unitlength=1.1mm
  \gasset{linecolor=black}

  \node[Nmarks=i,ilength=5](0)(10,0){$0$}
  \node(1)(25,0){$+1$}
  \node(2)(50,0){$+1$}
  \node[Nmarks=f,flength=5](3)(65,0){$0$}
  \node(4)(37,-15){$-1$}
  
  \gasset{ExtNL=y,NLangle=-90,NLdist=1}
  \nodelabel(0){$q_0'$}
  \nodelabel[NLangle=-120](1){$q_1'$}
  \nodelabel[NLangle=-40](4){$q_2'$}
  \nodelabel[NLangle=-60](2){$q_3'$}
  \nodelabel[NLangle=60](3){$q_f'$}
 
 \drawloop[loopangle=90](1){$a\in\Sigma\mid\varepsilon$}
 \drawloop[loopangle=90](2){$a\in\Sigma\mid\varepsilon$}
 \drawloop[loopangle=90](4){$a\in\Sigma\mid a$}
 \drawedge(0,1){$\#\mid \varepsilon$}
 \drawedge(2,3){$\#\mid \varepsilon$}
 \drawedge[ELside=r,curvedepth=-6](1,4){$\#,!\mid \varepsilon$}
 \drawedge[ELside=r,curvedepth=-6](4,2){$\#,!\mid \varepsilon$}
 \drawedge[ELside=r](2,1){${!}\mid{!}$}
 \drawedge(2,3){$\#\mid\varepsilon$}
\end{gpicture}
\begin{gpicture}[name=allconfigs-1,ignore]
  \gasset{Nw=7,Nh=7,loopdiam=4,loopheight=5,loopangle=90,curvedepth=0}
  \unitlength=1.1mm
  \node[Nmarks=i,ilength=5](0)(0,0){$0$}
  \node(1)(25,0){$+1$}
  \node[Nmarks=f,flength=5](2)(50,0){$0$}
  \node(3)(0,-25){$+1$}
  \node(4)(25,-25){$+1$}
  \node(5)(50,-25){$+1$}
  \gasset{ExtNL=y,NLangle=90,NLdist=1}
  \nodelabel(0){$q_{0}$}
  \nodelabel(1){$q_{1}$}
  \nodelabel(2){$q_{2}$}
  \nodelabel[NLangle=180](3){$q_{3}$}
  \nodelabel[NLangle=-145](4){$q_{4}$}
  \nodelabel[NLangle=0](5){$q_{5}$}
  \drawedge[ELside=l](1,2){$\lrend$}
  \drawedge[curvedepth=0,ELside=r,ELdist=-1](0,3){$\begin{array}{c} 
    \lrend, \\ \textcolor{red}{\drop{1}} \end{array}$}
  \drawedge[curvedepth=-2,ELside=l,ELdist=-1.5,ELpos=40](1,3){$\begin{array}{r} 
    a\in\Sigma, \\ \textcolor{red}{\drop{1}} \end{array}$}
  \drawloop[loopangle=-90](3){$a\in\Sigma, \textcolor{red}{\neg\pebble_{1}}$}
  \drawedge(3,4){$\lrend \mid (\#,b)$}
  \drawloop[loopangle=270,ELdist=0.5](4){$a\in\Sigma \mid (a,b)$}
  \drawedge(4,5){$\lrend$}
  \drawloop[loopangle=-90](5){$a\in\Sigma, \textcolor{red}{\neg\pebble_{1}}$}
  \drawedge[curvedepth=-2,ELside=r,ELdist=-4,ELpos=30](5,1){$\begin{array}{c} 
    a\in\Sigma\cup\{\lrend\}, \\ \textcolor{red}{\lift{1}} \end{array}$}
\end{gpicture}
\begin{gpicture}[name=allconfigs-k,ignore]
  \gasset{Nw=7,Nh=7,loopdiam=4,loopheight=5,loopangle=90,curvedepth=0}
  \unitlength=1.1mm
  \node[Nmarks=i,ilength=5](0)(0,0){$0$}
  \node(1)(25,0){$+1$}
  \node[Nmarks=f,flength=5](2)(50,0){$0$}
  \node(3)(0,-25){$+1$}
  \node[Nw=15,Nmr=10](4)(25,-25){$\C^{+1}_{k-1}$}
  \node(5)(50,-25){$+1$}
  \gasset{ExtNL=y,NLangle=90,NLdist=1}
  \nodelabel(0){$q_{0}$}
  \nodelabel(1){$q_{1}$}
  \nodelabel(2){$q_{2}$}
  \nodelabel[NLangle=180](3){$q_{3}$}
  \nodelabel[NLangle=0](5){$q_{5}$}
  \drawedge[ELside=l](1,2){$\lrend$}
  \drawedge[curvedepth=0,ELside=r,ELdist=-1](0,3){$\begin{array}{c} 
    \lrend, \\ \textcolor{red}{\drop{1}} \end{array}$}
  \drawedge[curvedepth=-2,ELside=l,ELdist=-1.5,ELpos=40](1,3){$\begin{array}{r} 
    a\in\Sigma, \\ \textcolor{red}{\drop{1}} \end{array}$}
  \drawloop[loopangle=-90](3){$a\in\Sigma, \textcolor{red}{\neg\pebble_{1}}$}
  \drawedge[ELpos=40](3,4){$\lrend$}
  \drawedge[ELpos=60](4,5){$\lrend$}
  \drawloop[loopangle=-90](5){$a\in\Sigma, \textcolor{red}{\neg\pebble_{1}}$}
  \drawedge[curvedepth=-2,ELside=r,ELdist=-4,ELpos=30](5,1){$\begin{array}{c} 
    a\in\Sigma\cup\{\#\}, \\ \textcolor{red}{\lift{1}} \end{array}$}
\end{gpicture}
\begin{abstract}
  Deterministic two-way transducers with pebbles (aka pebble transducers) capture the
  class of polyregular functions, which extend the string-to-string regular functions 
  allowing polynomial growth instead of linear growth.
  One of the most fundamental operations on functions is composition, and  
 (poly)regular functions can be realized as a composition of several simpler functions. 
  In general, composition of deterministic two-way transducers incur a doubly exponential
  blow-up in the size of the inputs.  A major improvement in this direction comes from the
  fundamental result of Dartois et al.~\cite{DFJL17} showing a polynomial construction for
  the composition of \emph{reversible} two-way transducers.  
  
  A precise complexity analysis for existing composition techniques of pebble transducers
  is missing.  But they rely on the classic composition of two-way transducers and 
  inherit the double exponential complexity.
  To overcome this problem, we introduce \emph{reversible} pebble transducers.  Our
  main results are efficient uniformization techniques for non-deterministic pebble
  transducers to reversible ones and efficient composition for reversible pebble
  transducers.
\end{abstract}
% 
%%%%%%%%%%%%%%%%%%%%%%
\section{Introduction}
The theory of string transformations has garnered a lot of attention in recent years. The simplest kind of such transformations are \emph{sequential functions} which  are captured 
by deterministic finite automata whose transitions are labelled by output words. An example of a sequential function is $f(w)=w'$ where $w,w'$ are words over an alphabet $\{a,b\}$, and 
$w'$ is obtained from $w$ by replacing in $w$, each $a$ with $bb$, and  each $b$ with $a$. 
\emph{Rational functions} strictly extend sequential functions by allowing the underlying
automaton to be \emph{unambiguous} instead of being deterministic.  For instance,
$f(ua)=au$ where $a\in \Sigma, u \in \Sigma^+$ which moves the last symbol of
the input to the first position is a rational function, and not
realizable by a sequential one. 
 
While both sequential and rational functions are captured by `one-way' transducers,  
\emph{regular functions} are realized by two-way transducers where the input is scanned in both directions. Regular functions strictly extend rational functions; for instance 
the regular function $f(w)=w^R$ where $w^R$ is the reverse of $w$ is not realizable as a rational function. Regular functions are also equivalent to Courcelle's word-to-word MSO transductions \cite{Courcelle94, EH01},  
streaming string transducers \cite{AC10} as well as combinator expressions \cite{AlurFreilichRaghothaman14,  BR-DLT18, BDK-lics18, DGK-lics18}.  

Sequential, rational as well as regular functions are all transformations of \emph{linear growth}, where the sizes of the outputs are linear in the sizes of the inputs. 
More recently,  \cite{bojanpebble} revisits  functions of \emph{polynomial growth}, that is, those whose output sizes are polynomial in the sizes of the input. 
These functions, aptly called \emph{polyregular functions}, date back to \cite{DBLP:journals/acta/Engelfriet15, DBLP:conf/mfcs/EngelfrietM02}.
A logical characterization for polyregular functions was given in \cite{DBLP:conf/icalp/BojanczykKL19}, namely, word-to-word MSO 
 interpretations. 
 MSO interpretations are equivalent to  polyregular functions  which are word-to-word functions recognised by deterministic, two-way pebble transducers \cite{DBLP:journals/corr/abs-1810-08760}. 

Prior work \cite{DBLP:journals/acta/Engelfriet15, DBLP:journals/corr/abs-1810-08760}  on pebble transducers show that they are closed under composition, and are equivalent
to polyregular functions.  \cite{DBLP:journals/acta/Engelfriet15} 
considers `basic' pebble transducers where one can check the presence or absence of a certain pebble at the head position, along with the actions of dropping and lifting pebbles following a stack discipline. The main result of \cite{DBLP:journals/acta/Engelfriet15}  is the closure under composition 
of pebble transducers, as well as a \emph{uniformization} of (non-determinsitic) pebble 
to deterministic ones with the same number of pebbles. 
The pebble transducers of \cite{DBLP:journals/corr/abs-1810-08760} allow \emph{comparison
tests}, namely, one can compare the positions on which pebbles are dropped, or the 
position of the head and of a pebble.
\cite{DBLP:journals/corr/abs-1810-08760} gives high level ideas
for the closure under composition for pebble transducers with comparison tests.

Our focus in this paper is on the complexity aspects of the composition and uniformization
of pebble transducers.
A practical setting where composition is useful is while synthesizing pebble transducers 
from polyregular functions \cite{DBLP:conf/icalp/BojanczykKL19}.  
  Polyregular functions \cite{bojanpebble} 
are defined as the smallest class of functions closed under composition that contains 
sequential functions, the squaring function as well as the iterated reverse. 
In general, the composition of (pebble-less) two-way transducers has at least a doubly
exponential blowup in the size of the input transducers \cite{CJ77}.  An important class
of two-way transducers for which composition is polynomial are \emph{reversible}
transducers \cite{DFJL17}.
Reversible transducers are those which are both deterministic and reverse-deterministic;
moreover they are expressively equivalent to two-way transducers \cite{DFJL17}.  This
makes reversibility a very attractive property for a transducer.  However, the notion of
reversible pebble transducers has not been considered yet.  Therefore it is unknown if
they have the same expressiveness as deterministic pebble transducers (polyregular
functions), and whether they are also amenable to an efficient composition like their
pebble-less counterparts.  Our paper fills this gap.
\paragraph*{Our Contributions.}  
\begin{enumerate}
  \item {\bf{Reversible Pebble Transducers}}.  We define reversible pebble transducers 
  with equality tests allowing to check whether two pebbles are dropped on the same 
  position, in addition to the basic tests allowing to check if a pebble is dropped on the 
  current head position.
  
  \item {\bf{Equivalent notions for pebble transducers}}. 
  We show that pebble transducers with basic tests only, as well as those extended with
  equality tests are equivalent.
  More precisely, given a $k$-pebble transducer with equality tests having $n$ states, we
  show that we can construct an equivalent $k$-pebble transducer with basic tests having
  $\mathcal{O}(nk^2)$ states.  The construction preserves both determinism and reverse-determinism.
  
  \item {\bf{Composition of Reversible Pebble Transducers}}.  We show that reversible
  pebble transducers are closed under composition.  Given points 1,2 above, we allow
  equality tests while studying composition.  To be precise, given reversible pebble
  transducers $\trans$ and $\trans'$ with states $Q$, $Q'$, having $n,m\geq 0$ pebbles, we
  construct a reversible pebble transducer $\trans''$ with $nm+n+m$
  pebbles and (i) $2\cdot|Q|\cdot|Q'|$ states when $m=0$, and (ii) 
  $\mathcal{O}(|Q|^{m+2}\cdot|Q'|\cdot(n+1)^{m+3})$ states when $m>0$. 
  
  To compare with the existing composition approaches  \cite{DBLP:journals/acta/Engelfriet15} and 
  \cite{DBLP:journals/corr/abs-1810-08760},  the number of pebbles  we require for composition matches with  \cite{DBLP:journals/acta/Engelfriet15} who showed that this is the optimal number of pebbles needed.  Regarding the number of states,  \cite{DBLP:journals/acta/Engelfriet15} relies on the composition of deterministic  two-way transducers, which, as mentioned before, incurs at least a  doubly exponential blowup. As can be seen from (ii) above, reversibility gives a much better complexity. 
  The high level description for composition in \cite{DBLP:journals/corr/abs-1810-08760}
  does not give details on the number of states, and incurs $n$ extra pebbles.  As such,
  we do not compare with them.
  
  \item\label{item:uniformize}
  {\bf{Uniformization of non-deterministic Reversible Pebble Transducers}}.  
  Given a non-deterministic pebble transducer $\trans$ with $k$ pebbles and $n$ states, we
  construct a reversible pebble transducer $\trans'$ with $2^{\mathcal{O}((kn)^2)}$ states which
  uniformizes the relation $\sem{\trans}$ computed by $\trans$, i.e., computes a function 
  $\sem{\trans'}\subseteq\sem{\trans}$ with same domain.
  
  The uniformization \cite{DBLP:journals/acta/Engelfriet15} of non-deterministic to
  deterministic pebble transducers is done using results from
  \cite{DBLP:journals/tocl/EngelfrietH01}, where uniformization is done via a sequence of
  tranformations, namely: (\textit{i}) transductions computed by a 0-pebble transducer are
  definable using non-determinsitic MSO transductions, which can be uniformized using
  deterministic MSO transductions, (\textit{ii}) deterministic MSO transductions are
  computed by deterministic two-way transducers.  The uniformization for a given
  $k$-pebble transducer \cite{DBLP:journals/acta/Engelfriet15} is then computed as a
  composition of a $k$-counting transduction and the deterministic two-way transducer
  obtained in (\textit{ii}).  Due to (\textit{ii}), this may incur a non-elementary
  complexity, while our uniformization technique only has an exponential complexity.
   
  \item {\bf{Polyregular functions as Reversible Pebble Transducers}}.
  Polyregular functions can be realized as (a) deterministic two-way pebble transducers
  \cite{DBLP:journals/corr/abs-1810-08760} or, (b) as a composition of sequential
  functions, squaring and iterated reverse functions \cite{bojanpebble}.
  This results in two possibilities to compute a reversible pebble transducer for a 
  polyregular function.
  
  \textbf{a.} If the polyregular function is presented as a deterministic
  two-way pebble transducers then we may use the uniformization construction of
  \cref{item:uniformize} to obtain an equivalent reversible pebble transducer.
  
  \textbf{b.} 
  If a polyregular function $f$ is given as a modular expression involving
  sequential, squaring and iterated reverse functions, then one can directly synthesize a
  reversible pebble transducer computing $f$ using (\textit{i}) the reversible transducers
  for squaring and iterated reverse given in \cref{fig:rev-squaring-v4,fig:iterRev}
  (\textit{ii}) reversible transducers for the sequential functions obtained using 
  \cite[Theorem~2]{DFJL17}, and (\textit{iii}) our composition of the resultant
  reversible pebble transducers.  Note that the reversible machines in
  \cref{fig:rev-squaring-v4,fig:iterRev} are small with 6/5 states and 1/0 pebbles
  respectively.
  Once again, this is better than constructing first from the modular expression a
  deterministic pebble transducer for function $f$ and then uniformizing it to get a
  reversible one.
   
\end{enumerate}

\section{Reversible Pebble Transducers}
\label{sec:definitions}

\paragraph*{Pebble Transducers with Equality Tests.} In this section, we consider pebble transducers enhanced with equality
tests between pebble positions, in addition to the classical 
query checking only if a pebble is present/absent at the head position.
The equality tests are inspired from \cite{bojanpebble},
where $\leq$-comparisons between any two pebble positions, as well as between the head and
any pebble are allowed.  Enhancing the basic tests, on whether a pebble is present/absent
at the head, by allowing to check for any pair of pebbles whether or not they are at the
same position is useful, especially when composing two pebble transducers.

A $k$-pebble automaton with equality tests ($k$-$\mathsf{PA}_{=}$) is given by a
tuple $\auto=(Q,\Sigma,\delta,k,q_i,q_f)$ where $Q$ is a finite set of states, 
$\Sigma$ is the finite input alphabet, 
$\delta$ is the transition relation (see below),
and $q_i,q_f \in Q$ respectively are initial and final states.  

Given an input word $u\in\Sigma^*$, we define \emph{positions}
$\pos{u}=\{1,\ldots,|u|\}$.  We visualise the input of a PA as a circular word $\#u$ where
$\#\notin\Sigma$ is an endmarker reached when moving right (resp.\ left) from the last
(resp.\ first) position of $u$.
Keeping this in mind, we work with an \emph{extended} set of positions
$\epos{u}=\{0,1,\ldots,|u|\}$.  For instance,
$\setlength{\arraycolsep}{1pt}\begin{array}{ccccc} 
\# & a & b & c & d  \\
0 & 1 & 2 & 3 & 4  \end{array}$.

The head of the automaton during a run over $u$, occupies a position from $\epos{u}$,
denoted $\head$.
The states $Q$ of the pebble automaton are partitioned into 3 disjoint sets, $Q_{+1} \uplus
Q_0 \uplus Q_{-1}$ signifying movement of the head to the right, no movement of the head
and movement of the head to the left, respectively.
Given $\head \in \epos{u}$ and a state $q \in Q$, we define
$\head+q=\head$ if $q \in Q_0$;
$\head+q=(\head+1 \mod (|u|+1))$ if $q \in Q_{+1}$; and
$\head+q=(\head-1 \mod (|u|+1))$ if $q \in Q_{-1}$.
The intuition is that if the head moves to the right of the last position of the word,
then it reads $\#$, and if the head moves to the left from $\#$, it reads the last
position of $u$.

Likewise, the pebbles are dropped on $\epos{u}$ and follow a stack policy.
We denote by $\pebbles\in\epos{u}^{\leq k}$ the stack of positions where 
pebbles are dropped.  

A configuration of a $k$-pebble automaton on the input word $u$ is a tuple
$C=(q,\pebbles,\head)\in Q\times\epos{u}^{\leq k}\times\epos{u}$. 
In this configuration, the head of the automaton reads position $\head$, there are
$|\pebbles|$ pebbles dropped and the $i$-th pebble ($1\leq i\leq|\pebbles|$) is
dropped on position $\pebbles_{i}\in\epos{u}$.
The initial configuration is $(q_{i},\varepsilon,0)$ and the accepting (final) 
configuration is $(q_{f},\varepsilon,0)$. Hence, in order to accept, the automaton should 
lift all pebbles and move back to the endmarker.
Wlog, we assume that $q_{i},q_{f}\in Q_{0}$ with $q_{i}\neq q_{f}$, and there are no 
transitions starting from $q_{f}$ or going to $q_{i}$.

Transitions are of the form $t=(q,a,\varphi,\op,q')$ with states $q,q'\in Q$, input letter
$a\in\Sigmae=\Sigma\cup\{\#\}$, test $\varphi$ which is a conjunction of atoms of the form
$(\head=\pebble_{i})$, $\neg(\head=\pebble_{i})$, $(\pebble_{i}=\pebble_{j})$ or
$\neg(\pebble_{i}=\pebble_{j})$ (with $1\leq i,j\leq k$),
and operation $\op$ of the form\footnote{We chose to 
specify with index $i$ of $\drop{}$ and $\lift{}$ which pebble is dropped or lifted.
This allows to get determinism or reverse-determinism without adding extra tests to 
transitions.  We could have simply used $\drop{}$ and $\lift{}$ at the expense of adding
tests to transitions to ensure (reverse-) determinism.}
$\drop{i}$, $\lift{i}$ or $\nop$ (with $1\leq i\leq k$).
Transition $t$ is \emph{enabled} at configuration $C=(q,\pebbles,\head)$ (denoted
$C\models t$) if:
\begin{enumerate}
  \item
  the letter read is $a$, i.e., $a=(\#u)_{\head}$

  \item $\pebbles,\head\models\varphi$: the test succeeds, where
  \begin{itemize}
    \item $\pebbles,\head\models(\head=\pebble_{i})$ if $1\leq i\leq|\pebbles|$ and $\pebbles_{i}=\head$,
  
    \item $\pebbles,\head\models(\pebble_{i}=\pebble_{j})$ if $1\leq i,j\leq|\pebbles|$ and 
    $\pebbles_{i}=\pebbles_{j}$,
  \end{itemize}

  \item  $\pebbles,\head\models\op$: the operation can be executed where 
  \begin{itemize}
    \item $\pebbles,\head\models\nop$ is always true, 
    and we let $\nop(\pebbles,\head)=\pebbles$,

    \item $\pebbles,\head\models\drop{i}$ if $|\pebbles|=i-1$, 
    and we let $\drop{i}(\pebbles,\head)=\pebbles\cdot\head$,
  
    \item $\pebbles,\head\models\lift{i}$ if $|\pebbles|=i$ and $\head=\pebbles_{i}$, 
    and we let $\lift{i}(\pebbles,\head)=\pebbles_{1}\cdots\pebbles_{i-1}$.
  \end{itemize}
\end{enumerate}
When $C\models t$ we have $C\xrightarrow{t}C'=(q',\pebbles',\head')$ where 
$\pebbles'=\op(\pebbles,\head)$ and $\head'=\head+q'$.

Notice that the top pebble may be lifted from the pebble stack only when it is on the
position being read.  
We often simply write $\pebble_{i}$ (resp.\ $\neg\pebble_{i}$) for
the atomic test $(\head=\pebble_{i})$ (resp.\ $\neg(\head=\pebble_{i})$).
Note that, if a transition $t=(q,a,\pebble_{i},\drop{j},q')$
is enabled at some configuration $C$, then $i<j$ and
$C'\models(\pebble_{i}=\pebble_{j})$ where $C\xrightarrow{t}C'$.

An equality test $\pebble_{i}=\pebble_{i}$ allows to check if the size of the pebble stack
is at least $i$.  This is not possible if we only use atoms of the form $\pebble_{i}$ or
$\neg\pebble_{i}$ (which is not an expressivity problem since one may always store the
size of the pebble stack in the state of the automaton).

A run of $\auto$ is a sequence $C_{0}\xrightarrow{t_{1}}C_{1}\cdots\xrightarrow{t_{n}}C_{n}$
where $C_{\ell}$ are configurations on the input word $u$ and $t_{\ell}\in\delta$ are
transitions of $\auto$.  The run is initial if $C_{0}=(q_{i},\varepsilon,0)$ and final if
$C_{n}=(q_{f},\varepsilon,0)$.  It is accepting if it is both initial and final.

The automaton is \emph{deterministic} if for all pairs of transitions
$t=(q,a,\varphi,\op,r)$ and $t'=(q,a,\varphi',\op',r')$, if $t$ and $t'$ may be
simultaneously enabled (at some configuration $C$) then
$t=t'$.
The fact that an operation $\op$ is enabled may be written as a test $\test{\op}$:
$\test{\nop}=\true$, 
$\test{\drop{i}}=(\pebble_{i-1}=\pebble_{i-1})\wedge\neg(\pebble_{i}=\pebble_{i})$ (with 
$\pebble_{0}=\pebble_{0}$ identified with $\true$), and
$\test{\lift{i}}=(\head=\pebble_i) \wedge \neg(\pebble_{i+1} =\pebble_{i+1})$.
Then, the transitions $t$ and $t'$ may be simultaneously enabled if and only if the test
$\varphi\wedge\test{\op}\wedge\varphi'\wedge\test{\op'}$ is satisfiable.
This gives a more syntactic definition of determinism.
In particular, a transtion $t$ with operation $\lift{i}$ and a transition $t'$ with 
atomic test $\neg\pebble_{i}$ cannot be simultaneously enabled.

We define the reverse $\reverse{\op}$ of an operation $\op$ by
$\reverse{\nop}=\nop$, $\reverse{\drop{i}}=\lift{i}$ and $\reverse{\lift{i}}=\drop{i}$.
Notice that if $\pebbles,\head\models\op$ then
$\op(\pebbles,\head),\head\models\reverse{\op}$ and
$\pebbles=\reverse{\op}(\op(\pebbles,\head),\head)$.

A transition $t=(q,a,\varphi,\op,q')$ is \emph{reverse-enabled} at configuration
$C'=(q',\pebbles',\head')$ (denoted $C'\models_{\rev}t$) if there is a configuration
$C=(q,\pebbles,\head)$ enabling $t$ such that $C\xrightarrow{t}C'$.
More explicitly, the only possibility for configuration $C$ is given by
$\head=\head'-q'$ (corresponding to $\head'=\head+q'$), and
$\pebbles=\reverse{\op}(\pebbles',\head)$ (corresponding to $\pebbles'=\op(\pebbles,\head)$).
Then, $C'\models_{\rev}t$ 
iff $\pebbles',\head'-q'\models\reverse{\op}$ 
and $\reverse{\op}(\pebbles',\head'-q'),\head'-q'\models\varphi$.
In this case, we write $C=t^{-1}(C')=(q,\reverse{\op}(\pebbles',\head'-q'),\head'-q')$.
In particular, a transtion $t$ with operation $\drop{i}$ and a transition $t'$ with atomic
test $\neg\pebble_{i}$ and operation $\nop$ cannot be simultaneously reverse-enabled.

The automaton is \emph{reverse-deterministic} if for all pairs of transitions
$t=(q,a,\varphi,\op,r)$ and $t'=(q',a,\varphi',\op',r)$, if $t$ and $t'$ may be
simultaneously reverse-enabled (at some configuration $C'$) then $t=t'$.

The pebble automaton is \emph{reversible} if it is both \emph{deterministic} and 
\emph{reverse-deterministic}.

A $k$-pebble transducer with equality tests ($k$-$\mathsf{PT}_{=}$)
$\trans=(Q,\Sigma,\delta,k,q_i,q_f,\Gamma,\mu)$ is a $k$-pebble automaton
($k$-$\mathsf{PA}_{=}$) 
$\auto=(Q,\Sigma,\delta,k,q_i,q_f)$ extended with an output alphabet $\Gamma$ and an
output function $\mu$ mapping each transition of $\auto$ to a word in $\Gamma^{*}$.
The semantics of $\trans$ is a relation 
$\sem{\trans}\subseteq\Sigma^{*}\times\Gamma^{*}$ consisting of all pairs $(u,v)$ 
such that there is an accepting run 
$C_{0}\xrightarrow{t_{1}}C_{1}\cdots\xrightarrow{t_{n}}C_{n}$ on $\#u$ with 
$v=\mu(t_{1})\cdots\mu(t_{n})$.

A pebble transducer is deterministic (reverse-deterministic, reversible) whenever the
underlying pebble automaton is.  The semantics of a deterministic (reverse-deterministic,
reversible) pebble transducer is a partial function.

\begin{remark}\label{rem:0peb2w}
	Even though the semantics of 0-pebble transducers is different to that of two-way
	transducers as defined in~\cite{DFJL17}, the two models are equivalent.
	Given a transducer with $n$ states using one semantics,
	one can build a transducer realising the same relation using the other semantics,
	with \( \mathcal{O}(n) \) states.
	Moreover, determinism and reverse-determinism are preserved by this translation (see \cref{sec:app:0peb2w}).
\end{remark}

\paragraph*{Basic Pebble Transducers.}
A (basic) pebble automaton ($k$-$\mathsf{PA}$) or transducer ($k$-$\mathsf{PT}$) is one in
which equality tests of the form $(\pebble_{i}=\pebble_{j})$ or
$\neg(\pebble_{i}=\pebble_{j})$ are not used.  In some cases, we use a bit vector
$\overline{b}=(b_{1},\dots,b_{k})\in\{0,1\}^{k}$ to denote a full test $\varphi$ which is
a conjunction over all $1\leq i\leq k$ of atoms 
$\pebble_{i}$ if $b_{i}=1$ and $\neg\pebble_{i}$ if $b_{i}=0$.

\begin{figure}[t]
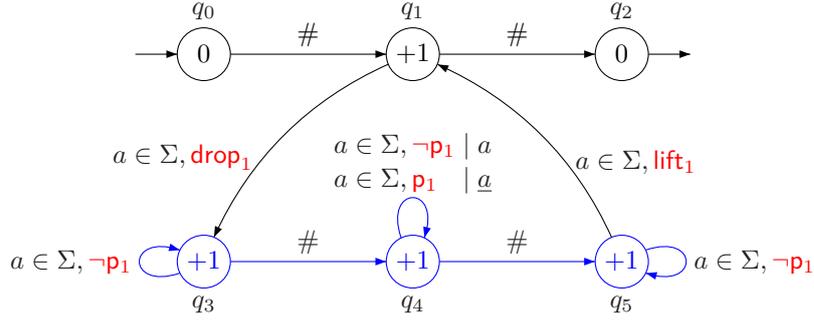

  \centering
  \gusepicture{squaring-RPT1}
  \caption{Reversible 1-pebble transducer for the function \emph{squaring}. The labels
  $i \in \{0,+1,-1\}$ inside the states represent that they lie in $Q_i$. 
  To keep the figure light, $\varepsilon$-outputs of transitions are omitted.
  The same transducer with state $q_{3}$ in $Q_{-1}$ would also compute the 
  squaring function.}
  \label{fig:rev-squaring-v4}
\end{figure}

\begin{example}\label{ex:squaring}
  \cref{fig:rev-squaring-v4} depicts a reversible $1$-pebble transducer computing the
  squaring function.  This function outputs the concatenation of as many
  copies of the input word as there are letters in it.  For the \( i^{th} \) copy, the \(
  i^{th} \) letter is marked.
	The blue part of the transducer copies its input and marks the positions on which the unique pebble is dropped.
  The black and red part initially drops the pebble on the first position, then lifts the pebble and drops it on the next position at each iteration. It moves to state $q_2$ and ends when the pebble is lifted from the last position.
	Without looking at the pebble operations, the transducer is deterministic (resp.\
	reverse-deterministic) at every state except state $q_5$ (resp.\ $q_3$).
	However, since the $\lift{1}$ operation can only be triggered at the position on which
	the pebble $1$ is dropped, the pebble guards $\lift{1}$ and $\lnot\pebble_{1}$ cannot be
	simultaneously enabled, and thus the transducer is deterministic at $q_5$.
	Symmetrically, the reverse of a $\drop{1}$ operation can only be enabled on the position
	marked by $\pebble_{1}$, and thus $\drop{1}$ and $\lnot\pebble_{1}$ cannot be
	simultaneously reverse-enabled, and then the transducer is reverse-deterministic at
	$q_3$, proving that we have a reversible transducer.
\end{example}

\begin{figure}[t]
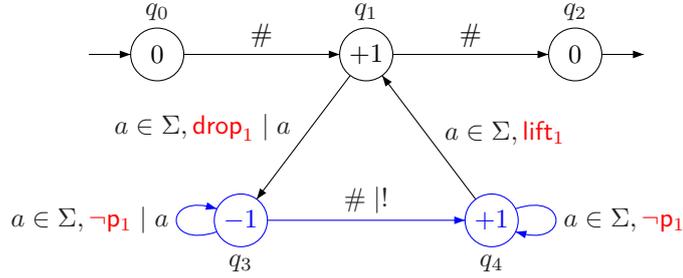

	\centering
	\gusepicture{all-prefixes-RPT1}
	\caption{Reversible 1-pebble transducer for the function \emph{all prefixes in reverse}.}
	\label{fig:rev-all-prefixes}
\end{figure}

\begin{example}
	\cref{fig:rev-all-prefixes} represents a reversible 1-pebble transducer for
	the function which maps a word to the concatenation of the reverse of all its prefixes,
	with a \( ! \) symbol between each one.
	A run of this transducer on the word \( abb \) produces the word
	\( a!ba!bba! \).
	
	\noindent Using the same arguments as in the previous example, one can show that the transducer is reversible.
\end{example}

%%%%%%%%%%%%%%%%%%%%%%

\section{Simulating equality tests of pebbles}
\label{sec:constr}

In this section, we show the equivalence between a basic $k$ pebble transducer
($k$-$\mathsf{PT}$) and a $k$ pebble transducer allowing equality tests
($k$-$\mathsf{PT}_{=}$).
The precise statement is given by \cref{thm:EqualityTests}, followed by
a proof sketch, the full proof can be found in~\cref{app:equiv}.

\begin{theorem}\label{thm:EqualityTests}
  Given a $k$-pebble transducer with equality tests $\mathcal{A}$ having $n$ states, one
  can construct a basic $k$-pebble transducer $\basic$ with $n2^{k^2}$ states such that
  $\sem{\mathcal{A}}=\sem{\basic}$.

  Moreover, if $\mathcal{A}$ is deterministic (reverse-deterministic, reversible) then so
  is $\basic$.
\end{theorem}

\begin{proof}[Sketch of proof]
  The main idea while simulating $\mathcal{A}$ is to keep track of which pairs of pebbles
  are on the same position in a coherent manner: that is, if $p_i=p_j$ and $p_j=p_\ell$
  then we also have $p_i=p_\ell$.  While constructing $\basic$, we store the equalities of
  the positions of the $k$ pebbles as a $k\times k$ boolean matrix $M$ where
  $M_{i,i}=1$ represents that the $i^{th}$ pebble has been dropped, and $M_{i,j}=1$ if the
  $i^{th}$ and $j^{th}$ pebbles are on the same position.  Then if we denote by $Q$ the
  set of states of $\mathcal{A}$, the set of states of $\basic$ is $Q'=Q\times
  \{0,1\}^{k^2}$.

  We now explain how the information is used and updated.
  Initially, no pebble is dropped, and the initial state of $\basic$ is $(q_i,
  \overline{0})$ where $\overline{0}$ denotes the $k \times k$ zero matrix and $q_i$ is
  the initial state of $\mathcal{A}$.  Similarly, a run is accepting only when all pebbles
  are lifted.  Thus the final state of $\basic$ is $(q_f, \overline{0})$ where $q_f$ is 
  the final state of $\mathcal{A}$.  
  Recall that basic pebble transducers can see whether a given pebble
  is at the same position as the reading head, thus transitions of $\mathcal{A}$ involving
  a guard of the form $\head=\pebble_i$ or $\neg(\head=\pebble_i)$ can be kept in
  $\basic$.  Tests of the form $\pebble_{i}=\pebble_{j}$ are replaced by $\true$ if 
  $M_{i,j}=1$ and by $\false$ otherwise.
  On transitions where there is no pebble lifting or dropping, the matrix $M$ does not
  need to be updated.  
  The $i^{th}$ pebble $\pebble_i$ can be dropped only if the $(i-1)^{th}$
  pebble is already dropped (that is, $M_{i-1,i-1}=1$) and the $i^{th}$ one is not
  ($M_{i,i}=0$).  On dropping $\pebble_i$, the matrix $M$ is updated by setting all
  coefficients $M_{i,j}$ and $M_{j,i}$ to $1$ if $j=i$ or $(j<i$ and $\head=p_j)$, and $0$
  otherwise.  Note
  that to do this, each transition of $\mathcal{A}$ with a $\drop{i}$ action is duplicated
  into $2^{i-1}$ disjoint transitions testing which subsets of $\pebble_1, \dots,
  \pebble_{i-1}$ are present at the head.  Only one transition can then be fired,
  depending on the bit vector generated by the pebbles at the current position.
  Similarly, the $i^{th}$ pebble $\pebble_i$ can be lifted if and only if
  $\head=\pebble_{i}$ and no larger pebbles are dropped $(i=k$ or
  $M_{i+1,i+1}=0)$.  On lifting $\pebble_i$, the matrix $M$ is updated by setting
  all coefficients of the $i^{th}$ row and column to $0$.  However to ensure the
  preservation of reverse-determinism, each $\lift{i}$ transition is duplicated into
  $2^{i-1}$ transitions that have complete and disjoint tests for the current bit vector
  over pebbles $\pebble_1,\dots,\pebble_{i-1}$.

  Finally, suppose that $\mathcal{A}$ is reversible.  The only extra transitions added to
  $\basic$ are when a pebble is dropped or lifted.  However, all the duplicate transitions
  that have been added in $\basic$ have disjoint tests (enforced by disjoint bit vectors).
  Thus $\basic$ is deterministic since $\mathcal{A}$ is.  On the other hand, consider we
  want to reverse a $\lift{i}$ transition in $\mathcal{A}$.  This means that the $i^{th}$
  pebble was at the previous position of the reading head, and further, there is only one
  bit vector (enriching the $i$th row and column of $M$) which was valid at the previous
  step.  The correct predecessor matrix is then the only one where the $i^{th}$ row and
  column is coherent with the bit vector of the previous position.
\end{proof}

%%%%%%%%%%%%%%%%%%
\section{Composition of pebble transducers}

The goal of this section is to give an efficient (in the size of the resulting transducer)
construction for the composition of pebble transducers.  Since the construction is rather
involved, we start with two simpler cases.  
The first one described in \cref{sec:reversing} shows how to compose a reversible
$n$-pebble transducer with the reverse function
$\mathsf{rev}\colon\Gamma^{*}\to\Gamma^{*}$ which takes a string of letters and outputs
the sequence of letters in reverse order.  For instance, $\mathsf{rev}(abac)=caba$.
The second simpler case is the composition of a reversible $n$-pebble transducer with a
deterministic $0$-pebble transducer.  This case is addressed in
\cref{sec:simple-case}.  Finally, the general case is solved in
\cref{sec:general-case}.

%%%%%%%%%%%%%%%%%%
\subsection{Reversing the output of a reversible pebble transducer}
\label{sec:reversing}

Let $\trans=(Q,\Sigma,\delta,n,q_i,q_f,\Gamma,\mu)$ be a reversible $n$-pebble 
transducer, possibly with equality tests.
The goal is to construct a reversible transducer
$\rtrans=(Q',\Sigma,\delta',n,q'_i,q'_f,\Gamma,\mu')$ which outputs the reverse
of the string produced by $\trans$: $\dom{\rtrans}=\dom{\trans}$ and for all
$u\in\dom{\trans}$, the string $\sem{\rtrans}(u)$ is the reverse of the string
$\sem{\trans}(u)$.

The set of states of $\rtrans$ is $Q'=Q$ but the polarity of the states is 
reversed: $Q'_{+1}=Q_{-1}$, $Q'_{0}=Q_{0}$ and $Q'_{-1}=Q_{+1}$. The initial and final 
states are exchanged: $q'_{i}=q_{f}$ and $q'_{f}=q_{i}$.
The transitions are also reversed in the following way.
Let $t=(q,a,\varphi,\op,q')$ be a transition of $\trans$.
Then, $\reverse{t}=(q',a,\reverse{(\varphi,\op)},q)$ is a transition of $\rtrans$ with
$\mu'(\reverse{t})=\mu(t)$.  We define 
$\reverse{(\varphi,\op)}=\op(\varphi),\reverse{\op}$: the operation is simply
$\reverse{\op}$ but the test $\op(\varphi)$ depends on both $\varphi$ and $\op$.
It remains to define $\op(\varphi)$. We want that 
$\pebbles,\head\models\varphi$ iff $\op(\pebbles,\head),\head\models\op(\varphi)$, 
assuming that $\op(\pebbles,\head)$ is defined.
We let $\op(\varphi_{1}\wedge\varphi_{2})=\op(\varphi_{1})\wedge\op(\varphi_{2})$
and $\op(\neg\varphi)=\neg\op(\varphi)$.
We let $\nop(\varphi)=\varphi$, and the remaining cases are given below
\begin{align*}
  \drop{\ell}(\head=\pebble_{i}) &= 
  \begin{cases}
    (\head=\pebble_{i}) & \text{if } i<\ell \\
    \false & \text{otherwise}
  \end{cases}
  &
  \drop{\ell}(\pebble_{i}=\pebble_{j}) &= 
  \begin{cases}
    (\pebble_{i}=\pebble_{j}) & \text{if } i,j<\ell \\
    \false & \text{otherwise}
  \end{cases}
  \\
  \lift{\ell}(\head=\pebble_{i}) &= 
  \begin{cases}
    (\head=\pebble_{i}) & \text{if } i<\ell \\
    \true & \text{if } i=\ell \\
    \false & \text{otherwise}
  \end{cases}
  &
  \lift{\ell}(\pebble_{i}=\pebble_{j}) &= 
  \begin{cases}
    (\pebble_{i}=\pebble_{j}) & \text{if } i,j<\ell \\
    (\head=\pebble_{i}) & \text{if } i<j=\ell \\
    (\head=\pebble_{j}) & \text{if } j<i=\ell \\
    \true & \text{if } i=j=\ell \\
    \false & \text{otherwise}
  \end{cases}
\end{align*}
We can easily check that $\pebbles,\head\models\varphi$ iff
$\op(\pebbles,\head),\head\models\op(\varphi)$ when $\op(\pebbles,\head)$ is defined.
Recall that, when $\pebbles'=\op(\pebbles,\head)$ is defined, 
then $\reverse{\op}(\pebbles',\head)$ is defined and we have
$\pebbles=\reverse{\op}(\pebbles',\head)$.
With $t=(q,a,\varphi,\op,q')$ and $\reverse{t}=(q',a,\op(\varphi),\reverse{\op},q)$,
we deduce that
\begin{equation}
  (q,\pebbles,\head) \xra{t} (q',\pebbles',\head')
  \quad\text{iff}\quad
  (q',\pebbles',\head'-q') \xra{\reverse{t}} (q,\pebbles,\head-q)
  \label{eq:reverse}
\end{equation}
We deduce that, for $u\in\Sigma^{*}$, the following sequence is a run of $\trans$ on $\lrend u$
\[
(q_{1},\pebbles^{1},\head_{1}) 
\xra{t_{1}} (q_{2},\pebbles^{2},\head_{2}) 
\xra{t_{2}} (q_{3},\pebbles^{3},\head_{3}) 
\cdots
\xra{t_{m-1}} (q_{m},\pebbles^{m},\head_{m}) 
\]
if and only if the following sequence is a run of $\rtrans$ on $\lrend u$
\[
(q_{1},\pebbles^{1},\head_{1}-q_{1}) 
\xla{\reverse{t_{1}}} (q_{2},\pebbles^{2},\head_{2}-q_{2}) 
\xla{\reverse{t_{2}}} (q_{3},\pebbles^{3},\head_{3}-q_{3}) 
\cdots
\xla{\reverse{t_{m-1}}} (q_{m},\pebbles^{m},\head_{m}-q_{m}) 
\]

Now, the run of $\trans$ is accepting iff 
$(q_{1},\pebbles^{1},\head_{1})=(q_{i},\varepsilon,0)$ and 
$(q_{m},\pebbles^{m},\head_{m})=(q_{f},\varepsilon,0)$.
Since $q_{i},q_{f}\in Q_{0}$, we deduce that the run of $\trans$ is accepting if and only 
if the corresponding run of $\rtrans$ is accepting, i.e., 
$(q_{m},\pebbles^{m},\head_{m}-q_{m})=(q'_{i},\varepsilon,0)=(q_{f},\varepsilon,0)$ and
$(q_{1},\pebbles^{1},\head_{1}-q_{1})=(q'_{f},\varepsilon,0)=(q_{i},\varepsilon,0)$.

We deduce that $\dom{\rtrans}=\dom{\trans}$ and, by definition of $\mu'$, for
$u\in\dom{\trans}$, the string $\sem{\rtrans}(u)$ is the reverse of the string
$\sem{\trans}(u)$.

To show that $\rtrans$ is deterministic, consider a pair of transitions 
$\reverse{t_{1}}=(q',a,\reverse{(\varphi_{1},\op_{1})},q_{1})$ and
$\reverse{t_{2}}=(q',a,\reverse{(\varphi_{2},\op_{2})},q_{2})$ where
$t_1=(q_1,a,\varphi_1,\op_1,q')$ and $t_2=(q_2,a,\varphi_2,\op_2,q')$ are transitions of
$\trans$.  Assume that $\reverse{t_{1}}$ is enabled at some configuration
$C'=(q',\pebbles',\head')$.  Using \eqref{eq:reverse}, we deduce that 
$t_{1}$ is reverse-enabled at configuration $C=(q',\pebbles',\head'+q')$. 
Similarly, if $\reverse{t_{2}}$ is enabled at $C'$ then $t_{2}$ is reverse-enabled at 
$C$. Since $\trans$ is reverse-deterministic, we deduce that $\rtrans$ is deterministic.

We can prove similarly that $\rtrans$ is reverse-deterministic using the fact that
$\trans$ is determinisitic.  In particular, we show using \eqref{eq:reverse} that
transition $\reverse{t}$ is reversed-enabled at some configuration $C'=(q,\pebbles,\head)$
if and only if transition $t$ is enabled at configuration $C=(q,\pebbles,\head+q)$.
We deduce that $\rtrans$ is reversible.

%%%%%%%%%%%%%%%%%%
\subsection{Composition: Simple case}
\label{sec:simple-case}

In this subsection, we prove the following result.

\begin{theorem}\label{thm:composition-simple-case}
  Let $\trans=(Q,\Sigma,\delta,n,q_i,q_f,\Gamma,\mu)$ be a \emph{reversible} $n$-pebble
  transducer (possibly with equality tests) computing a function
  $f\colon\Sigma^{*}\to\Gamma^{*}$.
  Let $\trans'=(Q',\Gamma,\delta',0,q'_i,q'_f,\Delta,\mu')$ be a \emph{deterministic}
  $0$-pebble transducer computing a function $g\colon\Gamma^{*}\to\Delta^{*}$.
  We can construct a \emph{deterministic} $n$-pebble transducer $\trans''$ with 
  at most $2\cdot|Q|\cdot[Q'|$ many states
  computing the composition $g\circ f\colon\Sigma^{*}\to\Delta^{*}$.
  Moreover, if $\trans'$ is \emph{reversible} then so is $\trans''$.
\end{theorem}

\begin{proof}
  Wlog, we assume that transducer $\trans$  
   outputs at most one letter on each transition,
  i.e., $\mu(t)\in\Gamma\cup\{\varepsilon\}$ for each transition $t$ of
  $\trans$.  This can be done while preserving reversibility in the following way:
  a transition $t=(q,a,\varphi,\op,s)$ producing $v_1\cdots v_n$ ($n\geq2$) is decomposed
  into $n$ transitions $t_1=(q,a,\varphi\wedge \overline{\op},\nop,(q,v_1))$, 
	$t_n=((q,v_1\ldots v_{n-1}),a,\varphi,\op,s)$ and for $1<i<n$, 
  $t_i=((q,v_1\ldots v_{i-1}),a,\varphi\wedge \overline{\op},\nop,(q,v_1\ldots v_i))$.  
  We set $\mu(t_i)=v_i$.
	Since we maintain the tests $\varphi$ and $\overline{\op}$ in the sequence of
	transitions $t_i$, even if a state $(q,v)$ happens to be used by several transitions
	they should be disjoint due to the reversibility of the inital transducer.
  It will be also more convenient to assume that, when running on $\lrend u$ with
  $u\in\Sigma^{*}$, the transducer $\trans$ writes $\lrend v$ with $v\in\Gamma^{*}$.  This is
  achieved by setting $\mu(t_{0})=\lrend$ where $t_{0}=(q_{i},\lrend,\true,{-},{-})$ is the
  unique transition enabled at the initial configuration $C_{0}=(q_{i},\varepsilon,0)$ of
  $\trans$ (pebble stack empty and head on position $0$).
  Finally, we also assume that the transducer $\trans'$ always fully read its input.  This
  will allow us to ensure that the transducer $\trans''$ only accepts words that belong to
  $\dom{\trans}$.  This can be achieved by addding to $\trans'$ a single state $r\in
  Q'_{+1}$, and if $(q'_i,\#,\varphi,\op,s)$ is the initial transition of $\delta'$, we
  replace it with $(q'_i,\#,\true,\nop,r)$, $(r,a,\true,\nop,r)$ for $a\neq\#$ and
  $(r,\#,\varphi,\op,s)$.  This construction preserves reversibility.

  We construct the $n$-pebble transducer 
  $\trans''=(Q'',\Sigma,\delta'',n,q''_i,q''_f,\Delta,\mu'')$ as a synchronized product 
  of $\trans$ and $\trans'$.
  The set of states is $Q''=(Q\times Q')\cup(\widehat{Q}\times(Q'_{-1}\cup Q'_{+1}))$
  where $\widehat{Q}=\{\hat{q}\mid q\in Q\}$ is a disjoint copy of $Q$.  
  We have $|Q''|\leq 2\cdot|Q|\cdot[Q'|$. We define
  \begin{align*}
    Q''_{0}  &= (Q\times Q') \cup (\widehat{Q_{0}}\times(Q'_{-1}\cup Q'_{+1})) \\
    Q''_{+1} &= (\widehat{Q_{+1}}\times Q'_{+1}) \cup (\widehat{Q_{-1}}\times Q'_{-1}) \\
    Q''_{-1} &= (\widehat{Q_{+1}}\times Q'_{-1}) \cup (\widehat{Q_{-1}}\times Q'_{+1}) 
    \,.
  \end{align*}
  The initial state is $q''_{i}=(q_{i},q'_{i})\in Q''_{0}$
  and the final state is $q''_{f}=(q_{i},q'_{f})\in Q''_{0}$.
  
  The intuition is that, in a state $(q,q')\in Q\times Q'$, transducer $\trans''$ will 
  synchronize a pair $(t,t')$ of transitions of $\trans$ and $\trans'$ where the output 
  $\mu(t)\in\Gamma$ of $t$ is the input letter of $t'$.
  On the other hand, when $\trans''$ is in a state $(\hat{q},q')$, it will simulate 
  the computation of $\trans$ forward if $q'\in Q'_{+1}$ (resp.\ backward if $q'\in 
  Q'_{-1}$) using transitions $t$ of $\trans$ producing $\mu(t)=\varepsilon$ until the 
  computation reaches a transition $t$ of $\trans$ producing a letter $\mu(t)\in\Gamma$.
  
  To handle the fact that the input of $\trans$ is a circular word, we extend
  $\delta$ with the transition
  $t_{f,i}=(q_{f},\lrend,\neg(\pebble_{1}=\pebble_{1}),\nop,q_{i})$ which is only enabled in 
  the final configuration $(q_{f},\varepsilon,0)$ of $\trans$ and moves to the initial 
  configuration $(q_{i},\varepsilon,0)$. We let $\mu(t_{f,i})=\varepsilon$.
  
  We define now the transitions $\delta''$ and the output function $\mu''$.
  Let $t=(q,a,\varphi,\op,s)$ be a transition of $\trans$.  
  If $\mu(t)=a'\in\Gamma$, for each transition $t'=(q',a',s')$ of
  $\trans'$, we introduce in $\trans''$ the synchronized transition $t''$ defined in 
  \cref{eq:simple-tr+,eq:simple-tr-,eq:simple-tr0} with
  $\mu''(t'')=\mu'(t')$.
  
  Firstly, we synchronize a transition and switch to the simulation mode that
  will search for either the next or previous production, depending on the polarity of the
  destination state.  If $\trans'$ is going to the right, we advance in the run of $\trans$
  and apply the transition $t$.  If $\trans'$ is going to the left, we need to start
  rewinding $\trans$ to compute the transition which led to state $q$:
  \begin{align}
    t''=(q,q') & \xra{a,\varphi,\op} (\hat{s},s') &&\text{if } s'\in Q'_{+1}
    \tag{\textsf{tr-a}}\label{eq:simple-tr+}
    \\
    t''=(q,q') & \xra{a,\varphi\wedge\overline{\op},\nop} (\hat{q},s') 
    &&\text{if } s'\in Q'_{-1}
    \tag{\textsf{tr-b}}\label{eq:simple-tr-}
    \\
    t''=(q,q') & \xra{a,\varphi\wedge\overline{\op},\nop} (q,s') &&\text{if } s'\in Q'_{0}
    \tag{\textsf{tr-c}}\label{eq:simple-tr0}
  \end{align}
  Secondly, if we are in simulation mode of $\trans$, we need to keep simulating until we
  reach a non empty production of $\trans$.  Hence if $\mu(t)=\varepsilon$ then we stay in
  the simulation mode of $\trans$.  For each $q'\in Q'_{-1}\cup Q'_{+1}$, we introduce in
  $\trans''$ the following transitions $t''$ with $\mu''(t'')=\varepsilon$.
  Note that depending on the polarity of $q'$, we either need to advance or rewind the
  computation of $\trans$.  Rewinding is done following the constructions detailled
  in~\cref{sec:reversing}.
  \begin{align}
    t''=(\hat{q},q') & \xra{a,\varphi,\op} (\hat{s},q') &&\text{if } q'\in Q'_{+1}
    \tag{\textsf{mv-a}}\label{eq:simple-mv+}
    \\
    t''=(\hat{s},q') & \xra{a,\reverse{(\varphi,\op)}} (\hat{q},q') &&\text{if } q'\in Q'_{-1}
    \tag{\textsf{mv-b}}\label{eq:simple-mv-}
  \end{align}
  Finally, if $\mu(t)=a'\in\Gamma$ we found the transition of $\trans$ which produces the
  input letter $a'$ to be read by $\trans'$.  Hence, we switch from the simulation mode of
  $\trans$ to its synchronization mode.  For each $q'\in Q_{-1}\cup Q_{+1}$, we add the
  following transitions $t''$ with $\mu''(t'')=\varepsilon$.

  \begin{align}
    t''=(\hat{q},q') & \xra{a,\varphi\wedge\overline{\op},\nop} (q,q') 
    &&\text{if } q'\in Q'_{+1}
    \tag{\textsf{sw-a}}\label{eq:simple-sw+}
    \\
    t''=(\hat{s},q') & \xra{a,\reverse{(\varphi,\op)}} (q,q') 
    &&\text{if } q'\in Q'_{-1}
    \tag{\textsf{sw-b}}\label{eq:simple-sw-}
  \end{align}
  After the proof, we illustrate this construction in \cref{ex:composition-simple-case}. 
  It may be useful to read this example before reading the rest of the proof.
  
  Now, we prove the correctness of the construction, i.e., that $\trans''$ computes the
  function $g\circ f$. After that, we will prove that the constructed transducer
  $\trans''$ is deterministic (\cref{claim:simple-det}), and that $\trans''$ is
  reverse-deterministic if $\trans'$ is reversible (\cref{claim:simple-rev-det}).

  We fix $u\in\dom{f}\subseteq\Sigma^{*}$ and $v=f(u)\in\Gamma^{*}$.  
  We consider the accepting run of $\trans$ on
  $\lrend u$ producing $\lrend v$:
  \[
  C_{0}=(q_{i},\varepsilon,0)\xra{t_{0}}C_{1}\xra{t_{1}}C_{2}\xra{t_{2}}
  \cdots C_{N}\xra{t_{N}}C_{N+1}=(q_{f},\varepsilon,0)
  \]
  where $C_\ell=(q_\ell,\pebbles_\ell,\head_\ell)$ are configurations and
  $t_\ell\in\delta\setminus\{t_{f,i}\}$ are transitions of $\trans$.  Notice that the
  configurations in this accepting run are pairwise distinct.  This follows from the fact
  that the automaton is deterministic and there are no transitions (other than $t_{f,i}$)
  starting from the final state $q_{f}$.

  For each position $i\in\epos{v}$, let $\enc{i}$ be the index of the transition in the
  above run producing the letter at position $i$ in $\lrend v$.  
  Since $\mu(t_{0})=\lrend$ we have $j_{0}=0$.
  Also, $\mu(t_{\ell})=\varepsilon$ if $\enc{i}<\ell<\enc{i+1}$ and
  $\mu(t_{\enc{i}})\in\Gamma\cup\{\#\}$ if $i\in\epos{v}$.
  
  The main idea is to encode the head position $\head'$ of $\trans'$ by the configuration
  $C_{\enc{\head'}}$. 
  More precisely, the encoding of a configuration
  $C'=(q',\head')$ of $\trans'$ on $\lrend v$ is defined as
  $\enc{C'}=((q_\ell,q'),\pebbles_\ell,\head_\ell)$ where $\ell=\enc{\head'}$.
  
  \begin{restatable}{claim}{simplesimulate}\label{claim:simple-simulate}
    There is a transition $(q',\head')\xra{t'}(s',\head'+s')$ of $\trans'$ on $\lrend v$ if and
    only if there is a nonempty run $\enc{(q',\head')}\xra{+}\enc{(s',\head'+s')}$ of
    $\trans''$ on $\lrend u$ which does not use an intermediate state in $Q\times Q'$.
  \end{restatable}
\begin{claimproof}
	($\implies$) Assume that $\trans''$ has a transition
	$(q',\head')\xra{t'}(s',\head'+s')$ on $\lrend v$.
	Let $a'$ be the letter of $\lrend v$ at position $\head'$. 
	Let $\ell=\enc{\head'}$ and $\ell'=\enc{\head'+s'}$.
	We have $\mu(t_{\ell})=a'$ and $t'=(q',a',s')$. 
	Let $t''_{\ell}$ be the transition of $\trans''$ constructed from the pair $(t,t')$
	using \cref{eq:simple-tr+,eq:simple-tr-,eq:simple-tr0}.
	
	If $s'\in Q'_{0}$ then $\ell'=\ell$ and
	$\enc{(q',\head')}\xra{t''_{\ell}}\enc{(s',\head'+s')}$: we are done.
	
	Assume that $s'\in Q'_{+1}$.  Then $\ell'>\ell$.\footnote{Or $\ell=|v|$ and
			$\ell'=0$ which can be handled similarly using the extra transition $t_{f,i}$
			producing $\varepsilon$ allowing to move from the final configuration of $\trans$ to
			its initial configuration.}
  For each $\ell<i<\ell'$, let $t''_{i}$ be the
	transition of $\trans''$ constructed from $t_{i}$ by \eqref{eq:simple-mv+}.  Let
	$t''_{\ell'}$ be the transition of $\trans''$ constructed from $t_{\ell'}$ by
	\eqref{eq:simple-sw+}.  We can easily check that
	\begin{align*}
	((q_{\ell},q'),\pebbles_{\ell},\head_{\ell}) \xra{t''_{\ell}}
	((\hat{q_{\ell+1}},s'),\pebbles_{\ell+1},\head_{\ell+1}) 
	\xra{t''_{\ell+1}\cdots t''_{\ell'-1}}
	((\hat{q_{\ell'}},s'),\pebbles_{\ell'},\head_{\ell'})\\
	\ldots
	\xra{t''_{\ell'}} ((q_{\ell'},s'),\pebbles_{\ell'},\head_{\ell'})
	\end{align*}
	Therefore, $\enc{(q',\head')}\xra{t''_{\ell}t''_{\ell+1}\cdots
		t''_{\ell'-1}t''_{\ell'}}\enc{(s',\head'+s')}$.
	
	The third case is $s'\in Q'_{-1}$.  Then $\ell'<\ell$.\footnote{Or $\ell=0$
			and $\ell'=|v|$ which can be handled similarly using the extra transition $t_{f,i}$.}
	For each $\ell'<i<\ell$, let $t''_{i}$ be the transition of $\trans''$ constructed
	from $t_{i}$ by \eqref{eq:simple-mv-}.  Let $t''_{\ell'}$ be the transition of
	$\trans''$ constructed from $t_{\ell'}$ by \eqref{eq:simple-sw-}.  We can easily check
	that
	\begin{align*}
		((q_{\ell},q'),\pebbles_{\ell},\head_{\ell}) &\xra{t''_{\ell}}
		((\hat{q_{\ell}},s'),\pebbles_{\ell},\head_{\ell}-q_{\ell}) \xra{t''_{\ell-1}}
		((\hat{q_{\ell-1}},s'),\pebbles_{\ell-1},\head_{\ell-1}-q_{\ell-1}) 
		\\
		\cdots &\xra{t''_{\ell'+1}} 
		((\hat{q_{\ell'+1}},s'),\pebbles_{\ell'+1},\head_{\ell'+1}-q_{\ell'+1})
		\xra{t''_{\ell'}} ((q_{\ell'},s'),\pebbles_{\ell'},\head_{\ell'}-0)
	\end{align*}
	We deduce that $\enc{(q',\head')}\xra{t''_{\ell}t''_{\ell-1}\cdots
		t''_{\ell'+1}t''_{\ell'}}\enc{(s',\head'+s')}$ as desired.
	
	($\Longleftarrow$) The converse can be shown similarly by noting that, depending on
	the polarity of $s'$, in a run $\enc{(q',\head')}\xra{+}\enc{(s',\head'+s')}$ which
	has no intermediate states in $Q\times Q'$, the sequence of transitions used is one of
	\begin{itemize}
		\item[(a)] a transition constructed with \eqref{eq:simple-tr+}, followed by a 
		(possibly empty) sequence of transitions constructed with \eqref{eq:simple-mv+}, 
		followed by a transition constructed with \eqref{eq:simple-sw+},
		
		\item[(b)] a transition constructed with \eqref{eq:simple-tr-}, followed by a 
		(possibly empty) sequence of transitions constructed with \eqref{eq:simple-mv-}, 
		followed by a transition constructed with \eqref{eq:simple-sw-},
		
		\item[(c)] a single transition constructed with \eqref{eq:simple-tr0}.
		\claimqedhere
	\end{itemize}
\end{claimproof}  
  Now, we can show that the transducer $\trans''$ computes the function $g\circ f$.  With
  the notation above, assuming that $v\in\dom{g}$ we consider the accepting run $\rho'$ of
  $\trans'$ on $\lrend v$:
  \begin{equation}
    C'_{0}=(q'_{i},0)\xra{t'_{0}}C'_{1}\xra{t'_{1}}C'_{2}\cdots 
    C'_{N'}\xra{t'_{N'}}C'_{N'+1}=(q'_{f},0) \,.
    \label{eq:rho'}
  \end{equation}
  Using \cref{claim:simple-simulate} we obtain the accepting run $\rho''$ of $\trans''$ on
  $\lrend u$
  \begin{align*}
    \enc{C'_{0}}=((q_{i},q'_{i}),\varepsilon,0)\xra{+}{}
    &\enc{C'_{1}}\xra{+}\enc{C'_{2}}
    \\
    \cdots~ &\enc{C'_{N'}}\xra{+}\enc{C'_{N'+1}}=((q_{i},q'_{f}),\varepsilon,0) \,.
  \end{align*}
  It is easy to see that $\rho''$ produces the same output string $g(v)$ as $\rho'$.
  We get $\sem{\trans''}(u)=g(v)=g\circ f(u)$.
  
  Conversely, an accepting run $\rho''$ of $\trans''$ on $\lrend u$ can be split according to 
  its visits to states in $Q\times Q'$:
  \[
  C''_{0}=((q_{i},q'_{i}),\varepsilon,0)\xra{+}C''_{1}\xra{+}C''_{2}\cdots 
  C''_{N'}\xra{+}C''_{N'+1}=((q_{i},q'_{f}),\varepsilon,0) \,.
  \]
  where the $C''_{i}$ are the configurations with state in $Q\times Q'$.
  We have $C''_{0}=\enc{C'_{0}}$.  By induction, and using \cref{claim:simple-simulate},
  we can easily show that, for $0\leq i\leq N'$, there are transitions $t'_{i}$ and
  configurations $C'_{i+1}$ such that $C'_{i}\xra{t'_{i}}C'_{i+1}$ on $\lrend v$ and
  $C''_{i+1}=\enc{C'_{i+1}}$. We deduce that \eqref{eq:rho'} gives an accepting run 
  $\rho'$ of $\trans'$ on $\lrend v$. Again, the output string $\sem{\trans''}(u)$ of $\rho''$ 
  is the same as the output string $g(v)$ of $\rho'$.
  Because we assumed that $\trans'$ reads its whole input, by emulating $\trans'$ we know
  that the transducer $\trans''$ fully simulates $\trans$, which ensures that it accepts
  an input $u$ if, and only if, it belongs to $\dom{\trans}$ and $\trans(u)$ belongs to
  $\dom{\trans'}$.

\begin{claim}\label{claim:simple-det}
	The transducer $\trans''$ from \cref{thm:composition-simple-case} is deterministic.
\end{claim}

\begin{claimproof}
	Consider two transitions $t''_{1}$ and $t''_{2}$ starting from some state $q''$,
	reading some $a\in\Sigma$ and both enabled at some configuration
	$C''=(q'',\pebbles,\head)$ of $\trans''$.
	
	Assume first that $q''=(\hat{q},q')$ with $q'\in Q'_{+1}$.
	For each $i\in\{1,2\}$, let $t_{i}=(q,a,\varphi_{i},\op_{i},s_{i})$ be the transition of 
	$\trans$ giving rise to $t''_{i}$ with \eqref{eq:simple-mv+} or \eqref{eq:simple-sw+}.
	Since $t''_{i}$ is enabled at configuration $C''$, we deduce that $t_{i}$ is enabled
	at configuration $C=(q,\pebbles,\head)$.
	Hence, $t_{1},t_{2}$ are both enabled at $C$.  We get $t_{1}=t_{2}$ by determinism of
	$\trans$ and therefore $t''_{1}=t''_{2}$.
	
	Assume now that $q''=(\hat{s},q')$ with $q'\in Q'_{-1}$.
	For each $i\in\{1,2\}$, let $t_{i}=(q_{i},a,\varphi_{i},\op_{i},s)$ be the transition of 
	$\trans$ giving rise to $t''_{i}$ with \eqref{eq:simple-mv-} or \eqref{eq:simple-sw-}.
	Since $t''_{i}$ is enabled at configuration $C''$, we deduce that $\reverse{t_{i}}$ is
	enabled at configuration $C=(s,\pebbles,\head)$.  Therefore, by \eqref{eq:reverse}, 
	transition $t_{i}$ is reverse-enabled at configuration $C'=(s,\pebbles,\head+s)$.
	Hence, $t_{1},t_{2}$ are both reverse-enabled at $C'$.  
	We get $t_{1}=t_{2}$ by reverse-determinism of $\trans$. Hence, $t''_{1}=t''_{2}$.
	
	Finally, assume that $q''=(q,q')$.
	For each $i\in\{1,2\}$, let $t_{i}=(q,a,\varphi_{i},\op_{i},s_{i})$, $a'_{i}=\mu(t_{i})$ 
	and $t'_{i}=(q',a'_{i},s'_{i})$ be the transitions of $\trans$ and $\trans'$ giving
	rise to $t''_{i}$ with \eqref{eq:simple-tr+}, \eqref{eq:simple-tr-} or
	\eqref{eq:simple-tr0}.
	Since $t''_{i}$ is enabled at $C''$, we deduce that $t_{i}$ is enabled
	at configuration $C=(q,\pebbles,\head)$.
	Hence, $t_{1},t_{2}$ are both enabled at $C$ and we get $t_{1}=t_{2}$ by determinism of
	$\trans$. It follows that $a'_{1}=\mu(t_{1})=\mu(t_{2})=a'_{2}$ and we get 
	$t'_{1}=t'_{2}$ by determinism of $\trans'$. Therefore $t''_{1}=t''_{2}$.
\end{claimproof}

\begin{claim}\label{claim:simple-rev-det}
	If the transducer $\trans'$ from \cref{thm:composition-simple-case} is reversible, then $\trans''$ is reverse-deterministic.
\end{claim}

\begin{claimproof}
	Consider two transitions $t''_{1}$ and $t''_{2}$ of $\trans''$ ending in some state
	$q''$, reading some $a\in\Sigma$ and both reverse-enabled at some configuration
	$C''=(q'',\pebbles,\head)$.
	
	Assume first that $q''=(\hat{s},s')$ with $s'\in Q'_{+1}$.
	For each $i\in\{1,2\}$, let $t_{i}=(q_{i},a,\varphi_{i},\op_{i},s)$ be the transition of 
	$\trans$ giving rise to $t''_{i}$ with \eqref{eq:simple-mv+} or \eqref{eq:simple-tr+}.
	Since $t''_{i}$ is reverse-enabled at configuration $C''$, we deduce that $t_{i}$ is
	reverse-enabled at configuration $C=(s,\pebbles,\head)$.
	Hence, $t_{1},t_{2}$ are both reverse-enabled at $C$ and we get $t_{1}=t_{2}=t$ by
	reverse-determinism of $\trans$.
	Now, either $\mu(t)=\varepsilon$ and both $t''_{1}$ and $t''_{2}$ are constructed 
	from $t$ and $s'$ with \eqref{eq:simple-mv+}. We get $t''_{1}=t''_{2}$.
	Or $\mu(t)=a'\in\Gamma$ and $t''_{i}$ is constructed from $t$ and some transition 
	$t'_{i}=(q'_{i},a',s')$ of $\trans'$ with \eqref{eq:simple-tr+}.
	Using the reverse-determinism of $\trans'$, we deduce that $t'_{1}=t'_{2}$ and 
	$t''_{1}=t''_{2}$.
	
	Assume now that $q''=(q,q')\in Q''_{0}$. 
	Depending on the polarity of $q'$, \emph{both} transitions 
	$t''_{1}$ and $t''_{2}$ are constructed with the \emph{same} case 
	\eqref{eq:simple-sw+}, \eqref{eq:simple-sw-} or \eqref{eq:simple-tr0} from some 
	transitions $t_{1}$ and $t_{2}$ of $\trans$. 
	In the cases \eqref{eq:simple-sw+} or \eqref{eq:simple-tr0}, the pebble stack does not
	change when executing $t''_{i}$ whose operation is $\nop$, and the head does not move
	since $q''\in Q''_{0}$.  Since $t''_{i}$ is reverse-enabled at $C''$, we deduce that
	$\pebbles,\head\models\varphi_{i}\wedge\overline{\op_{i}}$.  Hence, $t_{1}$ and
	$t_{2}$ are both enabled at configuration $C=(q,\pebbles,\head)$.
	It follows that $t_{1}=t_{2}$ using the determinism of $\trans$, and then 
	$t''_{1}=t''_{2}$, using the reverse-determinism of $\trans'$ in case 
	\eqref{eq:simple-tr0}. 
	Consider now the case \eqref{eq:simple-sw-} and let 
	$t_{i}=(q,a,\varphi_{i},\op_{i},s_{i})$. Since $t''_{i}$ is reversed-enabled at 
	$C''$, we deduce that $\reverse{t_{i}}$ is reverse-enabled at $C=(q,\pebbles,\head)$.
	By \eqref{eq:reverse}, this implies that $t_{i}$ is enabled at $(q,\pebbles,\head+q)$.
	We conclude as in the previous case that $t_{1}=t_{2}$ and $t''_{1}=t''_{2}$.
	
	Finally, assume that $q''=(\hat{q},q')\in\widehat{Q}\times Q'_{-1}$.
	For each $i\in\{1,2\}$, let $t_{i}=(q,a,\varphi_{i},\op_{i},s_{i})$ be the transition of 
	$\trans$ giving rise to $t''_{i}$ with \eqref{eq:simple-mv-} or \eqref{eq:simple-tr-}.
	In case \eqref{eq:simple-mv-}, we see that $\reverse{t_{i}}$ is reverse-enabled at 
	$(q,\pebbles,\head)$ and using \eqref{eq:reverse} we deduce that $t_{i}$ is enabled 
	at $C=(q,\pebbles,\head+q)$.
	In case \eqref{eq:simple-tr-}, we note that the pebble stack does not change while 
	executing transition $t''_{i}$ and the head moves by the polarity of $q''$.
	Since $t''_{i}$ is reverse-enabled at $C''$, we have
	$\pebbles,\head-q''\models\varphi_{i}\wedge\overline{\op_{i}}$.  Using
	$\head-q''=\head+q$ we deduce that $t_{i}$ is enabled at $C=(q,\pebbles,\head+q)$.
	We have shown that in both cases \eqref{eq:simple-mv-} or \eqref{eq:simple-tr-}, 
	transitions $t_{1}$ and $t_{2}$ are both enabled at $C$. 
	It follows that $t_{1}=t_{2}$ using the determinism of $\trans$, and then 
	$t''_{1}=t''_{2}$, using the reverse-determinism of $\trans'$ in case 
	\eqref{eq:simple-tr-}. 
	\claimqedhere
\end{claimproof}

  This concludes the proof of \cref{thm:composition-simple-case}.
\end{proof}

\newcommand{\tr}[2]{\overunderset{#1}{#2}{\rightarrow}}
\newcommand{\mv}[2]{\overunderset{#1}{#2}{\rightsquigarrow}}
\newcommand{\sw}[2]{\overunderset{#1}{#2}{\dashrightarrow}}
\begin{example}\label{ex:composition-simple-case}
	We illustrate the construction from \cref{thm:composition-simple-case},
	with \( \trans \) being a slight modification of the \emph{squaring}
	function realized by the transducer from \cref{fig:rev-squaring-v4},
	and \( \trans' \) is the function \emph{iterated reverse} realized by the
	transducer from \cref{fig:iterRev}.
	The \emph{squaring} function is modified as such:
	instead of outputting marked letters \( \underline{a} \) when reading
	the letter on which the pebble is placed, from state \( q_4 \),
	it outputs \( ! \).
	Now the output of this automaton is of a form that is expected for
	the function \emph{iterated reverse}.
	
	We show a fragment of the run of \( \trans'' \) on the input \( \lrend u = \lrend bcd \).
	On this input, \( \trans \) produces the output \( \lrend v = \lrend ! c d b ! d b c ! \),
	on which \( \trans' \) produces \( ! b d c ! c d b \).
	The fragment of the run illustrated below starts when transducer \( \trans' \) reads the
	second \( ! \) of $v$, and needs to rewind the computation of \( \trans \) in order to process
	the infix \( c d b \) of \( u \).
	Notice that to produce the second \( ! \) of \( v \),
	\( \trans \) is producing the second copy of \( u \),
	hence the pebble placed on the \( c \) in the initial configuration of the run.
	The fragment ends when \( \trans' \) is done reversing the current infix,
	and is about to find the previous \( ! \) symbol.

  Transitions of the type \eqref{eq:simple-mv-} are represented by \( \mv{}{} \),
	those of the type \eqref{eq:simple-sw-} by \( \sw{}{} \),
	and those of the type \eqref{eq:simple-tr-} 
  by \( \tr{}{} \).
	Notice that in this fragment of run, no transition where \( \trans' \)
	needs to go right on \( \lrend v \) are represented.
	The run is as follow:
	\begin{align*}
		\big( (q_4, q_1'), 2, 2 \big) &\tr{c}{1}
		\big( (\hat{q_4}, q_2'), 2, 1 \big) \sw{b}{2}
		\big( (q_4, q_2'), 2, 1 \big) \tr{b | b}{3}
		\big( (\hat{q_4}, q_2'), 2, 0 \big) \mv{\lrend}{4}
		\big( (\hat{q_3}, q_2'), 2, 3 \big) \\
		&\mv{d}{5}
		\big( (\hat{q_3}, q_2'), 2, 2 \big) \mv{c}{6}
		\big( (\hat{q_1}, q_2'), \varepsilon, 1 \big) \mv{b}{7}
		\big( (\hat{q_5}, q_2'), 1, 0 \big) \mv{\lrend}{8}
		\big( (\hat{q_4}, q_2'), 1, 3 \big) \\
		&\sw{d}{9}
		\big( (q_4, q_2'), 1, 3 \big) \tr{d | d}{10}
		\big( (\hat{q_4}, q_2'), 1, 2 \big) \sw{c}{11}
		\big( (q_4, q_2'), 1, 2 \big) \tr{c | c}{12}
    \big( (\hat{q_4}, q_2'), 1, 1 \big)
	\end{align*}
	
	Notice how after transition \( 3 \), because \( \trans' \) still requires to move left,
	the computation of \( \trans \) is rewound and goes back to the first copy of \( u \)
	produced by $\trans$.
	Transition \( 6 \) undoes the \( \drop{1} \) operation, hence the \( \varepsilon \)
	in the following configuration, and transition \( 7 \) undoes the \( \lift{1} \),
	effectively moving the pebble one position to the left.
\end{example}

%%%%%%%%%%%%%%%%%%
\subsection{Composition: General case}
\label{sec:general-case}

Now, we prove the general case of composition of pebble transducers.

\begin{theorem}\label{thm:composition-general-case}
  Let $\trans=(Q,\Sigma,\delta,n,q_i,q_f,\Gamma,\mu)$ be a \emph{reversible} $n$-pebble
  transducer (possibly with equality tests) computing a function
  $f\colon\Sigma^{*}\to\Gamma^{*}$.
  Let $\trans'=(Q',\Gamma,\delta',m,q'_i,q'_f,\Delta,\mu')$ be a \emph{deterministic}
  $m$-pebble transducer computing a function $g\colon\Gamma^{*}\to\Delta^{*}$.
  Let $r=(n+1)(m+1)-1$.
  We can construct a \emph{deterministic} $r$-pebble transducer
  $\trans''=(Q'',\Sigma,\delta'',r,q''_i,q''_f,\Delta,\mu'')$ \emph{with equality tests}
  computing the composition $g\circ f\colon\Sigma^{*}\to\Delta^{*}$.
  The number of states of $\trans''$ is at most $\mathcal{O}(|Q|^{m+2}\cdot|Q'|\cdot(n+1)^{m+3})$.
  Moreover, if $\trans'$ is \emph{reversible} then so is $\trans''$.
\end{theorem}

To simplify the construction, we first show that we may restrict to reversible pebble 
transducers that do not move their head when dropping or lifting a pebble.

\begin{lemma}\label{lem:drop-lift-move}
  For each reversible $n$-pebble transducer
  $\trans=(Q,\Sigma,\delta,n,q_i,q_f,\Gamma,\mu)$, we can construct a reversible
  $n$-pebble transducer $\trans'=(Q',\Sigma,\delta',n,q'_i,q'_f,\Gamma,\mu')$ with
  $|Q'|\leq 3|Q|$, computing the same function, and such that every transition in $\trans'$ 
  which moves the head has operation $\nop$. 
\end{lemma}

\begin{proof}
  Let $Q_\op=Q\times \{\drop{},\lift{} \}$. We let 
  $Q'=Q\cup Q_\op$, $q'_{i}=q_{i}$, $q'_{f}=q_{f}$, $Q'_{0}=Q_{0}\cup Q_\op$, 
  $Q'_{+1}=Q_{+1}$ and $Q'_{-1}=Q_{-1}$. We define now $\delta'$ and $\mu'$.
  Each transition $t=(q,a,\varphi,\nop,s)$ of $\trans$ also appears in $\trans'$,
  and for each transition $t=(q,a,\varphi,\op,s)$ of $\trans$ where $\op\neq\nop$, we
  define the following two transitions of $\trans'$:
  \begin{align*}
    t' &= (q,a,\varphi,\op,(q,\op)) \\
    t'' &= ((q,\op),\op(\varphi)\wedge\overline{\reverse{\op}},\nop,s)
  \end{align*}
  with output $\mu'(t')=\mu(t)$ and $\mu'(t'')=\varepsilon$.
 Note that we abuse notation as we forget the index of the action $\op$ in the state $(q,\op)$.
  We claim that $\trans'$ is reversible (see~\cref{claim:drop-lift-move:rev} below).

  Finally, it is easy to see that there is a one-to-one correspondence between the 
  accepting runs of $\trans$ and the accepting runs of $\trans'$, moreover this 
  correspondence preserves the output string produced. Hence, $\sem{\trans}=\sem{\trans'}$.

\begin{claim}\label{claim:drop-lift-move:rev}
	Transducer \( \trans' \) from \cref{lem:drop-lift-move} is reversible.
\end{claim}

\begin{claimproof}
	We call $t_i$ a transition that appears in both $\trans$ and $\trans'$ (i.e. a transition whose action is $\nop$).
	For $i\in\{1,2\}$, let $t'_{i}$ and $t''_{i}$ be the transitions constructed from 
	$t_{i}=(q,a,\varphi_{i},\op_{i},s_{i})$. 
	
	First, let us consider the case where two transitions are enabled at $C=(q,\pebbles,\head)$ for $\trans'$. Then they are either $t_i$ or $t'_i$ transitions.
	If $t'_{i}$ is enabled at $C=(q,\pebbles,\head)$ then $t_{i}$ is enabled at $C$ for $\trans$.
	Hence, in both cases we obtain $t_{1}$ and $t_{2}$ that are enabled at $C$ for $\trans$. By determinism of $\trans$ we get $t_{1}=t_{2}$. If $\op_i=\nop$ then $t_{1}=t_{2}$ in $\trans'$, or $\op_i\neq\op$ and we have $t'_{1}=t'_{2}$.
	
		Second, if two transitions $t'_{i}$ are reverse-enabled at $C'=((q,\op),\pebbles',\head)$
		then there are transitions $t_i=(q,a,\varphi_i,\op_i,s_i)$ in $\delta$ that are enabled at
		$C_i=(q,\reverse{\op_i}(\pebbles',\head),\head)$.
		Remark that given a pebble stack, the only $\drop{}$ action that can be reverse-enabled is
		$\drop{j}$ where $j$ is the size of the pebble stack.  Similarly, the only possible
		$\lift{j}$ action is when $j$ is the size of the pebble stack plus one.
		Then as the transitions $t'_{i}$ are reverse-enabled at $(q,\op)$, by construction of
		$\trans'$ the actions $\op_i$ of the $t_{i}$ have the same index.
		This means that $\reverse{\op_1}(\pebbles',\head)=\reverse{\op_2}(\pebbles',\head)$ and
		that $C_1=C_2$ and by determinism of $\trans$, $t_1=t_2$ which means $t'_1=t'_2$.
	
	Third, assume that $t''_{i}$ is enabled at $C'=((q,\op),\pebbles',\head)$.
	We claim that $t_{i}$ is enabled at $C_i=(q,\reverse{\op_i}(\pebbles',\head),\head)$.
	Indeed, let $\pebbles^{i}=\reverse{\op_{i}}(\pebbles',\head)$, which is well defined since 
	$\pebbles',\head\models\overline{\reverse{\op_{i}}}$. 
	We have $\pebbles'=\op_{i}(\pebbles^{i},\head)$ and using 
	$\pebbles',\head\models\op_{i}(\varphi_{i})$ we get $\pebbles^{i},\head\models\varphi_{i}$.
	Therefore $(q,\pebbles^{i},\head)\xra{t_{i}}(s,\pebbles',\head+s)$. 
	Thanks to the previous remark, the actions $\op_i$ are the same and hence $\pebbles^{i}$
	are equals.
	By determinism of $\trans$, we get that $t_1=t_2$ and hence $t''_1=t''_2$.
	
	For the last case, assume that two transitions are reverse-enabled at 
	$C'=(s,\pebbles',\head')$. 
	If both transitions also appear in $\trans$ then they are equal by reverse-determinism of $\trans$.
	If a transition (or both) is a $t''_{i}$ transition, 
	let $t_{i}=(q_{i},a,\varphi_{i},\op_{i},s)$. 
	We claim that $t_{i}$ is reverse-enabled at $C'$.
	Let $\head=\head'-s$ and
	$\pebbles^{i}=\reverse{\op_{i}}(\pebbles',\head)$, we can easily check that
	$(q,\pebbles^{i},\head)\xra{t_{i}}(s,\pebbles',\head+s)$.
	We get that the transition $t_i$ is reverse-enabled at $C'$ 
	and we conclude by reverse-determinism of $\trans$.
\end{claimproof}

\end{proof}

The rest of the section is devoted to the proof of \cref{thm:composition-general-case}.

As in \cref{sec:simple-case}, we assume that transducer $\trans$ outputs at most one letter
on each transition, and that it writes $\lrend v$ with $v\in\Gamma^{*}$ when running on
$\lrend u$ with $u\in\Sigma^{*}$.  Using \cref{lem:drop-lift-move}, we also assume that
transitions $t'=(q',a',\varphi',\op',s')$ of $\trans'$ do not both drop/lift a pebble and
move the head: if $\op'\neq\nop$ then $s'\in Q'_{0}$.

We fix $u\in\dom{f}\subseteq\Sigma^{*}$ and $v=f(u)\in\Gamma^{*}$.  
We consider the accepting run of $\trans$ on $\lrend u$ producing $\lrend v$:
\begin{equation}
  C_{0}\xra{t_{0}}C_{1}\xra{t_{1}}C_{2}\xra{t_{2}}\cdots C_{N}\xra{t_{N}}C_{N+1} 
  \label{eq:run-T-u}
\end{equation}
where $C_{\ell}=(q_{\ell},\pebbles_{\ell},\head_{\ell})$ 
are configurations and $t_{\ell}\in\delta$ are transitions of $\trans$.

As in \cref{sec:simple-case}, for each position $i\in\epos{v}$, we let $\enc{i}$ be
the index of the transition in the above run producing the letter at position $i$ in
$\lrend v$.
Again, the main idea is to encode a position $i\in\epos{v}$ by the configuration 
$C_{\enc{i}}$. Since $\trans'$ has pebbles, this encoding is used not only for the 
head $\head'$ of $\trans'$, but also for the pebbles dropped by $\trans'$.

Recall that the $\ell$-th configuration in \eqref{eq:run-T-u} is
$C_{\ell}=(q_{\ell},\pebbles_{\ell},\head_{\ell})$.
Consider a configuration $C'=(q',\pebbles',\head')$ of $\trans'$ on $\lrend v$. 
Let $k=|\pebbles'|$, $j_{i}=\enc{\pebbles'_{i}}$ for $1\leq i\leq k$, and $j=\enc{\head'}$. 
Let $[n+1]=\{1,\ldots,n+1\}$.
Define $\enc{C'}=C''=(q'',\pebbles'',\head'')$ where
\begin{align*}
  q'' &= ( q_{j},q',(q_{j_{i}})_{1\leq i\leq k},(1+|\pebbles_{j_{i}}|)_{1\leq i\leq k} )
  \in Q\times Q'\times Q^{k}\times[n+1]^{k} \subseteq Q''_{0} \\
  \pebbles'' &= \pebbles_{j_{1}}\head_{j_{1}} \cdots \pebbles_{j_{k}}\head_{j_{k}} \pebbles_{j} 
  \hspace{15mm}\text{and}\hspace{15mm}
  \head'' = \head_{j} \,.
\end{align*}

The set $Q''_{0}$ of $0$-states of $\trans''$ contains
$\bigcup_{k=0}^{m} Q\times Q'\times Q^{k}\times[n+1]^{k}$.
The initial and final states are 
$q''_{i}=(q_{i},q'_{i},(),())$ and $q''_{f}=(q_{i},q'_{f},(),())$ respectively.
It remains to define the other states in $Q''$, the transition function $\delta''$ and
the output function $\mu''$.

We explain below how a transition $C'\xra{t'}C'_{1}$ of $\trans'$ will be simulated in
$\trans''$ by a sequence of transitions of the form $\enc{C'}\xra{+}\enc{C'_{1}}$.

Consider a state $q''=(q,q',\overline{x},\overline{y})\in Q\times Q'\times 
Q^{k}\times[n+1]^{k}$ for some $0\leq k\leq m$, with 
$\overline{x}=(x_{i})_{1\leq i\leq k}$ and $\overline{y}=(y_{i})_{1\leq i\leq k}$.
Let $t=(q,a,\varphi,\op,s)$ be a transition of $\trans$ which produces an output letter
$a'=\mu(t)\in\Gamma$.  Consider a transition $t'=(q',a',\varphi',\op',s')$ of $\trans'$
which reads the output produced by $t$.  The goal is to synchronize the pair of
transitions $(t,t')$ as we did in \cref{sec:simple-case}.  We first explain how to write a
test $\xi$ checking whether both $t$ and $t'$ are ``enabled'' at a configuration
$C''=(q'',\pebbles'',\head)$ of $\trans''$.  This is intended to be used in particular
when $C''$ is of the form $\enc{C'}$ for some configuration $C'$ of $\trans'$.

We introduce some notation.
For each $0\leq\ell\leq k$, we let $d_\ell=y_{1}+\cdots+y_{\ell}$ which
can be recovered from the state $q''$ of $\trans''$.
Given an offset $d\geq0$ and a test $\varphi$ of $\trans$, we write 
$\varphi^{+d}$ the test obtained by adding $d$ to the pebble indices: 
$\pebble_{i}=\pebble_{j}$ is replaced with $\pebble_{i+d}=\pebble_{j+d}$
and $\head=\pebble_{j}$ is replaced with $\head=\pebble_{j+d}$.
Finally, let $\xi_{0}=(\pebble_{0}=\pebble_{0})\wedge\neg(\pebble_{n+1}=\pebble_{n+1})$
which will be used, shifted by some offset, to check the size of the pebble stack.

\begin{claim}\label{claim:xi}
  Let $C''=(q'',\pebbles'',\head)$ be a configuration of $\trans''$ with
  $q''=(q,q',\overline{x},\overline{y})$.  Let $\psi$ be a test of $\trans$ and $\psi'$ be
  a test of $\trans'$.
  \begin{enumerate}
    \item \label{item:0}  
    $\pebbles'',\head\models\xi_{0}^{+d_{k}}$ if and only if $d_{k}\leq|\pebbles''|\leq d_{k}+n$. 
    
    We assume below that we are in this case and we write $\pebbles''=
    \pebbles^{1}\head^{1} \cdots \pebbles^{k}\head^{k} \pebbles$ with 
    $|\pebbles^{\ell}\head^{\ell}|=y_{\ell}$ for $1\leq\ell\leq k$
    and $0\leq|\pebbles|\leq n$.
  
    \item \label{item:1}
    $\pebbles'',\head\models\psi^{+d_{k}}$ if and only if $\pebbles,\head\models\psi$.
  
    \item \label{item:2}
    We can construct a formula $\overline{\xi}(q'',\psi')$ such that,
    for all $\pebbles',\head'$ with $|\pebbles'|=k$ and
    (a) $\pebbles'_{i}=\pebbles'_{j}$ iff $x_{i}=x_{j}$, $\pebbles^{i}=\pebbles^{j}$ and
    $\head^{i}=\head^{j}$; and
    (b) $\head'=\pebbles'_{j}$ iff $q=x_{j}$, $\pebbles=\pebbles^{j}$ and 
    $\head=\head^{j}$, 
    we have $\pebbles'',\head\models\overline{\xi}(q'',\psi')$ 
    if and only if $\pebbles',\head'\models\psi'$.
  \end{enumerate}
\end{claim}

\begin{claimproof}
  \cref{item:0,item:1} are clear. For \cref{item:2} we let $\overline{\xi}(q'',\psi')$ be the 
  formula $\psi'$ in which we replace each atom of the form $\head'=\pebble'_{i}$ by
  \begin{itemize}
    \item  $\false$ if $i>k$ or $x_{i}\neq q$,
  
    \item  $\Big(\bigwedge_{\ell=1}^{y_{i}-1} \pebble_{\ell+d_{k}}=\pebble_{\ell+d_{i}-1} 
    \Big) \wedge (\head=\pebble_{d_{i}}) \wedge 
    \neg(\pebble_{y_{i}+d_{k}}=\pebble_{y_{i}+d_{k}})$ otherwise;
  \end{itemize}
  and each atom of the form $\pebble'_{i}=\pebble'_{j}$ by
  \begin{itemize}
    \item  $\false$ if $i>k$ or $j>k$ or $x_{i}\neq x_{j}$ or $y_{i}\neq y_{j}$,
  
    \item  $\bigwedge_{\ell=1}^{y_{i}} \pebble_{\ell+d_{i-1}}=\pebble_{\ell+d_{j}-1}$ 
    otherwise.
    \claimqedhere
  \end{itemize}
\end{claimproof}

Let $\xi=(\xi_{0}\wedge\varphi\wedge\overline{\op})^{+d_{k}}
\wedge\overline{\xi}(q'',\varphi'\wedge\overline{\op'})$.
From \cref{claim:xi}, assuming that $C''=\enc{C'}$, we see that 
$C''\models\xi$ if and only if $t$ is enabled at $C=(q,\pebbles,\head)$ and $t'$ is 
enabled at $C'=(q',\pebbles',\head')$.
We explain now how to synchronize the pair of transitions $(t,t')$ from state $q''$.
Given an offset $d\geq0$ and an operation $\op$ of $\trans$, we write $\op^{+d}$ the
operation shifted by $d$: $\drop{i}^{+d}=\drop{i+d}$, $\lift{i}^{+d}=\lift{i+d}$ and
$\nop^{+d}=\nop$.

\noindent\textbf{Case $\nop$:} 
We have $t'=(q',a',\varphi',\nop,s')$ with $s'\in Q'$.
We check whether the pair $(t,t')$ is enabled with $\xi$ and we implement the move 
induced by $s'$ of the head of $\trans'$ on $\# v$ as we did in \cref{sec:simple-case}.
To do so, we introduce the following transition with $\mu''(t'')=\mu'(t')$:
\begin{align}
  t''=(q,q',\overline{x},\overline{y}) & \xra{a,\xi,\op^{+d_{k}}}
  (\hat{s},s',\overline{x},\overline{y}) 
  &&\text{if } s'\in Q'_{+1}
  \tag{\textsf{gtr-a}}\label{eq:general-tr+}
  \\
  t''=(q,q',\overline{x},\overline{y}) & \xra{a,\xi,\nop} 
  (\hat{q},s',\overline{x},\overline{y}) 
  &&\text{if } s'\in Q'_{-1}
  \tag{\textsf{gtr-b}}\label{eq:general-tr-}
  \\
  t''=(q,q',\overline{x},\overline{y}) & \xra{a,\xi,\nop} 
  (q,s',\overline{x},\overline{y})
  &&\text{if } s'\in Q'_{0}
  \tag{\textsf{gtr-c}}\label{eq:general-tr0}
\end{align}
Here, $\widehat{Q}=\{\hat{q}\mid q\in Q\}$ is a disjoint copy of $Q$.
The polarity of a simulation state $(\hat{q},q',\overline{x},\overline{y})
\in\widehat{Q}\times(Q'_{-1}\cup Q'_{+1})\times Q^{k}\times[n+1]^{k}$ with $0\leq k\leq m$
is the product of the polarity of $q$ and the polarity of $q'$.
When in a simulation state $(\hat{q},q',\overline{x},\overline{y})$, we simulate 
$\trans$ forward or backward depending on the polarity of $q'$ using transitions 
producing $\varepsilon$ until we reach a transition producing a symbol from $\Gamma$.
For each transition $t=(q,a,\varphi,\op,s)$ of $\trans$ with $\mu(t)=\varepsilon$ we
introduce the following transition $t''$ in $\trans''$ with $\mu''(t'')=\varepsilon$:
\begin{align}
  t''=(\hat{q},q',\overline{x},\overline{y}) & \xra{a,\varphi^{+d_{k}},\op^{+d_{k}}} 
  (\hat{s},q',\overline{x},\overline{y}) &&\text{if } q'\in Q'_{+1}
  \tag{\textsf{gmv-a}}\label{eq:general-mv+}
  \\
  t''=(\hat{s},q',\overline{x},\overline{y}) & \xra{a,(\reverse{(\varphi,\op)})^{+d_{k}}} 
  (\hat{q},q',\overline{x},\overline{y}) &&\text{if } q'\in Q'_{-1}
  \tag{\textsf{gmv-b}}\label{eq:general-mv-}
\end{align}
and for each transition $t=(q,a,\varphi,\op,s)$ of $\trans$ with $\mu(t)\in\Gamma$ we
introduce the following transition $t''$ in $\trans''$ with $\mu''(t'')=\varepsilon$ in 
order to switch back from the simulation mode to the synchronization mode:
\begin{align}
  t''=(\hat{q},q',\overline{x},\overline{y}) & \xra{a,(\varphi\wedge\overline{\op})^{+d_{k}},\nop} 
  (q,q',\overline{x},\overline{y}) 
  &&\text{if } q'\in Q'_{+1}
  \tag{\textsf{gsw-a}}\label{eq:general-sw+}
  \\
  t''=(\hat{s},q',\overline{x},\overline{y}) & \xra{a,(\reverse{(\varphi,\op)})^{+d_{k}}}
  (q,q',\overline{x},\overline{y}) 
  &&\text{if } q'\in Q'_{-1}
  \tag{\textsf{gsw-b}}\label{eq:general-sw-}
\end{align}
As in the proof of \cref{thm:composition-simple-case}, we can show that $\trans''$ is 
reversible at simulation states of the form $(\hat{q},q',\overline{x},\overline{y})$.
We can also prove similarly the following analog of \cref{claim:simple-simulate}.

\begin{claim}\label{claim:genearl-simulate}
  Assuming that the operation of $t'$ is $\nop$, there is a transition of $\trans'$
  $C'=(q',\pebbles',\head')\xra{t'}C'_{1}=(s',\pebbles',\head'+s')$ on
  $\lrend v$ if and only if there is a nonempty run $\enc{C'}\xra{+}\enc{C'_{1}}$ of
  $\trans''$ on $\lrend u$ where intermediate states are simulation states.
\end{claim}

\noindent\textbf{Case $\lift{}$:} 
We have $t'=(q',a',\varphi',\lift{k'},s')$ with $1\leq k'\leq m$ and $s'\in Q'_{0}$ from
our assumption on $\trans'$ (\cref{lem:drop-lift-move}).
We still check that the pair $(t,t')$ is enabled with 
$\xi=(\xi_{0}\wedge\varphi\wedge\overline{\op})^{+d_{k}}
\wedge\overline{\xi}(q'',\varphi'\wedge\overline{\lift{k'}})$.
Thanks to $\overline{\xi}(q'',\overline{\lift{k'}})$, we get $k'=k$, $q=x_{k}$,
$\pebbles^{k}=\pebbles$ and $\head^{k}=\head$.
We use the lift-gadget given in \cref{fig:composition-lift} to pop the top $y_{k}$
pebbles, i.e., $\head^{k}\pebbles$ (in reverse order).
We enter and leave the lift-gadget with the following transitions where
$s''=(q,s',(x_{1},\ldots,x_{k-1}),(y_{1},\ldots,y_{k-1}))\in Q''_{0}$,
$\xi'=(\xi_{0}\wedge\varphi\wedge\overline{\op})^{+(d_{k}-y_{k})}
\wedge\overline{\xi}(s'',\op'(\varphi')\wedge\overline{\reverse{\op'}})$,
$\mu''(t''_{\mathsf{first}})=\mu'(t')$ and $\mu''(t''_{\mathsf{last}})=\varepsilon$:
\begin{align}
  t''_{\mathsf{first}} &= q'' \xra{a,\xi,\nop} (q'',-1)
  \tag{\textsf{lift-a}}\label{eq:general-lift-first}
  \\
  t''_{\mathsf{last}} &= (q'',0) \xra{a,\xi',\nop} s'' 
  \tag{\textsf{lift-b}}\label{eq:general-lift-last}
\end{align}

\begin{figure}[tbh]
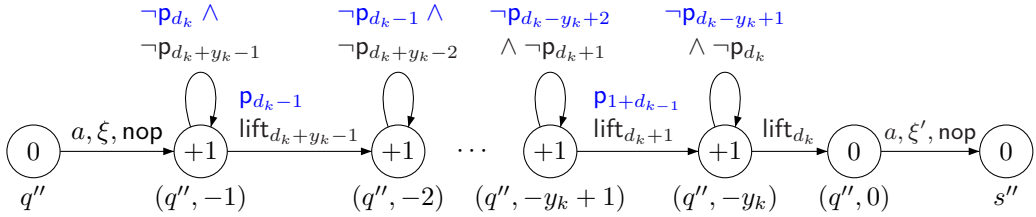

  \centering
  \gusepicture{composition-lift}
  \caption{Simulation of $\lift{k}$. 
  The first state is
  $q''=(q,q',(x_{1},\ldots,x_{k}),(y_{1},\ldots,y_{k}))\in Q''_{0}$
  and $d_{k}=y_{1}+\cdots+y_{k}$.
  The internal states are of the from $(q'',-\ell)\in Q''_{+1}$ for $0\leq\ell\leq y_{k}$.
  The last state is $s''=(q,s',(x_{1},\ldots,x_{k-1}),(y_{1},\ldots,y_{k-1}))\in Q''_{0}$.
  The test on the first transition is $\xi=(\xi_{0}\wedge\varphi\wedge\overline{\op})^{+d_{k}}
  \wedge\overline{\xi}(q'',\varphi'\wedge\overline{\op'})$.
  The test on the last transition is 
  $\xi'=(\xi_{0}\wedge\varphi\wedge\overline{\op})^{+(d_{k}-y_{k})}
  \wedge\overline{\xi}(s'',\op'(\varphi')\wedge\overline{\reverse{\op'}})$.
  }
  \label{fig:composition-lift}
\end{figure}

The blue tests are added to ensure reverse determinism at the internal states.
Also, when the before last transition is taken we see $\pebble_{d_{k}}$.  Hence, the head
is on the same position at the beginning and at the end of the gadget.

\noindent\textbf{Case $\drop{}$:} The operation of transition $t'$ is 
$\op'=\drop{k'}$ ($1\leq k'\leq m$).
We still check that the pair $(t,t')$ is enabled with 
$\xi=(\xi_{0}\wedge\varphi\wedge\overline{\op})^{+d_{k}}
\wedge\overline{\xi}(q'',\varphi'\wedge\overline{\drop{k'}})$.
Thanks to $\overline{\xi}(q'',\overline{\drop{k'}})$, we get $k'=k+1$.
We use the drop-gadget given in \cref{fig:composition-drop} to push
$\head$ and then $\pebbles$ (in this order) on the pebble stack. 
We enter and leave the drop-gadget with the following transitions where
$1\leq z\leq n+1$,
$s''=(q,s',(x_{1},\ldots,x_{k},q),(y_{1},\ldots,y_{k},z))\in Q''_{0}$,
$\xi'=(\xi_{0}\wedge\varphi\wedge\overline{\op})^{+(d_{k}+z)}
\wedge\overline{\xi}(s'',\op'(\varphi')\wedge\overline{\reverse{\op'}})$,
$\mu''(t''_{\mathsf{first}})=\mu'(t')$ and $\mu''(t''_{\mathsf{last}})=\varepsilon$:
\begin{align}
  t''_{\mathsf{first}} &= q'' \xra{a,\xi,\drop{d_{k}+z}} (q'',z,1)
  \tag{\textsf{drop-a}}\label{eq:general-drop-first}
  \\
  t''_{\mathsf{last}} &= (q'',z,z) \xra{a,\xi',\nop} s'' 
  \tag{\textsf{drop-b}}\label{eq:general-drop-last}
\end{align}

\begin{figure}[tbh]
  \centering
  \gusepicture{composition-drop}
  \caption{Simulation of $\drop{k+1}$ with $z$ any number such that $1\leq z\leq n+1$. \\
  The first state is $q''=(q,q',(x_{1},\ldots,x_{k}),(y_{1},\ldots,y_{k}))\in Q''_{0}$ 
  and $d_{k}=y_{1}+\cdots+y_{k}$. 
  The internal states are of the from $(q'',z,\ell)\in Q''_{+1}$ for $1\leq\ell\leq z$.
  The last state is $s''=(q,s',(x_{1},\ldots,x_{k},q),(y_{1},\ldots,y_{k},z))\in Q''_{0}$.
  The test on the first transition is $\xi=(\xi_{0}\wedge\varphi\wedge\overline{\op})^{+d_{k}}
  \wedge\overline{\xi}(q'',\varphi'\wedge\overline{\op'})$.
  The test on the last transition is 
  $\xi'=(\xi_{0}\wedge\varphi\wedge\overline{\op})^{+(d_{k}+z)}
  \wedge\overline{\xi}(s'',\op'(\varphi')\wedge\overline{\reverse{\op'}})$.
  }
  \label{fig:composition-drop}
\end{figure}

On the first transition of the gadget, the test $\xi$ makes sure that the pebble stack is 
of the form $\pebbles''=\pebbles^{1}\head^{1} \cdots \pebbles^{k}\head^{k} \pebbles$ with 
$|\pebbles^{\ell}\head^{\ell}|=y_{\ell}$ for $1\leq\ell\leq k$ and $0\leq|\pebbles|\leq n$.
The operation $\drop{d_{k}+z}$ on the first transition is only possible if 
$z=|\pebbles|+1$ and it allows to determine the size of $\pebbles$.
Since $\reverse{\op'}=\lift{k+1}$, the test $\xi'$ on the last transition contains 
$\overline{\xi}(s'',\overline{\lift{k+1}})$ which contains the test $\pebble_{d_{k}+z}$.
This test makes sure that at the end of the gadget, the head of $\trans''$ is on the same
position of the input word $\#u$ as it was at the beginning.
The drop gadget is deterministic at the internal states. The blue tests are added
to make this gadget reverse deterministic at the internal states.

We conclude by proving that the construction is preserves determinism and reverse-determinism.

\begin{claim}\label{claim:general-det}
	The transducer $\trans''$ from \cref{thm:composition-general-case} is deterministic.
\end{claim}

\begin{claimproof}
	Consider two transitions $t''_{1}$ and $t''_{2}$ of $\trans''$ reading some letter
	$a\in\Sigma$ and which are both enabled at some configuration $C''$ of $\trans''$.
	
\noindent\textbf{Case 1:} $C''=(q'',\pebbles'',\head)$ with
	$q''=(q,q',\overline{x},\overline{y})\in Q\times Q'\times
	Q^{k}\times[n+1]^{k}\subseteq Q''_{0}$ and $0\leq k\leq m$.
	For $i\in\{1,2\}$, consider the transitions $t_{i}=(q,a,\varphi_{i},\op_{i},s_{i})$ and 
	$t'_{i}=(q',a'_{i},\varphi'_{i},\op'_{i},s'_{i})$ used to construct $t''_{i}$
	using one of (\cref{eq:general-tr+,eq:general-tr-,eq:general-tr0,eq:general-lift-first,eq:general-drop-first}).
	
	Let $\xi_{i}=(\xi_{0}\wedge\varphi_{i}\wedge\overline{\op_{i}})^{+d_{k}}
	\wedge\overline{\xi}(q'',\varphi'_{i}\wedge\overline{\op'_{i}})$ be the test of 
	$t''_{i}$. Since $\pebbles'',\head\models\xi_{0}^{+d_{k}}$ we may write
	$\pebbles''=\pebbles^{1}\head^{1} \cdots \pebbles^{k}\head^{k} \pebbles$ with 
	$|\pebbles^{\ell}\head^{\ell}|=y_{\ell}$ for $1\leq\ell\leq k$ and $0\leq|\pebbles|\leq n$.
	By \cref{claim:xi}, we get first 
	$\pebbles,\head\models\varphi_{i}\wedge\overline{\op_{i}}$ and we deduce that $t_{i}$ 
	is enabled at $C=(q,\pebbles,\head)$. It follows that $t_{1}=t_{2}$ by determinism of 
	$\trans$. Hence, $t'_{1}$ and $t'_{2}$ both read the same letter $a'=\mu(t_{1})$ and 
	start from the same state $q'$. Let $\pebbles',\head'$ be constructed as in 
	\cref{item:2} of \cref{claim:xi}. We get 
	$\pebbles',\head'\models\varphi'_{i}\wedge\overline{\op'_{i}}$ and we deduce that 
	$t'_{i}$ is enabled at $C'=(q',\pebbles',\head')$.  It follows that $t'_{1}=t'_{2}$ by
	determinism of $\trans'$. We conclude that $t''_{1}=t''_{2}$.
	
\noindent\textbf{Case 2:} $C''=(q'',\pebbles'',\head)$ with
	$q''\in \widehat{Q}\times Q'\times Q^{k}\times[n+1]^{k}$ and $0\leq k\leq m$.
	
	The proof is similar to the proof for the corresponding case in \cref{claim:simple-det}.
	
\noindent\textbf{Case 3:} Internal state of the lift-gadget:
	$C''=((q'',-\ell),\pebbles'',\head)$ with
	$q''=(q,q',\overline{x},\overline{y})\in Q\times Q'\times 
	Q^{k}\times[n+1]^{k}$, $1\leq k\leq m$ and $0\leq\ell\leq y_{k}$.
	
	If $\ell\neq 0$ there are only two transitions starting from state $(q'',-\ell)$, 
	the self-loop with test $\neg\pebble_{d_{k}+y_{k}-\ell}$ and operation $\nop$ and the
	transition with operation $\lift{d_{k}+y_{k}-\ell}$. These two cannot be simultaneously 
	enabled.
	
	Assume that $\ell=0$. 
	For $i\in\{1,2\}$, consider the transitions $t_{i}=(q,a,\varphi_{i},\op_{i},s_{i})$ and 
	$t'_{i}=(q',a'_{i},\varphi'_{i},\op'_{i},s'_{i})$ used to construct $t''_{i}$
	using \eqref{eq:general-lift-last}.
	Let $\xi'_{i}=(\xi_{0}\wedge\varphi_{i}\wedge\overline{\op_{i}})^{+(d_{k}-y_{k})}
	\wedge\overline{\xi}(s'',\op'_{i}(\varphi'_{i})\wedge\overline{\reverse{\op_{i}'}})$ be
	the test of $t''_{i}$.  
	Since $\pebbles'',\head\models\xi_{0}^{+(d_{k}-y_{k})}$ we may write
	$\pebbles''=\pebbles^{1}\head^{1} \cdots \pebbles^{k-1}\head^{k-1} \pebbles$ with 
	$|\pebbles^{\ell}\head^{\ell}|=y_{\ell}$ for $1\leq\ell<k$ and $0\leq|\pebbles|\leq n$.
	By \cref{claim:xi}, we get first 
	$\pebbles,\head\models\varphi_{i}\wedge\overline{\op_{i}}$ and we deduce that $t_{i}$ 
	is enabled at $C=(q,\pebbles,\head)$. It follows that $t_{1}=t_{2}$ by determinism of 
	$\trans$. Hence, $t'_{1}$ and $t'_{2}$ both read the same letter $a'=\mu(t_{1})$ and 
	start from the same state $q'$. Let $\pebbles',\head'$ be constructed from 
	$\pebbles'',\head$ as in \cref{item:2} of \cref{claim:xi}. 
	We get 
	$\pebbles',\head'\models\op'_{i}(\varphi'_{i})\wedge\overline{\reverse{\op_{i}'}}$.
	We deduce that $t'_{i}$ is enabled at $C'=(q',\reverse{\op_{i}'}(\pebbles',\head'),\head')$.  
	It follows that $t'_{1}=t'_{2}$ by determinism of $\trans'$. 
	We conclude that $t''_{1}=t''_{2}$.
	
\noindent\textbf{Case 4:} Internal state of the drop-gadget:\\
	$C''=((q'',z,\ell),\pebbles'',\head)$ with
	$q''=(q,q',\overline{x},\overline{y})\in Q\times Q'\times 
	Q^{k}\times[n+1]^{k}$, $0\leq k<m$, $1\leq z\leq n+1$ and $1\leq\ell\leq z$.
	
	If $\ell\neq z$ there are only two transitions starting from state $(q'',z,\ell)$, 
	the self-loop with test $\neg\pebble_{d_{k}+\ell}$ and the transition with test
	$\pebble_{d_{k}+\ell}$.  These two cannot be simultaneously enabled.
	
	Assume that $\ell=z$. 
	As before, the self-loop with test $\neg\pebble_{d_{k}+z}$ cannot be simultaneously 
	enabled with a transition $((q'',z,z),a,\xi',\nop,s'')$ since the test $\xi'$ 
	checks $\pebble_{d_{k}+z}$ inside $\overline{\xi}(s'',\overline{\lift{k+1}})$.
	
	Assume now that $t''_{i}=((q'',z,z),a,\xi'_{i},\nop,s''_{i})$ for $i\in\{1,2\}$.
	Let $t_{i}=(q,a,\varphi_{i},\op_{i},s_{i})$ and
	$t'_{i}=(q',a'_{i},\varphi'_{i},\drop{k+1},s'_{i})$ 
	be the transitions used to construct $t''_{i}$ using \eqref{eq:general-drop-last}.
	
	The test on $t''_{i}$ is
	$\xi'_{i}=(\xi_{0}\wedge\varphi_{i}\wedge\overline{\op_{i}})^{+d_{k}+z}
	\wedge\overline{\xi}(s''_{i},\drop{k+1}(\varphi'_{i})\wedge\overline{\lift{k+1}})$.
	Since $\pebbles'',\head\models\xi_{0}^{+d_{k}+z}$ we may write
	$\pebbles''=\pebbles^{1}\head^{1} \cdots \pebbles^{k+1}\head^{k+1} \pebbles$ with 
	$|\pebbles^{\ell}\head^{\ell}|=y_{\ell}$ for $1\leq\ell\leq k$,
	$|\pebbles^{k+1}\head^{k+1}|=z$ and $0\leq|\pebbles|\leq n$.
	Since $\pebbles'',\head\models(\varphi_{i}\wedge\overline{\op_{i}})^{+d_{k}+z}$,
	we get $\pebbles,\head\models\varphi_{i}\wedge\overline{\op_{i}}$ by \cref{claim:xi}. 
	We deduce that $t_{1}$ and $t_{2}$ are both enabled at $C=(q,\pebbles,\head)$. 
	It follows that $t_{1}=t_{2}$ by determinism of $\trans$. 
	Hence, $t'_{1}$ and $t'_{2}$ both read the same letter $a'=\mu(t_{1})$ and 
	start from the same state $q'$.
	Let $\pebbles',\head'$ be constructed from $\pebbles'',\head$ as in 
	\cref{item:2} of \cref{claim:xi}.  Using $\pebbles'',\head\models
	\overline{\xi}(s''_{i},\drop{k+1}(\varphi'_{i})\wedge\overline{\lift{k+1}})$, we get
	$\pebbles',\head'\models\drop{k+1}(\varphi'_{i})\wedge\overline{\lift{k+1}}$.
	We deduce that $t'_{i}$ is enabled at $C'=(q',\lift{k+1}(\pebbles',\head'),\head')$.
	It follows that $t'_{1}=t'_{2}$ by determinism of $\trans'$.  We conclude that
	$t''_{1}=t''_{2}$.
\end{claimproof}

\begin{claim}\label{claim:general-rev-det}
	If $\trans'$ from \cref{thm:composition-general-case} is reverse-deterministic then $\trans''$ is reverse-deterministic.
\end{claim}

\begin{claimproof}
	Consider two transitions $t''_{1}$ and $t''_{2}$ of $\trans''$ reading some letter
	$a\in\Sigma$ and which are both reverse-enabled at some configuration $C''$ of $\trans''$.
	
\noindent\textbf{Case 1:} $C''=(s'',\pebbles'',\head)$ with
	$s''=(q,s',\overline{x},\overline{y})\in Q\times Q'\times
	Q^{k}\times[n+1]^{k}\subseteq Q''_{0}$ and $0\leq k\leq m$.
	If $s'\in Q'_{+1}\cup Q'_{-1}$ then both transitions $t''_{1}$ and $t''_{2}$ are 
	constructed from the same case \eqref{eq:general-sw+} or \eqref{eq:general-sw-} 
	respectively. We prove that $t''_{1}=t''_{2}$ as in the corresponding case of 
	\cref{claim:simple-rev-det}. We assume bolow that $s'\in Q'_{0}$.
	
	For $i\in\{1,2\}$, consider the transitions $t_{i}=(q,a,\varphi_{i},\op_{i},s_{i})$ and 
	$t'_{i}=(q'_{i},a'_{i},\varphi'_{i},\op'_{i},s')$ used to construct $t''_{i}$ using one 
	of (\cref{eq:general-tr0,eq:general-lift-last,eq:general-drop-last}). 
	Let $\xi'_{i}=(\xi_{0}\wedge\varphi_{i}\wedge\overline{\op_{i}})^{+d_{k}}
	\wedge\overline{\xi}(s'',\op'_{i}(\varphi'_{i})\wedge\overline{\reverse{\op_{i}'}})$ be
	the test of $t''_{i}$.  Notice that, if $t''_{i}$ is constructed from
	\eqref{eq:general-tr0} then $\op'_{i}=\nop$ and
	$\xi'_{i}=(\xi_{0}\wedge\varphi_{i}\wedge\overline{\op_{i}})^{+d_{k}}
	\wedge\overline{\xi}(s'',\varphi'_{i})=\xi_{i}$ since $\overline{\xi}(s'',\cdot)$ does
	not depend on the $Q'$ component of $s''$.
	Transition $t''_{i}$ is reverse-enabled at $C''$, its operation is $\nop$ and it does 
	not move the head, hence $\pebbles'',\head\models\xi'_{i}$. Therefore, 
	we may write
	$\pebbles''=\pebbles^{1}\head^{1} \cdots \pebbles^{k}\head^{k} \pebbles$ with 
	$|\pebbles^{\ell}\head^{\ell}|=y_{\ell}$ for $1\leq\ell\leq k$ and $0\leq|\pebbles|\leq n$.
	By \cref{claim:xi}, we get first 
	$\pebbles,\head\models\varphi_{i}\wedge\overline{\op_{i}}$ and we deduce that $t_{i}$ 
	is enabled at $C=(q,\pebbles,\head)$. It follows that $t_{1}=t_{2}$ by determinism of 
	$\trans$. Hence, $t'_{1}$ and $t'_{2}$ both read the same letter $a'=\mu(t_{1})$ and 
	end at the same state $s'$.  Let $\pebbles',\head'$ be constructed from
	$\pebbles'',\head$ as in \cref{item:2} of \cref{claim:xi}.  We get
	$\pebbles',\head'\models\op'_{i}(\varphi'_{i})\wedge\overline{\reverse{\op_{i}'}}$ and
	we deduce that $t'_{i}$ is reverse-enabled at $C'=(s',\pebbles',\head')$.  It follows that
	$t'_{1}=t'_{2}$ by reverse-determinism of $\trans'$ and then $t''_{1}=t''_{2}$.
	
\noindent\textbf{Case 2:} $C''=(q'',\pebbles'',\head)$ with
	$q''\in \widehat{Q}\times Q'\times Q^{k}\times[n+1]^{k}$ and $0\leq k\leq m$.
	
	The proof is similar to the proof for the corresponding case in \cref{claim:simple-rev-det}.
	
\noindent\textbf{Case 3:} Internal state of the lift-gadget:
	$C''=((q'',-\ell),\pebbles'',\head)$ with
	$q''=(q,q',\overline{x},\overline{y})\in Q\times Q'\times 
	Q^{k}\times[n+1]^{k}$, $1\leq k\leq m$ and $0\leq\ell\leq y_{k}$.
	
	If $1<\ell<y_{k}$ there are only two transitions ending at state $(q'',-\ell)$, 
	the self-loop with test $\neg\pebble_{d_{k}-\ell+1}$ and the transition with test
	$\pebble_{d_{k}-\ell+1}$.  These two cannot be simultaneously reverse-enabled.
	
	Assume that $\ell=1$. 
	As before, the self-loop cannot be simultaneously 
	reverse-enabled with a transition $(q'',a,\xi,\nop,(q'',-1))$ since the test $\xi$ 
	checks $\pebble_{d_{k}}$ inside $\overline{\xi}(q'',\overline{\lift{k}})$, and this is 
	incompatible with the test $\neg\pebble_{d_{k}}$ of the self-loop.
	
	Assume now that $t''_{i}=(q'',a,\xi_{i},\nop,(q'',-1))$ for $i\in\{1,2\}$.
	Since $t''_{i}$ is reverse-enabled at $C''$, we see that $t''_{i}$ is enabled at
	$(q'',\pebbles'',\head-1)$. By \cref{claim:general-det} we deduce that $t''_{1}=t''_{2}$.
	
\noindent\textbf{Case 4:} Internal state of the drop-gadget:\\
	$C''=((q'',z,\ell),\pebbles'',\head)$ with
	$q''=(q,q',\overline{x},\overline{y})\in Q\times Q'\times 
	Q^{k}\times[n+1]^{k}$, $0\leq k<m$, $1\leq z\leq n+1$ and $1\leq\ell\leq z$.
	
	If $\ell\neq 1$ there are only two transitions ending at state $(q'',z,\ell)$, 
	the self-loop with test $\neg\pebble_{d_{k}+z+\ell-1}$ and operation $\nop$, and the
	transition with operation $\drop{d_{k}+z+\ell-1}$.  These two cannot be simultaneously
	enabled.
	
	Assume that $\ell=1$. As before, the self-loop with test 
	$\neg\pebble_{d_{k}+z}$ and operation $\nop$ cannot be simultaneously 
	reverse-enabled with a transition having operation $\drop{d_{k}+z}$.
	
	Assume now that $t''_{i}=(q'',a,\xi_{i},\drop{d_{k}+z},(q'',z,1))$ for $i\in\{1,2\}$.
	Since $t''_{i}$ is reverse-enabled at $C''$, we see that $t''_{i}$ is enabled at
	$(q'',\lift{d_{k}+z}(\pebbles'',\head-1),\head-1)$. 
	By \cref{claim:general-det} we deduce that $t''_{1}=t''_{2}$.
\end{claimproof}

\section{Generators for  Pebble Transducers}
\subparagraph{Set of generators.}
In~\cite{bojanpebble}, the class of polyregular functions is defined as the smallest class of functions closed under composition that contains the sequential functions, the squaring function and the iterated reverse function. The iterated reverse function acts on an alphabet $\Sigma$ enriched with a special symbol
$!$ and maps a word $u_0!  u_1!  \ldots !  u_n$ to $u_0^r !  u_1^r !  \ldots !
u_n^r$.
By proving that each of these generators can easily be realized by a reversible pebble transducer and using \cref{thm:composition-general-case}, 
we prove that reversible pebble transducers realize all polyregular functions. This also gives a way to generate them using these basic blocks and composition.

Using~\cite[Theorem~2]{DFJL17} as well as \cref{rem:0peb2w}, we get that any sequential function, realized by a transducer $T$ with $n$ states, can be realized by a reversible $0$-pebble transducer with $\mathcal{O}(n^2)$ states.
A reversible $1$-pebble transducer for the squaring function was given in \cref{ex:squaring}.
Finally, we give a reversible $0$-pebble transducer for the iterated reverse in \cref{fig:iterRev}. 

\begin{figure}[t]
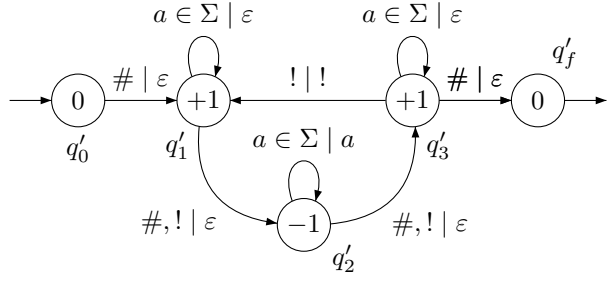

  \centering
  \gusepicture{iterRev}
  \caption{A $0$-pebble transducer realizing the iterated reverse function.  As there is
  no pebble, we omitted the tests and the pebble actions.}
  \label{fig:iterRev}
\end{figure}

\subparagraph{Uniformizing Pebble transducers.}

Another way to generate reversible pebble transducers is to start from a possibly non
deterministic pebble transducer and uniformize it by a reversible pebble transducer.  This
section provides a procedure to do this, while preserving the number of pebbles used by
the given transducer.
By uniformizing a relation $R$,
we mean extract a function $f$ such that $\dom{f}=\dom{R}$ and $f\subseteq R$.

\begin{theorem}\label{thm-uniform}
  Given a $k$-pebble transducer $\trans$ with $n$ states (possibly with equality tests),
  one can construct a $k$-pebble reversible transducer $R$ with $2^{\mathcal{O}((kn)^2)}$
  states such that $\sem{R}$ is a uniformization of $\sem{\trans}$.
\end{theorem}

\begin{proof}
  The proof relies on the composition of reversible pebble transducers, the uniformization
  result from~\cite{DFJL17} and \cref{lem-configsKpeb,lem-enrichmentEquality,lem-decompose} proved below.
  $\trans$ can be decomposed into a basic reversible $k$-pebble transducer $\C_k$, a 
  reversible $0$-pebble transducer $\C_{k}^{=}$ enabling equality tests, and a $0$-pebble 
  transducer $\trans_0$ with respectively $\mathcal{O}(k)$, $\mathcal{O}(2^{k^2})$ and $\mathcal{O}(k n)$ states.
  The transducer $\C_k$ associates to a word $u$ the word of length $|\#u|^{|\#u|^k}$ that is
  the sequence of possible configurations of a $k$ pebble transducer over the word $u$.
  The transducer $\C_{k}^{=}$ adds the truth values of equality tests to the configurations.
  The $0$-pebble transducer $\trans_0$ uses this information to simulate $\trans$.
  Then, by \cref{rem:0peb2w}, $\trans_0$ is transformed into a two-way transducer
  $\trans'_0$ with $\mathcal{O}(kn)$ states.
  We can then use~\cite[Theorem~4]{DFJL17} to obtain a reversible two-way transducer $RT$
  with $2^{\mathcal{O}((kn)^2)}$ states.
  Using \cref{rem:0peb2w} back, $RT$ is transformed into a $0$-pebble transducer $RP$ with $2^{\mathcal{O}((kn)^2)}$ states.
  We can conclude by composing $\C_k$ with $\C_{k}^{=}$ and composing the result with $RP$ using \cref{thm:composition-simple-case} twice.  
  We obtain a reversible $k$-pebble transducer $R$ with $2^{\mathcal{O}((kn)^2)}$ states.
\end{proof}

Given an integer $k\geq 1$ and a word $u$, we define the marking of $u$ for
$k$-configurations as the word $\C_{k}(u)$ on $(\Sigma\cup\{\#\})\times\{0,1\}^k$ which is
the lexicographically ordered sequence of every possible marking of $k$ positions
in $\#u$. For instance, we have $\C_{1}(ab)=
\setlength{\arraycolsep}{1pt}\begin{array}{ccccccccc} 
  \# & a & b & \# & a & b & \# & a & b \\ 1 & 0 & 0 & 0 & 1 & 0 & 0 & 0 & 1 
\end{array}$ and
\[
\C_{2}(ab)= \setlength{\arraycolsep}{1pt}\begin{array}{ccccccccccccccccccccccccccc} 
  \# & a & b & \# & a & b & \# & a & b &
  \# & a & b & \# & a & b & \# & a & b &
  \# & a & b & \# & a & b & \# & a & b 
  \\ 
  1 & 0 & 0 & 1 & 0 & 0 & 1 & 0 & 0 &
  0 & 1 & 0 & 0 & 1 & 0 & 0 & 1 & 0 &
  0 & 0 & 1 & 0 & 0 & 1 & 0 & 0 & 1 
  \\ 
  1 & 0 & 0 & 0 & 1 & 0 & 0 & 0 & 1 &
  1 & 0 & 0 & 0 & 1 & 0 & 0 & 0 & 1 &
  1 & 0 & 0 & 0 & 1 & 0 & 0 & 0 & 1 
\end{array} \,.
\]

\begin{lemma}\label{lem-configsKpeb}
  Given an integer $k\geq 1$ and an alphabet $\Sigma$, one can construct a reversible
  $k$-pebble transducer $\C_k$ with $\mathcal{O}(k)$ states such that for all input word
  $u\in\Sigma^*$, $\C_k(u)$ is the marking of $u$ for $k$-configurations.
\end{lemma}
\begin{proof}
  The machine for $\C_1$ is given in \cref{fig:allconfigs} (left).  It is a slightly
  modified version of the squaring function of \cref{fig:rev-squaring-v4}, mainly to
  accomodate the fact that $\C_1$ also produces the endmarker $\#$ at each iteration.
  
  For $k>1$, the machine $\C_k$ is defined using an enhanced $\C_{k-1}$ as a blackbox.
  Indeed, the marking for $k$-configurations over a word $u$ can be defined as the $k-1$
  marking iterated $|\#u|$ times, where each iteration is marked by the new pebble on a
  different position of $\#u$, in order.
  Then $\C_k$ uses a modified $\C^{+1}_{k-1}$ that outputs its marking plus the one of the
  new pebble.  It is described in \cref{fig:allconfigs} (right).
\end{proof}

\begin{figure}[t]
  \centering
  \gusepicture[scale=0.88]{allconfigs-1}
  \hfill
  \gusepicture[scale=0.88]{allconfigs-k}
  
  \caption{Reversible transducers for the function $\C_1$ (left) and $\C_k$ (right).  The
  production $b$ of $\C_1$ is the bit corresponding to $\head=\pebble_1$.  The transitions
  are actually duplicated.  The self-loop on $q_{4}$ reading $a\in\Sigma$ stands for two
  transitions $(q_4,a,\pebble_1,q_4)$ producing $(a,1)$ and $(q_4,a,\lnot \pebble_1,q_4)$
  producing $(a,0)$.  Similarly, the transition from $q_{3}$ to $q_{4}$ reading $\#$ is 
  duplicated. \\
  In $\C_{k}$, we use a copy $\C^{+1}_{k-1}$ of $\C_{k-1}$ where all the pebble indices
  are incremented by $1$.  The transition in $\C_{k}$ from $q_{3}$ labelled $\#$ goes to
  the initial state of $\C^{+1}_{k-1}$ which is a $0$-state, hence the head does not move.
  Similarly, the transition in $\C_{k}$ labelled $\#$ going to $q_{5}$ starts from the
  final state of $\C^{+1}_{k-1}$.  When it appears in $\C_k$, $b$ stands for the bit
  vector giving the truth value of each $\head=\pebble_{i}$ for $1\leq i\leq k$.  It is
  actually a macro for $2^k$ disjoint transitions.}
  \label{fig:allconfigs}
\end{figure}

\begin{lemma}\label{lem-enrichmentEquality}
  Given an integer $k\geq 1$ and an alphabet $\Sigma$, 
  one can construct a reversible $0$-pebble transducer $\C_{k}^{=}$ with 
  $\mathcal{O}(2^{k^{2}})$
  states that reads $\sem{\C_k}(u)$ and adds to each copy of $\#u$ in $\sem{\C_k}(u)$ a
  $k\times k$ boolean matrix $M$ such that $M_{i,j}=1$ if and only if $i$ and $j$ mark the
  same position in this copy of $\#u$.
\end{lemma}

\begin{proof}
  Informally, the transducer $\C_{k}^{=}$ reads a marking of $\#u$ and computes the matrix
  $M$ along this copy.  Then, when reaching a new pair $(\#,b)$, it goes back to output
  the copy enriched with the matrix $M$.  Special care has to be taken to ensure
  reverse-determinism.  In particular, the transducer $\C_{k}^{=}$ needs to undo the
  computation done before moving again to the next marked copy of $\#u$.

Let $C_k$ be defined as in \cref{lem-configsKpeb}.

We set $\C_{k}^{=} = (P,\Sigma\times\{0,1\}^k,\alpha,0,p_i,p_f,\Sigma\times\{0,1\}^k\times \{0,1\}^{k^2},\beta)$ 
where $P=\{r,p_i,p_f\}\cup(\{0,1\}^{k^2}\times\{c,\ell,u,w\})$ where $c$, $\ell$, $w$, $u$ and
$r$ are modes standing for compute, left, write, undo and reset respectively.
We have $P_{+1}=\{r\}\cup(\{0,1\}^{k^2}\times\{c,w\})$, 
$P_{-1}=\{0,1\}^{k^2}\times\{\ell,u\}$ and $P_0=\{p_i,p_f\}$.
The initial state is  $p_i$ and the final state is $p_f$. 
To avoid confusion, we use $\splend$ for the endmarker of $\C_{k}^{=}$. Hence, 
this transducer works on the circular word $\splend\sem{\C_k}(u)$.

Given a $k$ bit vector $b$, we define the boolean matrix $M_b$ where the $(i,j)$
coefficient of $M_b$ equals $1$ if, and only if, $b_i=b_j=1$.  We define below transitions
in $\alpha$, where $a$ denotes a letter of $\Sigma$ different from $\#$, $b$ is a $k$-bit
vector and $M$ is a $k\times k$ boolean matrix.
All transitions output $\varepsilon$ except those reaching a state in write mode 
(\cref{lem:enrichmentEqualityEnum4,lem:enrichmentEqualityEnum5}) where the output is the
input letter enriched by the matrix of the state.
\begin{enumerate}
  \item\label{lem:enrichmentEqualityEnumInit} $(p_i,\splend,(M_0,c))$ where $M_{0}$ is the zero 
  matrix (all coefficients are $0$),
  \item\label{lem:enrichmentEqualityEnum1} $((M,c),(a,b),(M',c))$ where $M'=M+M_b$,
  \item\label{lem:enrichmentEqualityEnum2} $((M,c),(\#,b),(M,\ell))$,
  \item\label{lem:enrichmentEqualityEnumLastIt} $((M_1,c),\splend,(M_1,\ell))$ as the last
  copy of $\#u$ cycles back to the endmarker $\splend$, $M_{1}$ is the matrix where all
  coefficients are $1$,
  \item\label{lem:enrichmentEqualityEnum3} $((M,\ell),(a,b),(M,\ell))$,
  \item\label{lem:enrichmentEqualityEnum4} $((M,\ell),(\#,b),(M,w))$ and produces $(\#,b,M)$,
  \item\label{lem:enrichmentEqualityEnum5} $((M,w),(a,b),(M,w))$ and produces $(a,b,M)$,
  \item\label{lem:enrichmentEqualityEnum6} $((M,w),(\#,b),(M,u))$,
  \item\label{lem:enrichmentEqualityEnumFinal} $((M_1,w),\splend,p_f)$,
  \item\label{lem:enrichmentEqualityEnum7} $((M,u),(a,b),(M',u))$ where $M'=M-M_b$,
  \item\label{lem:enrichmentEqualityEnum8} $((M_b,u),(\#,b),r)$,
  \item\label{lem:enrichmentEqualityEnum9} $(r,(a,b),r)$,
  \item\label{lem:enrichmentEqualityEnum10} $(r,(\#,b),(M_b,c))$.
\end{enumerate}
The correctness of the construction comes from the fact that on one given marked copy of $\#u$, every pebble mark exactly one position.
Hence on correct inputs, the matrices $M$ store coherent information related to the given copy.

By construction it should be clear that $\C_{k}^{=}$ is deterministic.  Notice that, if we
replace \cref{lem:enrichmentEqualityEnum6} with $((M,w),(\#,b),(M_{b},c))$ and remove
\cref{lem:enrichmentEqualityEnum7,lem:enrichmentEqualityEnum8,lem:enrichmentEqualityEnum9,lem:enrichmentEqualityEnum10},
we would get a deterministic transducer computing the same function, but this transducer
would not be reverse-deterministic. This is why we need the undo mode followed by the 
reset mode.

Reverse-determinism of
\cref{lem:enrichmentEqualityEnum2,lem:enrichmentEqualityEnum2,lem:enrichmentEqualityEnum3,lem:enrichmentEqualityEnum4,lem:enrichmentEqualityEnum5,lem:enrichmentEqualityEnum6,lem:enrichmentEqualityEnum9,lem:enrichmentEqualityEnum10} 
is direct. Reverse-determinism of
\cref{lem:enrichmentEqualityEnum1,lem:enrichmentEqualityEnum7} comes from the determinism
of matrix subtraction and addition respectively.
Note that \cref{lem:enrichmentEqualityEnum8} is limited to $M_b$ to ensure reverse-determinism. This is due to the fact that since we are undoing the computation, what should remain in the matrix stored at the end of the undoing is exactly the information about the pebbles on $\#$, which completes the undoing of the computation.
\cref{lem:enrichmentEqualityEnumInit,lem:enrichmentEqualityEnumLastIt,lem:enrichmentEqualityEnumFinal}
deal with the acutal endmarker $\splend$, either by starting the compute mode, going back
in left mode to write the last configuration or reaching the final state.
Note that only the complete matrix $M_1$ can reach the final state, as in the last
configuration all pebbles are on the last position of $\#u$.
\end{proof}

\begin{lemma}\label{lem-decompose}
  Given a $k$-pebble ($k\geq1$) transducer $\trans$ with $n$ states, one can construct a
  $0$-pebble transducer $\trans_0$ with $\mathcal{O}(k n)$ states such that
  $\sem{\trans}=\sem{\trans_0}\circ\sem{\C_{k}^{=}}\circ\sem{\C_k}$.
\end{lemma}

\begin{proof}
  The idea is to simulate the moves of the pebbles of $\trans$ with moves of the reading
  head of $\trans_0$ along the sequence of $k$-configurations produced by $\C_k$.  The
  tests $\varphi$ are validated using the information added by $\C_{k}^{=}$.

  Informally, the transducer $\trans_0$ remembers in its state the number of pebbles
  dropped by $\trans$.  It uses the configuration where the unused pebbles are on the same
  position as the reading head.  Then, when the simulation of $\trans$ drops a pebble $i$,
  $\trans_0$ is already reading a configuration where the pebble $i$ is at the current
  position.  The transducer $\trans_0$ simply needs to increment the number of pebbles
  dropped.  Conversely, if a pebble $i$ is lifted, it means the reading head is on a
  position where the pebble $i$ is present.  Hence it only needs to decrement its number
  of pebbles dropped.

  The move of the reading head uses the fact that the configurations are output by $\C_k$
  in lexicographic order.  If the head moves to the left while having $i$ pebbles
  dropped, $\trans_0$ moves left until it sees a position where pebbles $i+1$ to $k$ are
  all present.  The first such position is the previous one from where we started due to
  the lexicographic order.  It is also the correct position to maintain the invariant
  that all undropped pebbles are placed at the same position as the head.  A move right is
  treated symmetrically.

  Because the input word of $u$ is read in a cyclic fashion, this might create issue for
  $\trans_0$ as consecutive endmarkers $\#$ do not belong to the same marking.  More
  precisely, given any position in $\C_k(u)$, its corresponding endmarker is the closest
  (possibly itself) $\#$-labelled position on its left.  So if during the computation
  $\trans_0$ reaches the endmarker $\#$ moving left and need to keep moving left,
  $\trans_0$ needs to move to a position where all pebbles $1$ to $i$ are unchanged, but
  pebbles $i+1$ to $k$ are on the last letter of the word.  This is done by reaching the
  closest position on the right where pebbles $i+1$ to $k$ are on an endmarker $\#$ (or $\splend$ in the case where we reach the end of the word)
then moving to the position on the left.

  Symmetrically, if the computation of $\trans_0$ reaches an endmarker $\#$ with a single 
  right move, it means that pebbles $i+1$ to $k$ were on the last position of $u$, then 
  $\trans_{0}$ moves left to the closest position where pebbles $i+1$ to $k$ are all on the
  endmarker $\#$.

  \newcommand{\ml}{m_{\ell}}
  \newcommand{\mr}{m_{r}}
  
  Formally, let $\trans=(Q,\Sigma,\delta,k,q_i,q_f,\Gamma,\mu)$.  We define
  $\trans_0=(P,\Sigma',\gamma,0,p_i,p_f,\Gamma,\nu)$ where
  $\Sigma'=\Sigma\times\{0,1\}^{k}\times\{0,1\}^{k^{2}}$,
  $P\subseteq(Q\times\{0,\ldots,k\}\times\{s,\mr,\ml\})\cup\{p_{i},p_{f}\}$
  where $s$ stands for simulation, and $\ml$, $\mr$ stand for move left and right
  respectively.  We divide $P$ into 
  \begin{align*}
    P_{0} &= (Q\times\{0,\ldots,k\}\times\{s\}) \cup \{p_{i},p_{f}\}, \\
    P_{+1} &= Q\times\{0,\ldots,k\}\times\{\mr\}, \\
    P_{-1} &= Q\times\{0,\ldots,k\}\times\{\ml\} \,.
  \end{align*}

  To avoid the more convoluted cases, we assume that in $\trans$ transitions with action
  drop or lift do not move the reading head (this is especially needed for $\#$).  Note
  that this can be enforced by decomposing a drop and move transition into two transitions 
  (see \cref{lem:drop-lift-move}).

  For $0\leq i\leq k$, we denote by $b^{+i}$ any vector where $b^{+i}_j=1$ for all
  $i<j\leq k$ and we denote by $b^{-i}$ any vector where $b^{-i}_{j}\neq1$ for some $j>i$.
  
  Given an integer $0\leq i\leq k$, a bit vector $b$ and a matrix $M$, we say that
  $b,M,i\models (\head=\pebble_{\ell})$ (resp.\ $\pebble_{\ell}=\pebble_{\ell'}$) if
  $\ell\leq i$ and $b_{\ell}=1$ (resp.\ $\ell,\ell'\leq i$ and $M_{\ell,\ell'}=1$).

  The transitions in $\gamma$ and the output function $\nu$ are defined below, where $a$
  denotes a letter different from $\#$, $\sigma$ is any letter (possibly $\#$), $b$ is any
  $k$-bit vector and $M$ is a $k\times k$ boolean matrix.  We omit the tests and 
  operations in transitions as $\trans_0$ is pebbleless. All transitions output 
  $\varepsilon$ except those of \cref{item:T0-simulate} below.
  \begin{enumerate}
    \item  $(p_{i},\splend,(q_{i},0,\mr))$ and 
    $((q_{i},0,\mr),(\#,b^{+0},M_{1}),(q_{i},0,s))$.
  
    \item  $((q_{f},0,s),(\#,b^{+0},M_{1}),(q_{f},0,\ml))$ and 
    $((q_{i},0,\mr),\splend,p_{f})$.
  
    \item\label{item:T0-simulate}  
    $((q,i,s),(\sigma,b^{+i},M),(q',i',m))$ if 
    there exists a transition $t=(q,\sigma,\varphi,\op,q')$ in $\trans$ such that
    $b^{+i},M,i\models\varphi$ and $\op\in\{\nop,\drop{i+1},\lift{i}\}$ and
    $b_{i}=1$ if $\op=\lift{i}$. In this case,
    \[
    i'=
    \begin{cases}
      i & \text{if } \op=\nop \\
      i+1 & \text{if } \op=\drop{i+1} \\
      i-1 & \text{if } \op=\drop{i} 
    \end{cases}
    \qquad\text{and}\qquad
    m=
    \begin{cases}
      s & \text{if } q'\in Q_{0} \\
      \ml & \text{if } q'\in Q_{-1} \wedge \sigma\neq\# \\
      \mr & \text{if } q'\in Q_{+1} \text{ or } q'\in Q_{-1} \wedge \sigma=\# 
    \end{cases}
    \]
    The production of this transition is the same as $t$.

    Note that, the head of $\trans_{0}$ moves left or right when $q'\in Q_{-1}$ or $q'\in Q_{+1}$.
    But in general, the head of $\trans_{0}$ does not reach immediately the position 
    where the simulation of $\trans$ continues. We have to skip positions with a bit 
    vector of the form $b'^{-i'}$ and handle carefully the endmarker $\#$. This is the 
    purpose of the following transitions.
    
    \item If $q\in Q_{-1}$ then $((q,i,\ml),(\sigma,b^{-i},M),(q,i,\ml))$ and
    $((q,i,\ml),(\sigma,b^{+i},M),(q,i,s))$.
    
    \item If $q\in Q_{-1}$ then $((q,i,\mr),(a,b,M),(q,i,\mr))$ and
    $((q,i,\mr),(\#,b^{-i},M),(q,i,\mr))$ and
    $((q,i,\mr),(\#,b^{+i},M),(q,i,\ml))$ and
    $((q,i,\mr),\splend,(q,i,\ml))$.
    
    Note that, when moving right we reach $\#$ with a bit vector of the form $b^{+i}$ then
    the previous letter also has a bit vector of the form $b^{+i}$.  Hence, the 
    transition taken from $(q,i,\ml)$ will go directly to the simulation mode $(q,i,s)$. 
    If moving right we reach the endmarker $\splend$ then the previous letter has a bit
    vector of the form $b^{+0}$.
    
    \item If $q\in Q_{+1}$ then $((q,i,\mr),(\sigma,b^{-i},M),(q,i,\mr))$ and
    $((q,i,\mr),(a,b^{+i},M),(q,i,s))$.
    
    \item If $q\in Q_{+1}$ then
    $((q,i,\mr),(\#,b^{+i},M),(q,i,\ml))$ and
    $((q,i,\mr),\splend,(q,i,\ml))$ and \\
    $((q,i,\ml),(a,b,M),(q,i,\ml))$ and
    $((q,i,\ml),(\#,b^{-i},M),(q,i,\ml))$ and \\
    $((q,i,\ml),(\#,b^{+i},M),(q,i,s))$.
  \end{enumerate}
  
  We remark that $\trans_0$ is obtained from $\trans$ by decomposing its transitions into
  separate sequences of transitions.  The bit vectors and the matrices $M$ allow us to
  check the tests of transitions of $\trans$.  The producing transitions of $\trans_0$ 
  correspond to the ones of $\trans$ and the accepting runs of $\trans_{0}$ correspond to 
  the accepting runs of $\trans$.
  Hence for any input word $u$, $u\in\dom{\trans}$ if, and only if
  $\sem{\C_{k}^{=}}(\sem{\C_k}(u))\in\dom{\trans_0}$ and
  $\sem{\trans}=\sem{\trans_0}\circ\sem{\C_{k}^{=}}\circ\sem{\C_k}$.
\end{proof}

\bibliography{Arxiv-rev-peb}

\begin{thebibliography}{10}

\bibitem{AC10}
Rajeev Alur and Pavol {\v{C}}ern{\'y}.
\newblock Expressiveness of streaming string transducers.
\newblock In {\em 30th {I}nternational {C}onference on {F}oundations of
  {S}oftware {T}echnology and {T}heoretical {C}omputer {S}cience, {FSTTCS}
  2010}, volume~8 of {\em LIPIcs. Leibniz Int. Proc. Inform.}, pages 1--12.
  Schloss Dagstuhl. Leibniz-Zent. Inform., Wadern, 2010.

\bibitem{AlurFreilichRaghothaman14}
Rajeev Alur, Adam Freilich, and Mukund Raghothaman.
\newblock Regular combinators for string transformations.
\newblock In Thomas~A. Henzinger and Dale Miller, editors, {\em Joint Meeting
  of the 23rd {EACSL} Annual Conference on Computer Science Logic {(CSL)} and
  the 29th Annual {ACM/IEEE} Symposium on Logic in Computer Science (LICS),
  {CSL-LICS} '14, Vienna, Austria, July 14 - 18, 2014}, pages 9:1--9:10. {ACM},
  2014.

\bibitem{BR-DLT18}
Nicolas Baudru and Pierre-Alain Reynier.
\newblock From two-way transducers to regular function expressions.
\newblock In Mizuho Hoshi and Shinnosuke Seki, editors, {\em 22nd International
  Conference on Developments in Language Theory, {DLT} 2018}, volume 11088 of
  {\em Lecture Notes in Computer Science}, pages 96--108. Springer, 2018.

\bibitem{DBLP:journals/corr/abs-1810-08760}
Mikolaj Bojanczyk.
\newblock Polyregular functions.
\newblock {\em CoRR}, abs/1810.08760, 2018.
\newblock \href {https://doi.org/10.48550/arXiv.1810.08760}
  {\path{doi:10.48550/arXiv.1810.08760}}.

\bibitem{bojanpebble}
Mikolaj Bojanczyk.
\newblock Transducers of polynomial growth.
\newblock In Christel Baier and Dana Fisman, editors, {\em {LICS} '22: 37th
  Annual {ACM/IEEE} Symposium on Logic in Computer Science, Haifa, Israel,
  August 2 - 5, 2022}, pages 1:1--1:27. {ACM}, 2022.
\newblock \href {https://doi.org/10.1145/3531130.3533326}
  {\path{doi:10.1145/3531130.3533326}}.

\bibitem{BDK-lics18}
Mikolaj Bojanczyk, Laure Daviaud, and Shankara~Narayanan Krishna.
\newblock Regular and first-order list functions.
\newblock In {\em Proceedings of the 33rd Annual {ACM/IEEE} Symposium on Logic
  in Computer Science, {LICS} 2018, Oxford, UK, July 09-12, 2018}, pages
  125--134, 2018.
\newblock \href {https://doi.org/10.1145/3209108.3209163}
  {\path{doi:10.1145/3209108.3209163}}.

\bibitem{DBLP:conf/icalp/BojanczykKL19}
Mikolaj Bojanczyk, Sandra Kiefer, and Nathan Lhote.
\newblock String-to-string interpretations with polynomial-size output.
\newblock In Christel Baier, Ioannis Chatzigiannakis, Paola Flocchini, and
  Stefano Leonardi, editors, {\em 46th International Colloquium on Automata,
  Languages, and Programming, {ICALP} 2019, July 9-12, 2019, Patras, Greece},
  volume 132 of {\em LIPIcs}, pages 106:1--106:14. Schloss Dagstuhl -
  Leibniz-Zentrum f{\"{u}}r Informatik, 2019.
\newblock \href {https://doi.org/10.4230/LIPICS.ICALP.2019.106}
  {\path{doi:10.4230/LIPICS.ICALP.2019.106}}.

\bibitem{CJ77}
Michal~P. Chytil and Vojt{\v{e}}ch J{\'a}kl.
\newblock Serial composition of {$2$}-way finite-state transducers and simple
  programs on strings.
\newblock In {\em Automata, languages and programming ({F}ourth {C}olloq.,
  {U}niv. {T}urku, {T}urku, 1977)}, pages 135--137. Lecture Notes in Comput.
  Sci., Vol. 52. Springer, Berlin, 1977.

\bibitem{Courcelle94}
Bruno Courcelle.
\newblock Monadic second-order definable graph transductions: a survey [see
  {MR}1251992 (94f:68009)].
\newblock {\em Theoret. Comput. Sci.}, 126(1):53--75, 1994.
\newblock Seventeenth Colloquium on Trees in Algebra and Programming (CAAP '92)
  and European Symposium on Programming (ESOP) (Rennes, 1992).
\newblock \href {https://doi.org/10.1016/0304-3975(94)90268-2}
  {\path{doi:10.1016/0304-3975(94)90268-2}}.

\bibitem{DFJL17}
Luc Dartois, Paulin Fournier, Isma{\"{e}}l Jecker, and Nathan Lhote.
\newblock On reversible transducers.
\newblock In Ioannis Chatzigiannakis, Piotr Indyk, Fabian Kuhn, and Anca
  Muscholl, editors, {\em 44th International Colloquium on Automata, Languages,
  and Programming, {ICALP} 2017, Warsaw, Poland}, volume~80 of {\em LIPIcs},
  pages 113:1--113:12. Schloss Dagstuhl - Leibniz-Zentrum f{\"{u}}r Informatik,
  2017.
\newblock \href {https://doi.org/10.4230/LIPIcs.ICALP.2017.113}
  {\path{doi:10.4230/LIPIcs.ICALP.2017.113}}.

\bibitem{DGK-lics18}
Vrunda Dave, Paul Gastin, and Shankara~Narayanan Krishna.
\newblock {Regular Transducer Expressions for Regular Transformations}.
\newblock In Martin Hofmann, Anuj Dawar, and Erich Gr{\"a}del, editors, {\em
  {P}roceedings of the 33rd {A}nnual {ACM\slash IEEE} {S}ymposium on {L}ogic
  {I}n {C}omputer {S}cience ({LICS}'18)}, pages 315--324, Oxford, UK, July
  2018. ACM Press.

\bibitem{DBLP:journals/acta/Engelfriet15}
Joost Engelfriet.
\newblock Two-way pebble transducers for partial functions and their
  composition.
\newblock {\em Acta Informatica}, 52(7-8):559--571, 2015.
\newblock \href {https://doi.org/10.1007/S00236-015-0224-3}
  {\path{doi:10.1007/S00236-015-0224-3}}.

\bibitem{EH01}
Joost Engelfriet and Hendrik~Jan Hoogeboom.
\newblock M{SO} definable string transductions and two-way finite-state
  transducers.
\newblock {\em ACM Transactions on Computational Logic}, 2(2):216--254, 2001.
\newblock \href {https://doi.org/10.1145/371316.371512}
  {\path{doi:10.1145/371316.371512}}.

\bibitem{DBLP:journals/tocl/EngelfrietH01}
Joost Engelfriet and Hendrik~Jan Hoogeboom.
\newblock {MSO} definable string transductions and two-way finite-state
  transducers.
\newblock {\em {ACM} Trans. Comput. Log.}, 2(2):216--254, 2001.
\newblock \href {https://doi.org/10.1145/371316.371512}
  {\path{doi:10.1145/371316.371512}}.

\bibitem{DBLP:conf/mfcs/EngelfrietM02}
Joost Engelfriet and Sebastian Maneth.
\newblock Two-way finite state transducers with nested pebbles.
\newblock In Krzysztof Diks and Wojciech Rytter, editors, {\em Mathematical
  Foundations of Computer Science 2002, 27th International Symposium, {MFCS}
  2002, Warsaw, Poland, August 26-30, 2002, Proceedings}, volume 2420 of {\em
  Lecture Notes in Computer Science}, pages 234--244. Springer, 2002.
\newblock \href {https://doi.org/10.1007/3-540-45687-2\_19}
  {\path{doi:10.1007/3-540-45687-2\_19}}.

\end{thebibliography}

\appendix
\crefalias{section}{appendix}
\crefalias{subsection}{appendix}

\clearpage

\section{Pebbleless transducers}
\label{sec:app:0peb2w}

A $0$-pebble transducer can equivalently be seen as a two-way transducer.
In~\cite{DFJL17}, the authors define reversible two-way transducers, using a different 
semantics than the one used here for $0$-pebble transducers.
In the following, we call a $0$-pebble transducer one using the semantics defined in our paper.
On the other hand, a two-way transducer refers to a machine using the semantics defined in~\cite{DFJL17}.

We now define the semantics of two-way transducers.
Here, the input word is not considered as circular, and instead of the \( \# \) symbol,
two different endmarkers are added, \( \leftend \) at the begining of the word and \( \rightend \) at the end.
We denote with \( \Sigma_\leftend \) (resp.\ \( \Sigma_\rightend \), \( \Sigma_{\leftend\rightend} \))
the alphabet \( \Sigma \cup \left\{ \leftend \right\} \)
(resp.\ \( \Sigma \cup \left\{ \rightend \right\} \), \( \Sigma \cup \left\{ \leftend, \rightend \right\} \)).
The positions for the head of the transducer are in between letters, instead of being on the letters.
The word
$u=\setlength{\arraycolsep}{1pt}\begin{array}{cccccc} \# & a & b & c & d & \\ 0 & 1 & 2 & 3 & 4 & \end{array}$
from our semantics becomes
$\setlength{\arraycolsep}{0pt}\begin{array}{ccccccccccccc} & \leftend && a && b && c && d && \rightend & \\ 0 && 1 && 2 && 3 && 4 && 5 && 6 \end{array}$.
For an in between position \( \head \in \left\{ 0, \ldots, |u| + 1 \right\} \), we denote with \( \head + 1/2 \) the letter
on the right of the reading head.
For example, \( u(2+1/2) = b \).
Similarly, for a position \( \head \in \left\{ 1, \ldots, |u| + 2 \right\} \),
\( \head - 1/2 \) is the letter on the left of the reading head.
A two-way transducer is a tuple \( \twtrans = \left( Q, \Sigma, \delta, q_i, q_f, \Gamma, \mu \right) \) where
\( Q = Q^+ \uplus Q^- \) is a finite set of states,
with \( Q^+ \) the set of \emph{forward} states and \( Q^- \) the set of \emph{backward} states;
\( \Sigma \) is the alphabet; \( q_i, q_f \in Q^+ \) are respectively the (only) initial and accepting states;
\( \delta \subseteq Q \times \Sigma_{\leftend\rightend} \times Q \) is the set of transitions;
\( \Gamma \) is the output alphabet;
\( \mu \colon \delta \to \Gamma^{*} \) is the output function.
By convention, \( \left( q_i, \leftend, p \right) \) and \( \left( q, \rightend, q_f \right) \) are the
only types of transitions where \( q_i \) and \( q_f \) can appear, respectively.
Moreover, \( \left( q, \leftend, p \right) \in \delta \) implies that \( p \in Q^+ \),
and \( \left( q, \rightend, p \right) \in \delta \) implies that either \( p = q_f \) or \( p \in Q^- \).
Because the word is not circular, the head cannot move past an end-marker,
except in a transition $(q_i,\leftend,p)$ leaving the initial state, 
or in a transition $(q,\rightend,q_f)$ reaching the accepting state.

A configuration of a two-way transducer on an input word \( u \) is the pair \( \left( q, \head \right) \),
composed by the current state of the transducer and
the position of the reading head, with \( \head \in \left\{ 0, \ldots, |u|+2 \right\} \).
A transition \( t = \left( q, a, p \right) \) is enabled between configurations
\( C = \left( q, \head \right) \) and \( C' = \left( p, \head' \right) \),
denoted \( C \xrightarrow{t} C' \), if either:
\begin{itemize}
	\item \( q \in Q^+, p \in Q^+, u(\head + 1/2) = a,  \head' = \head + 1 \),
	\item \( q \in Q^+, p \in Q^-, u(\head + 1/2) = a,  \head' = \head \),
	\item \( q \in Q^-, p \in Q^-, u(\head - 1/2) = a,  \head' = \head - 1 \), or
	\item \( q \in Q^-, p \in Q^+, u(\head - 1/2) = a,  \head' = \head \).
\end{itemize}

The two following lemmas prove that the two semantics are effectively equivalent, and that reversibility is preserved between them.

\begin{lemma}\label{lem-0pebTo2way}
	Given a $0$-pebble transducer $P$ with $n$ states, 
	one can construct an equivalent two-way transducer $T$ with $\mathcal{O}(n)$ states.
	Moreover, if $P$ is reversible then so is $T$.
\end{lemma}
\begin{proof}
	The differences between the semantics defined in this document and the one in~\cite{DFJL17} are:
	\begin{enumerate}
		\item The $0$-pebble places the head on a position, and reads it to determine the
		transition, while the two-way transducers place the reading head between position, and
		read either the letter on the left or on the right depending on the polarity of the
		state.
		
		\item The $0$-pebble transducers have only one endmarker $\#$ and cycle on their
		input, while the two-way transducers have $\leftend$ and $\rightend$ marking the left and
		right limits of the input.
		
		\item The $0$-pebble transducers is able to not move its reading head via the set $Q_0$.
		
		\item The move of the reading head of a transition in the $0$-pebble transducer is
		determined fully by the state reached, while two-way transducers give a half move to
		each of the starting and reached states.
	\end{enumerate}
	
	We now explain how to handle each of these differences.
	Let \( P=(Q,\Sigma, \delta, 0, q_i, q_f, \Gamma, \mu) \) be 0-pebble transducer,
	with \( Q=Q_{+1} \uplus Q_0 \uplus Q_{-1} \).
	\begin{enumerate}
		\item For each state $q$ of $P$, we have a \emph{right} copy $\qright{q}$ in $T$ which
		is a $+$ state.  The copy $\qright{q}\in S^{+}$ works as if it is placed on the letter
		on its right. More precisely, a configuration $(q,\head)$ of $P$ with $\head\neq0$ is encoded 
		by the configuration $(\qright{q},\head)$ of $T$,
		a configuration $(q,0)$ of $T$ with $q\neq q_{f}$ is encoded 
		by the configuration $(\qright{q},|u|+1)$ of $T$,
		and the final configuration $(q_{f},0)$ of $P$ is encoded by the final configuration
		$(s_{f},|u|+2)$ of $T$.

		\item For each state $p\in Q_{0}$ of $P$, we add a \emph{stay} copy $\qstay{p}\in S^-$
		to $T$, and for each transition $(q,a,p)$ of $P$, we add in $T$ transitions
		$(\qright{q},a,\qstay{p})$ and $(\qstay{p},b,\qright{p})$ for all $b\in\Sigma_\leftend$.
		See \cref{fig:0peb2w:stay}.
		
		\item Similarly, for each state $p\in Q_{-1}$ of $P$, we add two left copies
		$\qlefto{p},\qleftt{p}\in S^{-}$ in $T$, and for each transition $(q,a,p)$ of
		$P$, there are in $T$ transitions $(\qright{q},a,\qlefto{p})$,
		$(\qlefto{p},b,\qleftt{p})$ and $(\qleftt{p}, b', \qright{p})$ for
		$b\in\Sigma$ and \( b'\in\Sigma_\leftend \).
		See \cref{fig:0peb2w:left}.
		
		\item The transducer $T$ treats $\rightend$ as $\#$.
		The cycling is handled with copies of states.  If $P$ reading $\#$ moves right to a
		state $p\in Q_{+1}$, the two-way transducer $T$ switches to the copy $\qminus{p}$
		and moves back to $\leftend$ where it switches back to $\qright{p}$.  Moving left from
		the first position of the word is handled symmetrically with a $\qplus{p}$ copy of
		$p\in Q_{-1}$ as shown in \cref{fig:0peb2w:left}.
	\end{enumerate}
	
	\begin{figure}[t]
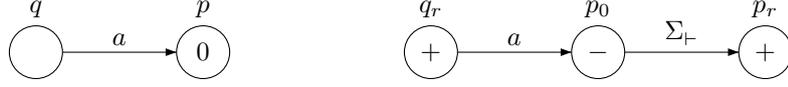

		\centering
		\gusepicture{zeroP}
		\hspace{2cm}
		\gusepicture{zeroT}
		\caption{A transition towards a \( 0 \) state in P. Gadget simulating in \( T \) the transition from the left.}
		\label{fig:0peb2w:stay}
	\end{figure}
	
	\begin{figure}[t]
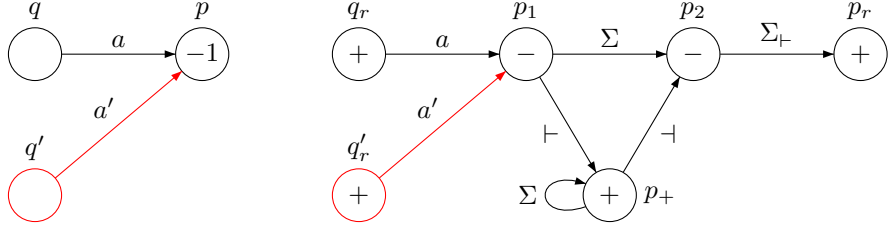

		\centering
		\gusepicture{minusP}
		\hfil
		\gusepicture{minusT}
		\caption{Simulation of a transition going left on the input tape.
			On the left, two transitions towards a \( -1 \) state in transducer \( P \).
			On the right, the gadget in \( T \) simulating these transitions.}
		\label{fig:0peb2w:left}
	\end{figure}
	
	Formally, we define the two-way transducer 
	\( T=(S^+ \uplus S^-, \Sigma, \delta', s_{i},s_{f}, \Gamma, \mu') \) by:
	\begin{itemize}
		\item \( S^+ = S_r \uplus S_+ \uplus \left\{ s_{i},s_{f} \right\} \) where
		\( S_r = \left\{ \qright{q} \mid q \in Q \right\} \), and
		\( S_+ = \left\{ \qplus{q} \mid q \in Q_{-1} \right\} \).
		
		\item \( S^- = S_0 \uplus S_1 \uplus S_2 \uplus S_- \) where
		\( S_0 = \left\{ \qstay{q} \mid q \in Q_{0} \setminus \left\{ q_{i}, q_{f} \right\} \right\} \),
		\( S_1 = \left\{ \qlefto{q} \mid q \in Q_{-1} \right\} \),
		\( S_2 = \left\{ \qleftt{q} \mid q \in Q_{-1} \right\} \), and
		\( S_- = \left\{ \qminus{q} \mid q \in Q_{+1} \right\} \).
		
		\item \( \{ \left( s_{i}, \leftend, q_{ir} \right) \} \cup \{ \left( q_{ir}, b, q_{ir}
		\right) \mid b \in \Sigma \} \subseteq \delta' \): initially, $T$ reads \( \leftend \)
		and moves to the left of \( \rightend \) to reach the encoding $(q_{ir},|u|+1)$ of the
		initial configuration $(q_{i},0)$ of $P$.
		
		\item If \( \left( q, a, p \right) \in \delta \) with $a\neq\#$ then
		$(\qright{q},a,s')\in\delta'$ with
		$s'=
		\begin{cases}
			\qright{p} & \text{if } p \in Q_{+1} \\
			\qstay{p} & \text{if } p \in Q_{0}\setminus\{q_{i},q_{f}\} \\
			\qlefto{p} & \text{if } p \in Q_{-1} 
		\end{cases}$
		
		\item If \( \left( q, \#, p \right) \in \delta \) then
		$(\qright{q},\rightend,s')\in\delta'$ with
		$s'=
		\begin{cases}
			\qminus{p} & \text{if } p \in Q_{+1} \\
			\qstay{p} & \text{if } p \in Q_{0}\setminus\{q_{i},q_{f}\} \\
			s_{f} & \text{if } p=q_{f} \\
			\qlefto{p} & \text{if } p \in Q_{-1} 
		\end{cases}$
		
		\item For all \( q \in Q_{+1} \), we have
		\( \left\{ \left( \qminus{q}, b, \qminus{q} \right) \mid b \in \Sigma \right\}
		\cup\{ \left( \qminus{q}, \leftend, \qright{q} \right) \}	\subseteq \delta' \).
		
		\item For all \( q \in Q_{0} \), we have
		\( \left\{ \left(  \qstay{q}, b, \qright{q}  \right) \mid b \in \Sigma_{\leftend} \right\} \subseteq \delta' \).
		
		\item For all \( q \in Q_{-1} \), we have
		$\{(\qlefto{q},b,\qleftt{q})\mid b\in\Sigma\}\cup\{(\qleftt{q},b',\qright{q})\mid 
		b'\in\Sigma_\leftend\}\subseteq\delta'$ and \\
		$\{(\qlefto{q},\leftend,\qplus{q})\}\cup
		\{(\qplus{q},b,\qplus{q})\mid b\in\Sigma\}\cup
		\{(\qplus{q},\rightend,\qleftt{q}\}\subseteq\delta'$.
		
		\item \( \mu'(s, a, s') = \varepsilon \) if \( s \notin S_{r} \) or $s=q_{ir}$ and
		$a\in\Sigma$, and \( \mu'(\qright{q}, a, s') = \mu(q, a, p) \) with \( s' \in S \)
		obtained from \( p \in Q \) as described above.
	\end{itemize}

	We now show that if \( P \) is reversible, then so is \( T \).
	
	First, notice that whithin the gadgets for each state, $T$ is deterministic:
	if \( s \in S_0 \cup S_1 \cup S_2 \cup S_+ \cup S_- \cup \left\{ s_{i},s_{f} \right\} \),
	then by definition of \( \delta' \) determinism is ensured.
	For instance, if \( s=\qlefto{q} \in S_1 \), the only transitions starting from $s$ are
	\( \left( s, b, \qleftt{q} \right) \) for \( b \in \Sigma \) and
	\( \left( s, \leftend, \qplus{q} \right) \).
	All these transitions are labelled with different letters.
	
	Next, we need to deal with transitions starting from some $s=\qright{q}\in S_{r}$.  Then,
	for all $a\in\Sigma_{\rightend}$, we have \( \left( \qright{q}, a, s' \right) \in
	\delta' \) iff \( \left( q, a, p \right) \in \delta \) where $s'$ is
	\emph{deterministically} defined from $p$ as explained in the definition of $\delta'$.
	Hence, if $P$ is deterministic then so is $T$.
	
	We now show that $T$ is reverse-deterministic.
	We proceed by case analysis on the target state \( s' \) of a transition.
	If \( s'=\qleftt{q} \in S_2 \), the incoming transitions are
	\( \left( \qlefto{q}, b, \qleftt{q} \right) \) for all \( b \in \Sigma \),
	and \( \left( q_+, \rightend, \qleftt{q} \right) \).
	All of them are labelled with different letters, ensuring reverse-determinism for 
	states in $S_{2}$.
	The same holds for states \( s'=\qplus{q} \in S_+ \).
	If \( s'=\qright{q} \in S_r \), and \( q \notin Q_{+1} \), then the same holds again:
	the only way to reach \( \qright{q} \) is whithin its gadget
	(either from $s_{i}$ reading $\leftend$ or from  $q_{ir}$ reading $a\in\Sigma$ if $q=q_{i}$, 
	or from $q_{0}$ if $q\in Q_{0}\setminus\{q_{i}\}$, or from $\qleftt{q}$ if $q\in Q_{-1}$).
	The labels of all these transitions are different, ensuring reverse-determinism.
	The states we have to be careful about are the entry points of each gadget.
	
	If \( s'=\qlefto{p} \in S_1 \), then \( \left( s, a, s' \right) \in \delta' \) 
	iff $s=\qright{q}$ and \( \left( q, a, p \right) \in \delta \) with $a\in\Sigma_{\leftend}$.
	It follows that if \( P \) is reverse-deterministic, then so is $T$ at states in $S_{1}$.
	We proceed similarly when $s'=p_{0}\in S_{0}$.
	
	If \( s'=\qminus{p} \in S_- \), there is a transition $(\qminus{p},b,\qminus{p})$ for each \( b \in \Sigma \)
	(that is part of the gadget allowing to simulate going past the right end marker),
	and a transition \( \left( \qright{q}, \rightend, \qminus{p} \right) \) for each
	\( \left( q, \#, p \right) \in \delta \).
	Here again, if \( P \) is reverse-deterministic, then so is $T$ at states in $S_{-}$.
	
	Finally, assume that \( s'=\qright{p} \in S_r \) with \( p \in Q_{+1} \).
	Then, we have a transition $(s,a,\qright{p})\in\delta'$ iff either $a=\leftend$ and $s=\qminus{p}$
	(part of the gadget allowing to go past the right end marker), or $a\in\Sigma$ and 
	$s=\qright{q}$ with $(q,a,p)\in\delta$.  
	As before, if \( P \) is reverse-deterministic then at most one transition
	reaches \( s'=\qright{p} \) labelled by any given letter.
	
	With this we conclude that if \( P \) is reverse-deterministic,
	then so is \( T \).
\end{proof}

\begin{lemma}\label{lem-2wayTo0peb}
	Given a (reversible) two-way transducer $T$ with $n$ states, 
	one can construct an equivalent $0$-pebble (reversible) transducer $P$ with $n$ states.
\end{lemma}
\begin{proof}
	Let $T=(Q^{+}\cup Q^{-},\Sigma,\delta,q_{i},q_{f},\Gamma,\mu)$.
	The $0$-pebble transducer $P$ is defined with the same set of states and transitions as 
	$T$, with the following minor change.
	
	First, we need to handle the position of the reading head and the changes on how the
	moves are decided.  We do so by maintaining the fact that the reading head in $P$ is
	always on the position to be read.  Given a transition $(p,a,q)$ in the two-way
	semantics, the state $p$ decides which position to read, and the state $q$ decides which
	position will be read next.  It is either the next one if $q\in Q^+$ or the previous one
	if $q\in Q^-$.  Then,  keeping the same polarity for the states of $T$ and the ones of
	$P$ (except for $q_{i},q_{f}$), we obtain the same runs.  More precisely, the states of
	$P$ are partioned as $S_{+1}=Q^{+}\setminus\{q_{i},q_{f}\}$, $S_{0}=\{q_{i},q_{f}\}$ and
	$S_{-1}=Q^{-}$.

	We also keep essentially the same transitions in $P$ as in $T$, simply replacing
	$\leftend$ and $\rightend$ in transitions of $T$ with $\#$ in transitions of $P$, since
	$0$-pebble transducers have only one endmarker $\#$ standing for both $\leftend$ and
	$\rightend$.  We need to argue that this does not impact reversibility.  In fact, the only
	transitions in $T$ reading $\leftend$ are of the form $(p,\leftend,q)$ where $p\in
	Q^-\cup\{q_{i}\}$ and $q\in Q^+\setminus\{q_{f}\}$.  For $\rightend$, these transitions are
	of the form \( \left( p, \rightend, q \right) \) where \( p \in Q^+\setminus\{q_{i}\} \)
	and \( q \in Q^-\cup\{q_{f}\} \).  So there is no overlap between transitions in $T$
	reading $\leftend$ and those reading $\rightend$.  Hence, $\#$ can acts as both
	$\leftend$ and $\rightend$ while preserving reversibility.
\end{proof}

\section{Simulating equality tests of pebbles}
\label{app:equiv}
In this section, we show that one can simulate a $k$-pebble transducer  $\auto$ with equality tests   using a basic $k$-pebble transducer $\basic$. Since the main difference between the two is in the tests, one where we can easily check whether two pebbles are on the same position or not, we need to simulate this aspect in the basic pebble one. Let $\auto=(Q, \Sigma, \delta, k, q_0, q_f, \Gamma, \mu)$ be a $k$-$\mathsf{PT}_{=}$.

The idea is to come up with finite abstractions of the relative  positions of pebbles.
Recall that basic pebble transducers allow tests that check if a pebble is on the head or not, so we can directly use these while constructing $\basic$. In $\basic$, we work with guards that are  bit vectors $\in \{0,1\}^k$ which represent a pebble configuration at a head position. Thus, the  all 1-bitvector $(1,1,\dots,1)$ denotes 
all pebbles are at the head and represents the guard $\bigwedge_{i=1}^k(\head=\pebble_i)$
while the bitvector $(1,0,\dots, 0)$ where only the first bit is a  1 represents 
$\head=\pebble_1 \wedge \bigwedge_{i=2}^k \neg (\head=\pebble_i)$ and so on. The guards 
used in $\basic$ are therefore ``complete'' in the sense, for each transition 
of $\auto$ testing $\head=\pebble_j$ (or $\neg(\head=\pebble_j)$), 
we have $2^{k-1}$ transitions in $\basic$ guarded by the $2^{k-1}$ distinct bitvectors 
where all of them have the $j$th bit 1 (resp., $j$th bit 0) and all combinations
for others.

Thus, the key elements in the construction of $\basic$ are the following.
\begin{enumerate}
	\item[(a)] Extend the states of $\auto$ with a $k \times k$ matrix $M$ which tracks 
	the positions of all pebbles, in such a way that $M_{i,i}=1$ iff $\pebble_i$ is dropped, 
	and $M_{i,j}=M_{j,i}=1$ iff $\pebble_i, \pebble_j$ are at the same position. Thus, the set of states goes from $Q$ in $\auto$ to $Q \times M$ in $\basic$. In essence, 
	guards in $\auto$ of the form $\pebble_i=\pebble_j$  are only enabled in $\basic$ from states of the form $(q,M)$ where $M_{i,j}=M_{j,i}=1$.

	\item[(b)]  As mentioned above, we have a blow up in the transitions where 
	we have one transition for  each possibility of pebble arrangement for any head position.  
\end{enumerate}

We construct the basic pebble transducer $\mathcal{B}=(Q', \Sigma, \delta', k, q'_0, q'_f, \Gamma, \mu')$ as follows.

\begin{enumerate}
	\item $Q'=Q \times \{0,1\}^{k\times k}$. 
	We extend the set of states 
	$Q$ with information regarding the relative positions of the $k$ pebbles.
	A state in $Q'$ thus has the form $(q, \alpha)$ where  $\alpha=(\alpha_{i,j})_{1 \leq i, j \leq k}$ is a 	$k \times k$ Boolean matrix, with the encoding that $\alpha_{i,j}=1$ iff $\pebbles_i=\pebbles_j$, that is, both pebbles $\pebble_i, \pebble_j$ are at the same position. $\alpha_{i,i}=1$ is an encoding of the fact that pebble $\pebble_i$ has been dropped somewhere on the word. Finally, $(q, \alpha) \in Q'_i$, $i \in \{+1,-1,0\}$ 
	iff $q \in Q_i$.  
	
	\item The initial state $q'_i=(q_i, \overline{0})$ where $\overline{0}$ denotes the 
	$k \times k$ zero matrix. This signifies that no pebbles are dropped at the beginning.
	\item The final state $q'_f=(q_f, \overline{0})$. Recall that, for acceptance, all pebbles are lifted. 
	\item The transitions $\delta'$ are defined in two stages.  
	For each transition 
	$t=(q,a, \varphi, \op, q')$ in $\auto$, we define transitions $t'=((q, \alpha), a,	b, \op, (q', \alpha'))$ where $a \in \Sigma$, $b \in \{0,1\}^k$
	such that $(\alpha,b)$ constitute a consistent abstraction pair : that is, $\alpha, b$ are consistent with each other,  the guard $\varphi$ as well as the operation $\op$.

	\subsection*{Basic Consistency of the pair ($\alpha, b$)}		
	This is a  basic consistency check done using  invariants I1-I5. The component $\alpha$ in source state and the bitvector guarding a transition from a source state $(q, \alpha))$ must be consistent.

	For each $\alpha$ which is part of the source state $(q, \alpha)$ in $t'$, 
	an outgoing transition can be decorated with the bitvector $b$ only if the following invariants (I1)-(I5) are  true. 
	
	\begin{itemize}
		\item [I1] The entries of $\alpha$ are symmetric and transitive. It is easy to see that for $1 \leq i, j \leq k$, if pebble  $\pebble_i$ and pebble $\pebble_j$ are at the same position, then $\alpha_{i,j}=1=\alpha_{j,i}$. Likewise, if $\pebble_i, \pebble_j$ are not at the same position, then $\alpha_{i,j}=0=\alpha_{j,i}$. Likewise, if $\pebble_i, \pebble_j$ are at the same position, and    $\pebble_j, \pebble_n$ are at the same position, then $\pebble_i, \pebble_n$ are at the same position. That is, 
		$\alpha_{i,j}=1 \wedge \alpha_{j,n}=1 \Rightarrow \alpha_{i,n}=1$.
		\item[I2] For $1 \leq i \leq k$, $(b_i=1) \Rightarrow (\alpha_{i,i}=1)$. Recall that $b_i=1$ represents that pebble $\pebble_i$ is on the head position, hence indeed it has been dropped implying $\alpha_{i,i}=1$. 
		\item[I3] For $1 \leq  i, j \leq k$, $(b_i=1=b_j) \Rightarrow \alpha_{i,j}=1$. 
		Indeed, if pebbles $\pebble_i, \pebble_j$ are both at the head position, they are at the same position. Note that by I1, we also get $\alpha_{j,i}=1$.
		\item[I4] For $1 \leq  i, j \leq k$, $\alpha_{i,j}=1 \Rightarrow (\alpha_{i,i}=1 \wedge \alpha_{j,j}=1 \wedge b_i=b_j)$. If pebbles $\pebble_i, \pebble_j$ are at the same position, then indeed they have both been dropped, and they are both either at the head position, or both are not on the head. 
		\item[I5] $(1 \leq i < j \leq k \wedge \alpha_{j,j}=1) \Rightarrow \alpha_{i,i}=1$. 
		By the pebble stack policy, if a higher indexed pebble $\pebble_j$ has been dropped, then indeed, all lower indexed pebbles $\pebble_i$ have been dropped. 
	\end{itemize}
	
	\subsection*{Transition Consistency : $\alpha, b \models \varphi, \op$}		
	Next, we present the consistency conditions 
	between $\alpha, b$ and $\varphi$ as well as $\op$. 
	Formally, 	$\alpha, b \models \varphi$ is defined as follows. 
	\begin{itemize}
		\item 	If $\varphi$ is an atom of the form $\head=\pebble_i$, then $\alpha, b \models (\head=\pebble_i)$ if $b_i=1$. 
		\item If $\varphi$ is an atom of the form $\pebble_i=\pebble_j$, then $\alpha, b \models (\pebble_i=\pebble_j)$ if $\alpha_{i,j}=1$. 
	\end{itemize}
	We now proceed to define $\alpha, b \models \op$. This depends on the operation $\op$ and $b$.  Here, we also specify the target $\alpha'=\op(\alpha,b)$ obtained.  
	\begin{itemize}
		\item If $\op=\nop$. $\alpha, b \models \nop$ for all $\alpha, b$. In this case, $\alpha'=\alpha$ since no pebble has been touched.
		\item $\op=\push~\pebble_n$. $\alpha, b \models \push~\pebble_n$ if $\alpha_{n,n}=0$ and $\alpha_{n-1,n-1}=1$ if $n >1$. $\alpha'$ is defined as follows. 
		First of all, $\alpha'_{i,j}=1=\alpha_{i,j}$ for all 	$i, j < n$. Likewise, 
		$\alpha'_{i,j}=0$ if $i > n$ or $j > n$. Last, $\alpha'_{n,n}=1$, and 
		$\alpha'_{n,j}=1$ if $j<n$ and $b_j=1$.  $\alpha'_{j,n}=0$ for $j>n$ or $b_j=0$.
		\item $\op=\pop~\pebble_n$. $\alpha, b \models \pop~\pebble_n$ if $b_n=1$ and $\alpha_{n+1,n+1}=0$ if $n<k$. $\alpha'$ is defined as follows.
		$\alpha'_{i,j}=\alpha_{i,j}$ if $i, j < n$, and $\alpha'_{i,j}=0$ if $i \geq n$ or $j \geq n$.
	\end{itemize}
	
	\item Finally, the output $\mu$ associated to a transition $t$ in $\auto$ is the same as that associated with $t'$ in $\basic$.

\end{enumerate}

This concludes the construction of $\basic$ from $\auto$. For a $\alpha, b$ pair, checking if they are consistent with each other takes time linear in their sizes. Likewise, given a $\varphi$ and $\op$, checking whether $\alpha, b \models \varphi, \op$ takes time linear 
in their sizes.

\subsection{Configuration based consistency}
\cref{sec:constr} presented a construction of the transducer $\basic$
without equality checks, where the new transitions were obtained by adding on finite abstractions  of the pebble space to the state, and expanding the  transitions 
to include an abstraction of the pebbles on the head.

To prove the correctness of our construction, we have to reason that from any configuration $C=(q, \head, \pebbles)$ which enables a transition $t=(a, \varphi, \op)$ in $\auto$,  there is a configuration $C'=((q, \alpha), \head, \pebbles)$ in $\basic$ which enables the corresponding transition $t'=(a, b, \op)$ and conversely. Towards this, we define 
a map $\Phi=(\Phi_1, \Phi_2)$ between configurations of $\auto$ and $\basic$ as follows. 

Given a configuration $C=(q, \head, \pebbles)$  enabling  $t=(a, \varphi, \op)$ in  $\auto$, $\Phi(C)$  is the unique configuration $C'=((q, \Phi_1(\pebbles)), \head, \pebbles)$ of $\basic$ enabling  $t'=((q, \Phi_1(\pebbles)),a, \Phi_2(\head, \pebbles), \op)$ where $\Phi_1, \Phi_2$ respectively compute 
$\alpha \in \{0,1\}^{k \times k}$ and the bitvector guard $b \in \{0,1\}^k$ as follows.

\begin{align*}
	\Phi_1(\pebbles) = \quad \left\{
	\begin{array}{r l}
		\alpha_{i,j}=1 & \text{if}~i, j \leq |\pebbles| \wedge \pebbles_i=\pebbles_j \\
		\alpha_{i,j}=0 & \text{otherwise}	
	\end{array}
	\right. \\
	\Phi_2(\head, \pebbles)= \quad \left\{
	\begin{array}{r l}
		b_i=1 & \text{if}~i \leq |\pebbles| \wedge \head=\pebbles_i ~~~~~~~~~~\\
		b_i=0 & \text{otherwise}	
	\end{array}
	\right. \\
\end{align*}

\begin{lemma}\label{lem:basic}
	Let $C=(q, \head, \pebbles)$ be a configuration in $\auto$.
	The pair $(\Phi_1(\pebbles), \Phi_2(\head, \pebbles))$ satisfies basic 
	consistency. 
	
\end{lemma}
\begin{proof}
	\begin{enumerate}
		\item First of all, the $\alpha$ computed by $\Phi_1$ is symmetric and transitive. Hence, (I1) is preserved. 
		\item 
		Next, consider $1 \leq i \leq k$ such that $i \leq |\pebbles|$ and 
		$h=\pebbles_i$. Then $\Phi_2$ sets $b_i=1$. We also have $\pebbles_i=\pebbles_i$, and $\Phi_1$ assigns $\alpha_{i,i}=1$.  Now,  for $i \leq |\pebbles|$ if we do not have $h=\pebbles_i$, then 
		$\Phi_2$ sets $b_i=0$. 
		In both cases,  (I2) is preserved by $\Phi_1, \Phi_2$. 
		\item Next, consider 
		$1 \leq i, j \leq |\pebbles|$ and  $\head=\pebbles_i=\pebbles_j$.
		Then $\Phi_2$ sets $b_i=1=b_j=1$. However, 
		$\head=\pebbles_i=\pebbles_j \Rightarrow \pebbles_i=\pebbles_j$
		and $\Phi_1$ sets $\alpha_{i,j}=1$. Noe suppose 
		$\pebbles_i=\pebbles_j$, but $\head \neq \pebbles_i$ (and 
		$\head \neq \pebbles_j$). Then $b_i=0=b_j$. However,  
		$\Phi_1$ sets $\alpha_{i,j}=1$. In both cases,  (I3)  holds. 
		\item Likewise, for $1 \leq i, j \leq k$ and $\pebbles_i=\pebbles_j$, $\Phi_1$ assigns $\alpha_{i,j}=1$. Then clearly, either $\head=\pebbles_i=\pebbles_j$ or 
		$\head \neq (\pebbles_i=\pebbles_j)$, and in both cases, $b_i=b_j$ 
		under $\Phi_2$. This shows that (I4) holds. 
		\item Finally, 
		for all $1 \leq j \leq |\pebbles|$,  $\pebbles_j=\pebbles_j$. Thus, 
		for $i < j$, $\pebbles_i=\pebbles_i$ by the stack discipline. 
		$\Phi_1$ assigns 
		$\alpha_{j,j}=1$ and  $\alpha_{i,i}=1$, preserving (I5).
	\end{enumerate}
	
\end{proof}

\begin{lemma}\label{lem:trans}
	Let $C=(q, \head, \pebbles)$ be a configuration in $\auto$. For 
	any guard $\varphi$ and operation $\op$ in $\auto$, 
	$(\head, \pebbles) \models \varphi$ (resp., $(\head, \pebbles) \models \op$) $\Leftrightarrow (\Phi_1(\pebbles), \Phi_2(\head, \pebbles)) \models \varphi$ (resp., $(\Phi_1(\pebbles), \Phi_2(\head, \pebbles)) \models\op$).	Finally, 
	$\op(\Phi_1(\pebbles), \Phi_2(\pebbles, \head))=\Phi_1(\op(\pebbles, \head))$.
	
\end{lemma}
\begin{proof}
	
	$(\head, \pebbles) \models \varphi, \op$ has been defined earlier. We perform a case analysis on $\varphi$. 
	\begin{enumerate}
		\item $\varphi=(\head=\pebbles_i$). $(\head, \pebbles) \models \varphi$ iff $1 \leq i \leq |\pebbles|$ and $\pebbles_i=\head$. By construction of $\Phi_1, \Phi_2$, 
		$b_i=1$. By definition of transition consistency,  we obtain $(\Phi_1(\pebbles), \Phi_2(\head, \pebbles)) \models \varphi$.	
		\item $\varphi=(\pebbles_i=\pebbles_j$). $(\head, \pebbles) \models \varphi$ iff $1 \leq i \leq |\pebbles|$ and $\pebbles_i=\pebbles_j$. By construction of $\Phi_1$, 
		$\alpha_{i,j}=1$. By definition of transition consistency,  we obtain $(\Phi_1(\pebbles), \Phi_2(\head, \pebbles)) \models \varphi$.	
	\end{enumerate}
	We perform a case analysis on $\op$. 
	\begin{enumerate}
		\item $\op=\nop$. This is trivial. In this case, 
		$\op(\pebbles, \head)=\pebbles$. So, $\Phi_1(\op(\pebbles, \head))=\Phi_1(\pebbles)$.
		This coincides with the definition of $\alpha'=\op(\Phi_1(\pebbles), \Phi_2(\pebbles, \head))$ in $\basic$ since $\nop$ does not change anything. 
		\item $\op=\push~\pebbles_n$. $(\head, \pebbles) \models \push~\pebbles_n$ iff $|\pebbles|=n-1$. By construction of $\Phi_1$, $\alpha_{n,n}=0$ as $n > |\pebbles|$ and 
		$\alpha_{n-1,n-1}=1$ as $n-1 \leq |\pebbles|$. By definition of transition consistency,  we obtain $(\Phi_1(\pebbles), \Phi_2(\head, \pebbles)) \models \op$. 
		
		Now, $\op(\pebbles,\head)=\pebbles'=\pebbles.\head$, and 
		$|\pebbles'|=n$. Then $\Phi_1(\pebbles')$ assigns $\alpha_{n,n}=1$, 
		$\alpha_{i,n}=1$ for all $\pebbles_i=\pebbles_n$ and 
		$\alpha_{n+1, n+1}=0$. By construction of $\basic$, it can be seen that  
		$\alpha'=\op(\Phi_1(\pebbles), \Phi_2(\head, \pebbles))$ is precisely what we get above.

		\item $\op=\pop~p_n$. $(\head, \pebbles) \models \pop~p_n$ iff $|\pebbles|=n \geq 1$ and $\head=\pebbles_n$. By construction of $\Phi_2$, $b_n=1$, and by $\Phi_1$, 
		$\alpha_{n+1,n+1}=0$ as $n+1 >  |\pebbles|$ and $\alpha_{n,n}=1$ as $n \leq |\pebbles|$. 
		By definition of transition consistency,  we obtain $(\Phi_1(\pebbles), \Phi_2(\head, \pebbles)) \models \op$.
		
		Now, $\op(\pebbles,\head)=\pebbles'=\pebbles_1\dots \pebbles_{n-1}$. 
		$\Phi_1(\pebbles')$ assigns $\alpha_{i,j}=0$ for $i, j \geq n$ as  $n > |\pebbles'|$, 
		$\alpha_{i,j}=1$ for all $i, j \leq |\pebbles'|=n-1$ and $\pebbles_i=\pebbles_j$. 
		The definition of $\basic$ precisely assigns the same to $\alpha'=\op(\Phi_1(\pebbles), \Phi_2(\head, \pebbles))$. 
	\end{enumerate}
	
	Thus, $(\head, \pebbles) \models \varphi, \op$ $\Rightarrow$ $(\Phi_1(\pebbles), \Phi_2(\head, \pebbles)) \models \varphi, \op$. \cref{figcommu} shows the commutation 
	of applying $\op$ and $\Phi$.

	Conversely, assume $(\Phi_1(\pebbles), \Phi_2(\head, \pebbles)) \models \varphi, \op$. 
	Let $\alpha=\Phi_1(\pebbles), b=\Phi_2(\head, \pebbles)$. 
	We have to show that $\head, \pebbles \models \varphi, \op$. 
	First we consider $\varphi$. 
	\begin{enumerate}
		\item $\varphi=(\head=\pebbles_i)$. By definition, $\alpha, b \models \varphi$ if $b_i=1$. 
		By the definition of $\Phi_2(\head, \pebbles)$, $b_i$ is set to 1 when $\head=\pebbles_i$
		and $i \leq |\pebbles|$. These are precisely the conditions needed for 	$\head, \pebbles \models \varphi$.
		\item $\varphi=(\pebbles_i=\pebbles_j)$. 	By definition, $\alpha, b \models \varphi$ if $\alpha_{i,j}=1$. By construction of $\Phi_1(\pebbles)$, $\alpha_{i,j}$ is assigned 1 
		when $1 \leq i, j \leq |\pebbles|$ and $\pebbles_i=\pebbles_j$. These are precisely the conditions needed for $\head, \pebbles \models \varphi$.  
	\end{enumerate}
	
	Now, we consider $\op$. 
	\begin{enumerate}
		\item  $\op=\nop$. In this case, neither the pebble stack nor the head is affected. Hence, 
		if $\alpha, b \models \op$, we also have $\head, \pebbles  \models \op$.
		\item $\op=\push~\pebbles_n$. By definition, $\alpha,b \models \op$ if $\alpha_{n,n}=0, 
		\alpha_{n-1,n-1}=1$. By definition of $\Phi_1$, $\alpha_{n-1,n-1}$ is assigned 1 when $\pebbles_{n-1}=\pebbles_{n-1}$ and $n-1 \leq |\pebbles|$; also, $\alpha_{n,n}$ is assigned 0 when $n > |\pebbles|$ and $\neg(\pebbles_n=\pebbles_n)$. Together, 
		$n-1 \leq |\pebbles|$ and $n > |\pebbles|$ gives us $|\pebbles|=n-1$. 
		The definition 	 of $\Phi_1$ also ensures that $\pebbles_i=\pebbles_i$ for all $i \leq n-1$. 	 These are precisely the conditions needed for 
		$\head, \pebbles \models \push~\pebbles_n$.
		\item $\op=\pop~\pebbles_n$. By definition, $\alpha,b \models \op$ if $b_n=1$ and $\alpha_{n+1,n+1}=0$, $n < |\pebbles|$. By the definition of $\Phi_1(\pebbles)$ and  $\Phi_2(\head, \pebbles)$, $b_n$ is assigned 1 when 
		$n \leq |\pebbles|$ and $\head=\pebbles_n$. Likewise, $\alpha_{n+1,n+1}$ is assigned 0 
		when $n+1 > |\pebbles|$ or $\neg(\pebbles_{n+1}=\pebbles_{n+1})$. $n \leq |\pebbles|$ and 
		$n+1 > |\pebbles|$ together say that $|\pebbles|=n$. 
		By definition of $\head,\pebbles \models \op$, we see that 
		$\head,\pebbles \models \pop~\pebbles_n$ since we have $|\pebbles|=n$ and $\head=\pebbles_n$. 
	\end{enumerate}
	
	Thus, $(\Phi_1(\pebbles), \Phi_2(\head, \pebbles)) \models \varphi, \op$
	$\Rightarrow$ 
	$(\head, \pebbles) \models \varphi, \op$.  
\end{proof}

\begin{lemma}\label{lem:conf}
	Let $(\alpha, b)$ be a pair satisfying basic consistency.	Then 
	there exists $\head, \pebbles$  such that $(\alpha,b)=(\Phi_1(\pebbles), \Phi_2(\head, \pebbles))$. 
	
\end{lemma}
\begin{proof}
	Given $(\alpha, b)$, first we construct $|\pebbles|=\max\{i \mid \alpha_{i,i}=1\}$.  
	Recall that $\alpha$ is an equivalence relation on $\{1,2, \dots, |\pebbles|\}$. Let $[i]$, $1 \leq i \leq |\pebbles|$ denote the equivalence class of $i$ under $\alpha$. $[i]=\{j \mid \alpha_{i,j}=1\}$.

	For any word $\#u$ such that $|u|>|\pebbles|$, 
	define $f:\{1,2,\dots, |\pebbles|\} \rightarrow \{0,1,\dots, |u|\}$. $f$ is a function 
	which assigns members of $[i]$  to a unique position on $\#u$. For $i, j \in  \{1,2, \dots, |\pebbles|\}$, $i \sim j$ iff $f(i)=f(j)$. Define 
	$\pebbles_j=f(j)$, for $1 \leq j \leq |\pebbles|$. For all $j \in [i]$, this gives 
	$\pebbles_j=\pebbles_i$ since $f(j)=f(i)$.  With this definition 
	of $\pebbles$,  we show that $\Phi_1(\pebbles)=\alpha$. Let the $k \times k$ Boolean matrix defined by $\Phi_1$ be called $\beta$. We will show that $\alpha=\beta$.

	By definition,  
	$\Phi_1$ assigns $\beta_{i,j}=1$  iff $\pebbles_i=\pebbles_j$ for $1 \leq i, j \leq |\pebbles|$.  By the construction of $\pebbles$, we have $\pebbles_i=\pebbles_j$ iff $f(i)=f(j)$ iff $i \sim j$. Since $\alpha_{i,j}=1$ iff $i \sim j$, we obtain $\alpha=\beta$. Hence, $\alpha=\Phi_1(\pebbles)$. 
	
	Next, we define $\head$ using $\alpha,b$, and construct $b'=\Phi_2(\head, \pebbles)$, and 
	show that $b'=b$.

	First, let us consider the case when $b=\overline{0}$. In this case, 
	define $\head=v$ where $v \in \{0,1,\dots,|u|\}$ is such that $v \neq f(i)$ for any $i$.
	Otherwise, define $\head=f(i)$ if $b_i=1$. Notice that, by definition 
	of $f$, $f(i)=f(j)$ for all $j \in [i]$; this gives $b_j=b_i=1$ for all $j \in [i]$.  
	Finally, notice that for all $j \in [i]$,  $b_j=b_i$ since $(\alpha,b)$ satisfies basic consistency.

	Now by definition of $\Phi_2$, $b'_i=1$ if $i \leq |\pebbles|$ and $\head=\pebbles_i=f(i)$, and  $b'_i=0$ when $\head=v \notin f(\{1,2,\dots,|\pebbles|\}$ ($\head \neq \pebbles_i$ for all $i$).   We now show that $b'=b$. 
	By construction, $b'_i=1$ when $\head=f(i)$, and this is true when $b_i=1$.
	In this case we also have $b'_j=1$ for all $j \in [i]$ : then we know also that 
	$b_j=1$, since $b_i=b_j$ for all $j \in [i]$.  Finally, $b'_i=0$ when $\head \neq f(i)$ for any $i$. That is, $\head \neq \pebbles_i$ for any $1 \leq i \leq |\pebbles|$; in this case, 
	$b_i=0$ as well. Thus, $b'=b=\Phi_2(\head, \pebbles)$.  
\end{proof}

\begin{figure}
	\begin{center}
		\includegraphics[width=5cm]{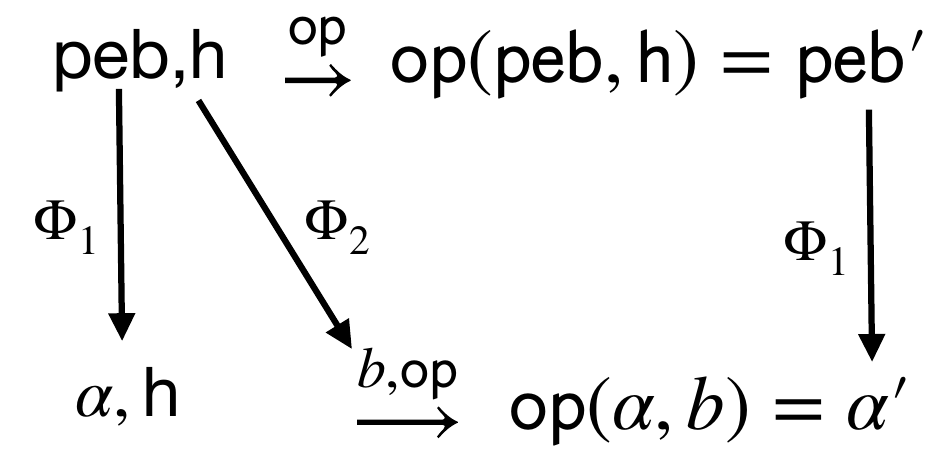}
	\end{center}
	\caption{Figure showing commutation of $\Phi$ and $\op$}
	\label{figcommu}
\end{figure}

\begin{lemma}\label{lem:run}
	Let $(q, \head, \pebbles) \xrightarrow{t=(q, a, \varphi, \op,q')} (q', \head', \pebbles')$ in  $\auto$ be a transition. Then  we have $((q,\Phi_1(\pebbles)), \head, \pebbles) \xrightarrow{t'=((q,\Phi_1(\pebbles)), a,\Phi_2(\head,\pebbles), \op, (q', \Phi_1(\pebbles')))} ((q',\Phi_1(\pebbles')), \head', \pebbles')$ in  $\basic$. 
	Conversely, for any transition $((q,\Phi_1(\pebbles)), \head, \pebbles) \xrightarrow{t'=((q,\Phi_1(\pebbles)), a,\Phi_2(\head,\pebbles), \op, (q', \Phi_1(\pebbles')))} ((q',\Phi_1(\pebbles')), \head', \pebbles')$ in  $\basic$, we can find a transition 
	$(q, \head, \pebbles) \xrightarrow{t=(q, a, \varphi, \op,q')} (q', \head', \pebbles')$ in  $\auto$. 
	
\end{lemma}
\begin{proof}
	
	Assume that $C_1=(q, \head, \pebbles) \xrightarrow{t=(q, a, \varphi, \op,q')} C_2=(q', \head', \pebbles')$ in  $\auto$.  To enable a corresponding transition in $\basic$, we go through the following steps.

	\begin{enumerate}
		\item \emph{Choose a $(\alpha,b)$}: First of all, by the construction of $\basic$, we know that there is a transition $t'=((q,\alpha), a,b, \op, (q',\alpha'))$ corresponding to $t$ 
		so that $(\alpha,b)$ respects basic consistency and transition consistency. 
		\cref{lem:basic,lem:trans} tell us that the pair $(\Phi_1(\pebbles), \Phi_2(\head,\pebbles))$ satisfies basic and transition consistency, given  $C_1$ and $t$.  
		Thus, this pair is a good candidate for $(\alpha, b)$ in $t'$.

		\item \emph{Show that the chosen $(\alpha,b)$ enables the transition.} We show that from configuration $C'_1=((q,\Phi_1(\pebbles), \head, \pebbles)$, $t'$ is  
		enabled in $\basic$.  
		\cref{lem:trans} tells us that $(\Phi_1(\pebbles), \Phi_2(\head, \pebbles)) \models \varphi, \op$. The resultant configuration has the new head as $\head'$ and the new pebble stack as $\pebbles'$. It remains to show that the resultant configuration in $\basic$ is $C'_2$; we need to fix $\alpha'$ in $t'$ and $C'_2$.

		\item \emph{Choose the target state.}
		\cref{lem:trans} also tells us the resultant $\alpha'$ obtained after executing $\varphi, \op$ in $t$.  $\alpha'=\op(\alpha,b)=\op((\Phi_1(\pebbles), \Phi_2(\head, \pebbles))=\Phi_1(\op(\pebbles,\head))=\Phi_1(\pebbles')$.

	\end{enumerate}

	Thus, we have obtained $\alpha=\Phi_1(\pebbles)$, $b=\Phi_2(\pebbles, \head), \alpha'=\Phi_1(\pebbles')$ so that $t'=((q,  \Phi_1(\pebbles)), (a, \Phi_2(\pebbles, \head)), \op, (q', \Phi_1(\pebbles'))$ is enabled from $C'_1$. The resultant configuration 
	is indeed $C'_2=((q', \Phi_1(\pebbles')), \pebbles', \head')$ in 
	$\basic$.

	Now for the reverse direction.
	Let 
	\[((q,\Phi_1(\pebbles)), \head, \pebbles) \xrightarrow{t'=((q,\Phi_1(\pebbles)), a,\Phi_2(\head,\pebbles), \op, (q', \Phi_1(\pebbles')))} ((q',\Phi_1(\pebbles')), \head', \pebbles')\]
	in  $\basic$. 
	Let $(\alpha, b)=(\Phi_1(\pebbles), \Phi_2(\pebbles, \head))$. 
	\begin{enumerate}
		\item By construction of $\basic$ and \cref{lem:basic} we know that 
		$(\alpha, b)$ satisfies basic consistency
		\item By construction of $\basic$ and  \cref{lem:trans} we know that $\alpha, b \models \op$
		\item By construction of $\basic$, there is a transition $t=(q,a,\varphi, \op, q')$ in $\auto$ such that $\alpha, b \models \varphi$. By \cref{lem:trans}, we know that 
		this implies $\head, \pebbles \models \varphi$. 
	\end{enumerate}
	Hence, indeed $t$ can be fired from $(q, \pebbles, \head)$ in $\auto$ obtaining 
	$(q', \pebbles', \head')$. 
	This establishes the reverse direction. 
\end{proof}

\begin{lemma}\label{lem:det}
	$\basic$ is deterministic if $\auto$ is deterministic.	

\end{lemma}
\begin{proof}
	Assume $\auto$ is deterministic. 
	
	Consider a transition $t'=((q,\alpha), a,b, \op, (q',\alpha'))$ in $\basic$ enabled from a configuration $C'=((q,\alpha),\pebbles, \head)$.   
	We claim there is a unique transition $t$ in $\auto$ corresponding to $t'$ enabled 
	from a configuration $C=(q,\pebbles, \head)$. 
	Assume on the contrary, that we have 
	$t_1=(q,a,\varphi_1,\op_1, q_1)$ and $t_2=(q,a,\varphi_2, \op_2,q_2)$ enabled from $C$.  By the construction  of $t'$ in $\basic$ from $t_1, t_2$, we know that $\op_1=\op_2=\op$ and 
	$q_1=q_2=q'$. Since $t_1, t_2$ are both enabled from $C$ and $\auto$ is deterministic, we have 
	$\varphi_1=\varphi_2$. Hence, $t_1=t_2$ is a unique transition in $\auto$ corresponding 
	to $t'$.

	Now we show that $\basic$ is deterministic. 	
	Let two transitions 
	$t'_i=((q, \alpha), a,b, \op_i, (q_i',\alpha_i'))$, $i=1,2$ in $\basic$ be enabled from 
	$C'=((q,\alpha),\pebbles, \head)$. 
	From the above, we know that there is a unique $t_i$ in $\auto$ corresponding to each $t'_i$.
	Let $t_i=(q,a,\varphi_i, \op_i, q'_i)$. 
	
	From \cref{lem:conf}, we know that there exists $\pebbles, \head$ such that $\alpha=\Phi_1(\pebbles), b=\Phi_2(\pebbles,\head)$. 
	By \cref{lem:trans} we know that $\alpha,b \models \varphi_i \Rightarrow (\head, \pebbles)\models \varphi_i$, and $\alpha,b \models \op_i \Rightarrow (\head, \pebbles)\models \op_i$. If $t_1,t_2$ are enabled from $C=(q, \pebbles, \head)$, then by determinism of $\auto$, 
	we have $t_1=t_2$. Hence, $\op_1=\op_2=\op$. Also, by \cref{lem:trans}, $\alpha'_i=\op_i(\alpha,b)=\op(\alpha,b)$ for $i=1,2$.  
	This implies $t'_1=t'_2$. Thus, $\basic$ is deterministic. 
\end{proof}

\begin{lemma}\label{lem:rev}
	$\basic$ is reverse deterministic if $\auto$ is reverse deterministic.  	

\end{lemma}
\begin{proof}
	Assume $\auto$ is reverse deterministic.  
	Let $t'_i=((q_i, \alpha_i), a,b, \op_i, (q', \alpha'))$ for $i=1,2$ be two transitions 
	in $\basic$. If $t'_1, t'_2$ are both simultaneously reverse enabled, we want to show that 
	$t'_1=t'_2$. 
	Assume that we have configurations $C'_1, C'_2, C'$ such that 
	$C'_1=((q_1, \alpha_1),  \pebbles^1, \head^1) \xrightarrow{t'_1} C'=((q', \alpha'), \pebbles', \head')$, and $C'_2=((q_2, \alpha_2),  \pebbles^2, \head^2) \xrightarrow{t'_2} C'=((q', \alpha'), \pebbles', \head')$. 
	
	\begin{enumerate}
		\item By \cref{lem:conf}, we know that $\alpha_1=\Phi_1(\pebbles^1), \alpha_2=\Phi_1(\pebbles^2)$. 
		Also, $\head^1=\head^2=\head'-q'=\head$
		\item By \cref{lem:run}, corresponding to $t'_1, t'_2$, for $i=1,2$ we 
		have the transitions 
		
		$(q_i, \pebbles^i, \head) \xrightarrow{t_i=(q_i, a, \varphi_i, \op_i, q')} (q', \pebbles', \head')$ in $\auto$. 
		Since $\auto$ is reverse deterministic, we know that 
		$t_1=t_2$. $t_1=t_2$ implies that $\op_1=\op_2$, $\varphi_1=\varphi_2$, $q_1=q_2$. 
		\item From \cref{sec:reversing}, we know that $\pebbles^i=\reverse{\op}(\pebbles', \head)$ for $i=1,2$. That is, $\pebbles^1=\pebbles^2$.
		\item Hence, $\alpha_1=\Phi_1(\pebbles^1)=\Phi_1(\pebbles^2)=\alpha_2$.
		\item Thus we have $t'_1=t'_2$.
	\end{enumerate}
	Hence, $\basic$ is reverse deterministic. 
\end{proof}

\begin{theorem}
	The basic $k$-pebble transducer $\basic$ constructed above is deterministic and reversible if  $\auto$ is. Further, $\sem{\basic}=\sem{\auto}$.
\end{theorem}
\begin{proof}
	\cref{lem:run,lem:det} show that every run in $\auto$ has a corresponding 
	unique run in $\basic$ and conversely.  By construction of $\basic$, the outputs of corresponding transitions 
	in $\basic$ are the same as those of $\auto$. 
\end{proof}

\end{document}